\documentclass[cmp]{svjour}
\pdfoutput=1
\usepackage{amssymb}
\usepackage{amsmath}
\usepackage{graphicx}
\usepackage{bm}
\usepackage{verbatim}

\usepackage{color}

\usepackage[normalem]{ulem}

\newcommand{\mc}[1]{\mathcal{#1}}
\newcommand{\veps}{\varepsilon}

\newcommand{\lra}{\longrightarrow}

\newcommand{\tr}{\mathrm{tr}}

\newcommand{\w}{\omega}
\newcommand{\N}{\mathbb{N}}

\newcommand{\R}{\mathbb{R}}
\newcommand{\C}{\mathbb{C}}
\newcommand{\I}{\mathbb{I}}
\DeclareMathOperator{\spann}{span}

\DeclareMathOperator{\graph}{graph}

\newcommand{\bra}[1]{\langle #1|}
\newcommand{\ket}[1]{|#1\rangle}
\newcommand{\braket}[2]{\langle #1|#2\rangle}
\newcommand{\figg}[2]{\ensuremath{\vcenter{\hbox{{\includegraphics[ scale=#1]{#2}}}}}}
\newcommand{\figgg}[2]{\ensuremath{\vcenter{\hbox{{\includegraphics[ scale=#1]{Dib/#2}}}}}}

\def\cont{\tfrac{c}{\ \ }}

\DeclareMathOperator{\trace}{tr} 
\DeclareMathOperator{\spanned}{span}
\DeclareMathOperator{\spec}{spec}

\DeclareMathOperator{\Real}{Re}

\begin{document}

\title{Frustration free gapless Hamiltonians for Matrix Product States}

\author{C. Fern\'andez-Gonz\'alez\inst{1,2}
\and N. Schuch\inst{3,4}
\and M. M. Wolf\inst{5}
\and J. I. Cirac\inst{6}
\and D.~P\'erez-Garc\'ia\inst{2,7}}

\institute{Departamento de F\'{\i}sica de los Materiales, Universidad
Nacional de Educaci\'on a Distancia (UNED), 28040 Madrid, Spain
\and
Departamento de An\'alisis Matem\'atico \& IMI, Universidad
Complutense de Madrid, 28040 Madrid, Spain
\and 
Institut f\"ur Quanteninformation, RWTH Aachen, 52056 Aachen,
Germany
\and
Institute for Quantum Information,
California Institute of Technology, MC 305-16, Pasadena CA 91125, U.S.A.
\and
Department of Mathematics, Technische Universit\"at M\"unchen,
85748 Garching, Germany
\and
Max-Planck-Institut f\"ur Quantenoptik, Hans-Kopfermann-Str. 1,
D-85748 Garching, Germany
\and
Instiuto de Ciencias Matemáticas, ICMAT (CSIC-UAM-UC3M-UCM), Campus de Cantoblanco, 28049 Madrid, Spain}

%%% REMOVE FOR CMP STYLE
\def\makeheadbox{}
\date{}
%%% REMOVE FOR CMP STYLE

\maketitle

\begin{abstract}
For every Matrix Product State (MPS) one can always construct a so-called
parent Hamiltonian. This is a local, frustration free, Hamiltonian which has
the MPS as ground state and is gapped. Whenever that
parent Hamiltonian has a degenerate ground state space (the so-called non-injective case),
we construct another 'uncle' Hamiltonian which is also local and frustration free, has the same ground state space, but is
gapless, and its spectrum is $\R^+$. The construction is obtained by linearly
perturbing the matrices building up the state in a random direction, and then taking the limit
where the perturbation goes to zero. For MPS where the parent
Hamiltonian has a unique ground state (the so-called injective case) we
also build such uncle Hamiltonian with the same properties in the 
thermodynamic limit.
\end{abstract}

\section{Introduction}

One of the aims of condensed matter physics is to understand the low temperature behavior of locally interacting quantum systems. For that, one has to {\it solve} the Hamiltonian which is postulated to capture the physics of the system under study. This means obtaining a good description of its ground state which, in turn, allows to make predictions on the observable properties of the system. However, solving a Hamiltonian is, except for very few models, completely out of reach analytically and in general also very hard even numerically. This has motivated a reverse engineering approach, where one designs particular wavefunctions which try to capture some of the physical properties of the system, such as symmetries or frustration, and tries to find appropriate Hamiltonians for them for which one can give a rigorous mathematical treatment.

A first paradigmatic example of this approach is the AKLT state \cite{aklt}, introduced by Affleck, Kennedy, Lieb and Tasaki in 1988 as a way to understand the one dimensional antiferromagnetic spin-1 Heisenberg model. In this case, an associated nearest neighbor Hamiltonian was already introduced in \cite{aklt} and its spectral gap rigorously proven --as opposed to the gap of the rest of Haldane's phase where only numerical evicende is known. The AKLT construction was generalized later in 1992 by Fannes, Nachtergaele and Werner \cite{fannes:FCS}. They constructed the family of Finitely Correlated States, nowadays known as Matrix Product States (MPS), and found a gapped local Hamiltonian (called the {\it parent} Hamiltonian) for each one of them\footnote{The very same year, White introduced the Density Matrix Renormalization Group (DMRG) algorithm \cite{White}, which turned out to be extremely successful as a way to derive ground states from 1D gapped local Hamiltonians. It was only realized later that DMRG was a way to find the closest MPS to the target state, indicating that the family of MPS was indeed large enough to describe the low temperature physics of all gapped 1D systems. This was finally proven by Hastings in \cite{Hastings} (see also \cite{Verstraete-faithfully}) and later, with exponentially better parameters, by Arad, Kitaev, Landau and Vazirani in \cite{arad}.}. The other paradigmatic example is the Resonating Valence Bond (RVB) state, postulated in the seminal work of Anderson \cite{RVBS-Anderson} in 1987 as a way of explaining high-Tc superconductivity. This, together with the fact that such state on a geometrically frustrated lattice would be a topological quantum spin liquid, has attracted the attention of theoretical and experimental physicists for many years \cite{balents}. In this case, only very recently \cite{RVBS-PEPS,Seidel:RVBS}, and using the 2D version of the MPS parent Hamiltonian described in 
\cite{perez-garcia:pepsasgroundstates,schuch:peps-sym}, an associated local Hamiltonian has been found and the properties of that Hamiltonian are still under study.

The interest in this reverse engineering approach (obtaining Hamiltonians for a target quantum state) has boosted in recent years thanks to the role it plays in quantum information theory and in the problem of classifying quantum phases of matter. In the first case, one identifies valuable resource states: GHZ states \cite{GHZ}, cluster states \cite{1Dcluster}, topological codes \cite{kitaev:toriccode}, ..., and tries to find ways to engineer them and stabilize them for long times. In essentially all these cases, the desired states are MPS or their 2D generalization (PEPS), so the parent Hamiltonian construction  \cite{fannes:FCS} gives us already a solution. In the second, to show that two quantum states $\ket{\psi_0}$, $\ket{\psi_1}$ are in the same phase, one must engineer a smooth path of gapped Hamiltonians $H_\lambda$, $\lambda\in [0,1]$, so that $\ket{\psi_0}$ is a ground state of $H_0$ and $\ket{\psi_1}$ is a ground state of $H_1$. Here, the parent Hamiltonian construction was the the way to engineer the paths of Hamiltonians which led to the final classification of 1D phases presented in \cite{chen:classification,kitaev:classification,pollmann,schuch:phaseclassification}.

This parent Hamiltonian, despite being so useful as the main general reverse engineering construction, has a clear weakness, as pointed out by Chen et al in \cite{chen:parentH-robustness}. It is in general not robust against perturbations in the matrices which define the MPS. This issue is crucial for the applications described above. In this work, we analyze this problem and solve it in full detail by (1) identifying the Hamiltonian (we call it {\it uncle}) which is being perturbed when perturbing the MPS, (2) analyzing the properties of the uncle Hamiltonian and (3) characterizing for which perturbations this uncle Hamiltonian coincides with the parent.

In particular, we show that the parent Hamiltonian is robust if and only if the MPS is injective \cite{perez-garcia:mps-reps}. When this is not the case, which corresponds to systems with discrete symmetry breaking \cite{perez-garcia:mps-reps} (block-injective MPS), the uncle Hamiltonian has the following striking properties: (1) It has the same -finite
dimensional- ground space as the parent Hamiltonian, (2) it is also
frustration free, but (3) it is gapless and its spectrum
equals $[0,+\infty)$ in the thermodynamic limit, whereas parent Hamiltonians always show a spectral gap. These properties hold for
\emph{any} block-injective MPS, and for \emph{generic}
perturbations. Unlike the parent Hamiltonian, the uncle Hamiltonian
may change continuously under perturbations of the MPS even in the
presence of a discrete symmetry.

As a by-product we obtain new examples which shed new light on the properties inherent to critical systems and phase transitions, in the line initiated in \cite{wolf:phasetransitions-mps}. Notice that all our examples are generic, as gapless as they can be since the spectrum is the whole positive real line, but they are frustration free and there is no power-law decay of correlations within the ground space. Indeed, modifying our construction, for any injective MPS (which always has exponentially decaying correlations), one can obtain simple frustration free models with the same spectral properties but with the given MPS as unique ground state in the thermodynamic limit -- however, the finite dimensional ground spaces do not coincide in this injective case.  Finally, our analysis could help in the design of MPS-based algorithms to simulate systems with discrete symmetry breaking in 1D or, more generally, systems with topological order in 2D.

The paper is organized as follows:  In Section \ref{sec:defs} we start with the basic definitions. Section \ref{sec:example} works out the illustrative example of the GHZ-state which helps in understanding the main theorems presented in Section \ref{sec:main}. Section \ref{sec:injectiveMPS} finally deals with the case of injective MPS. For the sake of clarity, we have moved some of the technical proofs to the appendices, where one can also find how to treat formally the thermodynamic limit using the GNS construction and some basic facts we need about  spectral theory of unbounded operators on a Hilbert space. 

\section{\label{sec:defs}
Definitions}

In this section, we will provide the necessary definitions: We start by introducing Matrix Product States, their normal form,
and the injectivity property. Next, we show how parent Hamiltonians for
MPS are constructed.  We find that parent Hamiltonians change
discontinuously under certain perturbations, which motivates the
introduction of uncle Hamiltonians.

\subsection{Matrix Product States}

\begin{definition}[Matrix Product States]\label{MPS}
A state $\ket{\psi}\in (\C ^d)^{\otimes L}$ is called a Matrix Product
State (MPS) if it can be written as 
\begin{equation}
\label{eq:mps-def}
\ket{M(A)} = \sum _{i_1,...,i_L} 
    \tr[A_{i_1} \cdots A_{i_L}] \ket{ i_1,\dots,i_L},
\end{equation}
where the $A_i$, $i=1,\dots,d$, are $D\times D$ matrices ($D$ is called
the bond dimension).
\end{definition}

\noindent The matrices $\{A_i\}_i$ can also be thought of as a tensor $A$
with three indices $(A_i)_{\alpha \beta}$, where the matrix indices
$\alpha$ and $\beta$ are referred to as ``virtual indices'', and the index
$i$ as ``physical index''. Note that this definition of MPS is restricted
to translationally invariant states with periodic boundary conditions,
which we will be concerned with in this paper.

Throughout this paper, we will use a graphical notation for tensor
contractions such as in Eq.~(\ref{eq:mps-def}), cf.~\cite{schuch:peps-sym}:
A tensor with $k$ indices will be denoted as a box with $k$ legs; e.g.,
the tensor $(A_i)_{\alpha\beta}$ is denoted as
\[
\figgg{0.18}{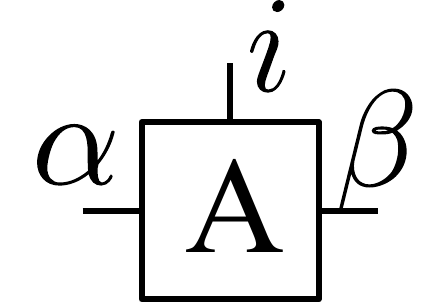}
\ \raisebox{-0.5em}{$\equiv$}
\raisebox{-0.2em}{\figgg{0.18}{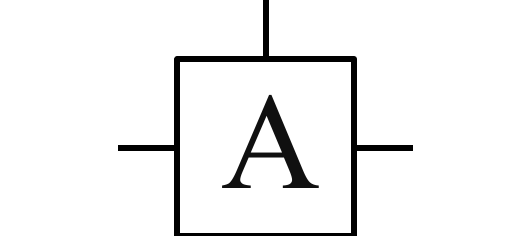}}
\]
with the upper leg by convention denoting the physical index. We will
generally omit the labels.  The contraction of two tensors, i.e., summing
over a joint index, is denoted by concatenating the corresponding indices;
e.g., $\sum_\beta
(A_i)_{\alpha\beta}(B_j)_{\beta\gamma}$ is written as
\[
\figgg{0.16}{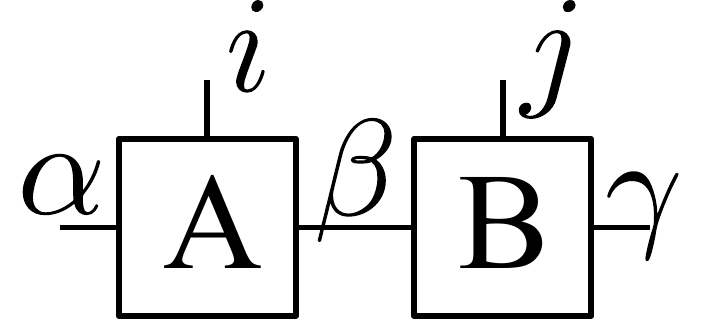}
\ \raisebox{-0.5em}{$\equiv$}\ \ 
\raisebox{-0.2em}{\figgg{0.16}{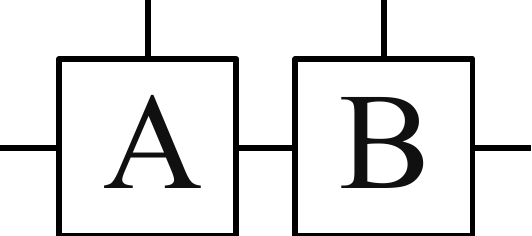}}
\ \raisebox{-0.7em}{\ .}\ 
\]
We will also use the more compact notation 
\[A\cont B\]
 instead.  In this
graphical language, an MPS, Eq.~(\ref{eq:mps-def}), can be expressed as 
\[
\figgg{0.18}{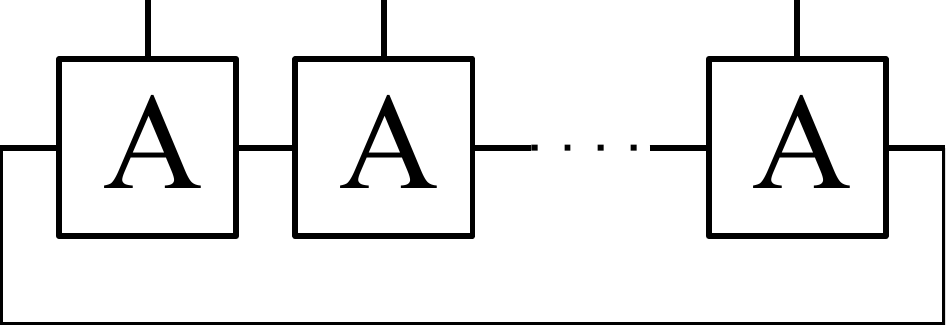}\quad .
\]
Given a tensor $T\equiv (T_i)_{\alpha\beta}$, we will denote by
$\bra{\alpha}T\ket{\beta}$ the vector $\sum_i
\bra{\alpha}T_i\ket{\beta}\ket{i}=\raisebox{0.3em}{\figgg{0.16}{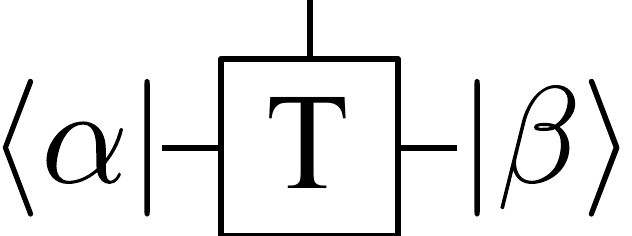}}$\,\
in the physical space.  More generally, matrix operations are generally
meant to act on the virtual degrees of freedom; e.g., $C=A\oplus B$ denotes
the tensor with components $C_i=A_i\oplus B_i$.

In order to compute expectation values etc., we need to contract tensor
networks with their adjoint.  Tensors with legs pointing down are always
complex conjugated,\footnote{Throughout, $\bar\cdot$ will denote the
complex conjugate and $\cdot^*$ the adjoint operator.}
\[
\raisebox{-0.35em}{\figgg{0.16}{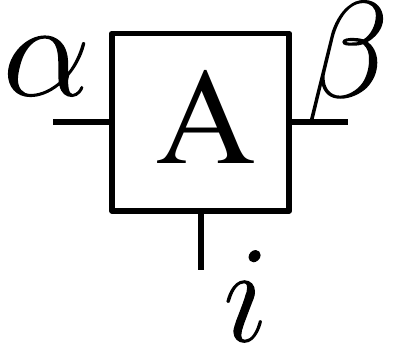}}
\ \equiv (\bar A_i)_{\alpha\beta}\ \equiv (A_i^*)_{\beta\alpha}\ .
\]
An object of particular interest is the ``transfer operator''
\[
E_A^B = \figgg{0.18}{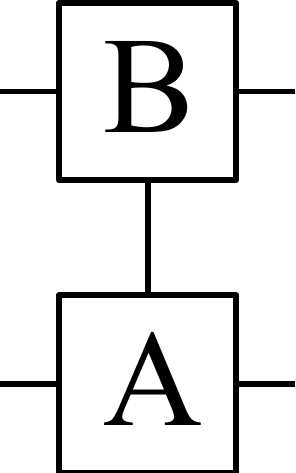}\quad ,
\]
which we will occasionally interpret as a map from the right to the left
indices, $E_A^B:X\mapsto \sum_i A_i X  B_i^*$.

\subsection{Normal form}

MPS are stable under blocking: If we e.g.\ block pairs of physical sites,
$j_k\equiv(i_{2k-1},i_{2k})$, the state can again be expressed as an MPS
in the blocked indices, with tensors $A_{j_k}\equiv
A_{i_{2k-1}}A_{i_{2k}}$. By blocking a finite number of
sites~\cite{sanz:wielandt} and appropriate gauge transformations, any MPS
can be brought into a standard
form~\cite{fannes:FCS,perez-garcia:mps-reps}.  
\begin{theorem}[Standard form for
MPS~\cite{fannes:FCS,perez-garcia:mps-reps}, injectivity] 
\label{thm:std-form-inj}
After blocking, any MPS can be written in a standard form where the
matrices $A_i$ have the following properties: 
\begin{enumerate}
\item The $A_i$ are block-diagonal:
    $A_i=\oplus_{j=1}^\mathcal{D} A_i^j\otimes \Gamma_j$, 
    where $A_i^j\in\mathcal{M}_{l_j}$ (the space of $l_j\times l_j$
    matrices) and the $\Gamma_j$ are positive diagonal
    matrices.
\item The $A_i$ span the space of block-diagonal matrices:
    ${\rm span}_i A_i=\bigoplus_{j=1}^\mathcal{D}
    \mathcal{M}_{l_j}\otimes \Gamma_j$.
\item For all $j$, and for $A^j$ denoting the MPS tensor defined by the submatrices $A^j_i$,  the map $\mathcal E_j:=E_{A^j}^{A^j}$ has spectral radius one, with $1$ as the unique eigenvalue of modulus $1$, and with eigenvectors 
    $\mathcal E_j(\I)=\I$ and
    $\mathcal E^*_j(\Lambda_A^j)=\Lambda^j_A$, where
    $\Lambda^j_A>0$, $\trace(\Lambda_A^j)=1$.
\end{enumerate}
Property 2 with every $\Gamma_j=1$ is called \emph{block-injectivity}; in particular, if 
$\mathcal{D}=1$, and $\Gamma_1=1$, $A$ is called \emph{injective}. 
\end{theorem}
For the rest of the paper, we will always consider MPS in this standard
form.

\begin{definition}[Span of a tensor, Projector corresponding to a tensor]
Given a tensor $(T_i)_{\alpha\beta}$ with two virtual indices $\alpha$,
$\beta$, and one physical index $i$ (which can be a blocked index), we
define the \emph{span of $T$} as
\[
\spann\{T\}:=\spann\big\{\sum_i\tr[T_iX]\,\ket{i}\,
    \big\vert\, X \in \mc M_D\big\}\ .
\]
Also, we define the \emph{projector corresponding to $T$, $\Pi[T]$}, as the orthogonal projector onto $\spann\{T\}^\perp$.
\end{definition}
In particular, $T$ can arise from blocking two or more tensors $A$ from an
MPS, $T=A\cont A\cdots$. E.g.,
\begin{equation}
\label{eq:span-2}
\spann\{A\cont A\}
    \equiv \spann\big\{\sum_{i,j}
	\tr[A_{i}A_{j}X]\ket{i,j}\big\vert X\in \mc M_D\big\}\ .
\end{equation}

Note that injectivity of a tensor $T$ implies $\dim (\spanned\{T\})=D^2$.

\begin{lemma}[Continuity of projector of a tensor]
\label{lemma:span-continuity}
Let $T(\veps)$ be a family of tensors.  If $T(0)$ is injective and
$T(\veps)$ is continuous around $0$, then $\Pi[T(\veps)]$ is continuous at
$0$.

More generally, if $T(0)$ is block-injective and $T(\veps)$ is continuous
and block-diagonal in the same basis around $0$, then $\Pi[T(\veps)]$
is continuous around $0$.
\end{lemma}

\begin{proof}
The proof follows directly from the fact that a basis $X_k$ of the space
of \mbox{(block-)}dia\-gonal matrices yields a continuously changing basis of
$\spann\{T(\veps)\}$ by virtue of $X_k\mapsto \sum_i\tr[T_i(\veps)X_k]\ket i=\ket{v_k(\veps)}$.

From these bases, the Gram-Schmidt orthogonalization process leads also continuously to orthonormal changing bases of $\spann\{T(\veps)\}$, say $\{\ket{e_k(\veps)}\}$: Beginning from $\ket{v_1(\veps)}\ne 0$, its normalization is continuous. Then, orthonormalization of $\ket{v_2(\veps)}$ with respect to $\ket{e_1(\veps)}$ only involves projecting onto the complement of $\spann\{\ket{e_1(\veps)}\}$ and normalization, and both operations are continuous as long as we keep in every moment away from the null vector (this is where the importance of injectivity or block-injectivity stems, in the need of the dimension of $\spann\{T(0)\}$ to be kept locally). One can also consider the process as projecting $\ket{v_2(\veps)}$ onto $\spann\{\ket{e_1(\veps)}\}$, and then substracting the result to $\ket{v_2(\veps)}$, which yields a non-zero vector that must be normalized. With the rest of the vectors the process is similar with finitely many steps, which makes the orthonormalization process continuous if the injectivity conditions are held and every step keeps away enough from 0. Note that injectivity of $T(0)$ and continuity of $T$ at 0 is enough to ensure injectivity of $T(\varepsilon)$ for small $\varepsilon$.

The final projection $\Pi[T(\veps)]$ is also continuous in $\veps$, since projecting any vector onto every linear space $\spann\{\ket{e_k(\veps)}\}$ and substracting the result to the vector are continuous in $\veps$.

In the case $T(0)$ is block-injective, the fact that $T(\varepsilon)$ is also block-diagonal makes the rank of $\{\ket{v_k(\varepsilon)}\}_k$ constant around 0. This allows to guarantee the same result about continuity.
\hspace*{\fill}\qed

\end{proof}

\subsection{The parent Hamiltonian}
Let us now turn towards Hamiltonians for MPS.  We will restrict to
translationally invariant Hamiltonians, and denote the local terms by
lowercase letters, e.g., $h$. When necessary, subscripts indicate the
sites $h$ acts on, e.g., $h_{i,i+1}$.  We will identify the local operator
$h$ and the global operator $h\otimes\I$.  The total Hamiltonian will be
denoted by the corresponding uppercase letter, e.g., $H=\sum
h:=\sum_{i=1}^N h_{i,i+1}$. Generally, indices will wrap around the ends
of the chain (e.g., here $N+1\equiv 1$).

Every MPS $\ket{M(A)}$ induces Hamiltonians to which it is an exact ground
state: The reduced state on, say, two adjacent sites, $\rho_2$, is
supported on $\spann\{A\cont A\}$, and thus for $h:=\Pi[A\cont
A]$, $h\ket{M(A)}=0$.
If $D^2<d^2$ -- this can be achieved by blocking, and is the case for the
standard form of Theorem~\ref{thm:std-form-inj} -- $\Pi[A\cont A]$ is
non-trivial, and we obtain a non-trivial Hamiltonian with two-body terms
$h$ which has $\ket{M(A)}$ as its ground state.

If we use $\ket{M(A)}$ in the standard form of
Thm.~\ref{thm:std-form-inj}, this Hamiltonian is particularly
well-behaved:
\begin{definition}[Parent Hamiltonian]
\label{def:parentham}
Let $\ket{M(A)}$ be a \mbox{(block-)}injective MPS, i.e., satisfying
condition 2 of Theorem~\ref{thm:std-form-inj}, and let $h=\Pi[A\cont A]$.
Then, the Hamiltonian $H=\sum h$ is called the \emph{parent
Hamiltonian}.
\end{definition}

\begin{theorem}[Ground state space of parent Hamiltonian~\cite{nachtergaele:spectralgap,perez-garcia:mps-reps,schuch:peps-sym}]
\label{thm:gs-subspace-parent}
The parent Hamiltonian has a $\mc D$-fold degenerate ground state space
spanned by $\ket{M(A^j)}$; in particular, $\ket{M(A)}$ is one of its
ground states. Also, the parent Hamiltonian is gapped in the thermodynamic
limit. 
\end{theorem}
\begin{remark}\label{remark:parent-ham-weaker-cond}
Note that in order to ensure the correct ground state subspace in
Theorem~\ref{thm:gs-subspace-parent}, a weaker condition than the
injectivity of each tensor used in Definition~\ref{def:parentham} is
enough: It is sufficient to take projectors onto the orthogonal
complement of the span of $k+1$ sites, where $k$ is chosen such that
blocking $k$ sites makes the tensor injective~\cite{sanz:wielandt}.
\end{remark}

\subsection{Uncle Hamiltonians}

The parent Hamiltonian construction can be interpreted as a map from the
set of MPS to the set of Hamiltonians, $\hat H:A\mapsto \hat H(A)$, which
associates to any MPS $\ket{M(A)}$ its parent Hamiltonian $\hat H(A)$.
While this map is well-behaved in terms of the properties of $\hat
H(A)$, we are also interested in its \emph{continuity}:  If we change $A$
smoothly, $A\rightarrow A+\veps P$, does $\hat H(A)$ change smoothly as
well?  If it were so, this would allow us to study perturbations of the
system by looking at perturbations of the MPS tensor $A$.  For injective
MPS in their standard form, Lemma~\ref{lemma:span-continuity} tells us
that this is indeed the case.  On the other hand, if $A$ is
block-injective, Lemma~\ref{lemma:span-continuity} requires $P$ to be
block-diagonal as well, and it is indeed easy to see that \emph{random}
perturbations $P$ will lead to a discontinuous change in $\hat H(A+\veps
P)$, as $\mathrm{rank}(\spann\{A\cont A\})$ is different for $\veps\ne0$.

This discontinuity motivates the introduction of uncle Hamiltonians, which
are robust under specific perturbations of the MPS tensor.

\begin{definition}[Uncle Hamiltonians]
Let $\ket{M(A)}$ be a MPS in standard form, and let $P\in \mc M_D$. Then,
the \emph{uncle Hamiltonian induced by $P$} is the Hamiltonian
\begin{equation}
\label{eq:parent-def}
H_P' := \lim_{\varepsilon\rightarrow0} \hat H(A+\varepsilon P)\ ,
\end{equation}
whenever this limit exists.
This is, the uncle Hamiltonian induced by a given perturbation is the limit of the parent Hamiltonian of
the perturbed MPS for the perturbation going to zero.
\end{definition}
For $A$ an injective tensor, the limit exists and is equal to
the parent Hamiltonian for every tensor perturbation, following Lemma~\ref{lemma:span-continuity}:
Parent and uncle Hamiltonian coincide for injective MPS. Thus, we focus
our attention on the uncle Hamiltonians of block-injective MPS; most of
this paper is devoted to studying their properties. 

In the following, $H=\sum h$ will denote the parent Hamiltonian, while
$H'_P=\sum h'_P$ denotes the uncle Hamiltonian, where we will occasionally
omit the subscript $P$. Recall that this limit Hamiltonian depends on the perturbation, therefore a uniform limit for every perturbation as usually understood does not exist in general.

\section{\label{sec:example}
Example: The GHZ state, the Ising model, and the XY model}

We will start our discussion of the properties of the uncle with the GHZ
state, which has the Ising model as its
parent Hamiltonian. The unnormalized GHZ state can be expressed as an MPS
\begin{equation*}
\ket{\mathrm{GHZ}}
    = \ket{00\cdots0}+\ket{11\cdots1},
    =\sum _{i_1,\dots,i_L} 
	\tr[A_{i_1} \dots A_{i_L}] \ket{ i_1,\dots,i_L}
\end{equation*}
with  $i_j \in \{0,1\}$, where 
$A_0= \left( \begin{smallmatrix}
1 & 0 \\ 0 & 0 \end{smallmatrix} \right)$ and
$A_1=\left( \begin{smallmatrix}
0 & 0 \\ 0 & 1 \end{smallmatrix} \right)$. Following
Definition~\ref{def:parentham},
the parent Hamiltonian of the Ising model can be constructed from the
span of two sites,
\[
\spann\{A\cont A\}=\spann\{\ket{00},\ket{11}\}\ ,
\]
which is indeed (up to an additive constant) the well-known Ising
Hamiltonian $\tfrac12\I-\big[\ket{00}\bra{00}+\ket{11}\bra{11}\big]$.

\subsection{Uncle Hamiltonian for the GHZ state}

Let us now construct the uncle Hamiltonian for the GHZ state. According to
the definition, we first need to fix a perturbation $P$ of the MPS tensor
$A$,
\begin{equation}
    \label{eq:P-GHZ}
P_0=\left( \begin{matrix}
a_0 & b_0 \\ c_0 & d_0 \end{matrix} \right)\mbox{\quad and \quad}
P_1=\left( \begin{matrix}
a_1 & b_1 \\ c_1 & d_1 \end{matrix} \right)\ .
\end{equation}
Next, we need to consider the MPS $\ket{M(C_{\varepsilon,P})}$ with
$C_{\varepsilon,P}=A+\varepsilon P$, and construct its parent Hamiltonian.
For a generic $P$, we need to block two sites to reach injectivity.
Thus, following Remark~\ref{remark:parent-ham-weaker-cond}, we need to
construct the terms of the parent Hamiltonian as the projector onto the
complement of the span on three sites,
\[
\mc S_3(\varepsilon)=\spann\{C_{\varepsilon,P}\cont C_{\varepsilon,P}\cont 
C_{\varepsilon,P}\}\ .
\]
$S_3(\varepsilon)$ is spanned by the four vectors $v_{\alpha\beta}=
\bra\alpha C_{\varepsilon,P}\cont C_{\varepsilon,P}\cont 
C_{\varepsilon,P}\ket\beta$, which are straightforwardly found to be
\begin{align*}
v_{00}&=\ket{000}+O(\varepsilon)\ ,\\
v_{01}&= \varepsilon\big[b_0 \ket{000} +(b_0+b_1)  \ket{001} + 
    (b_0+b_1) \ket{011}+ b_1\ket{111}\big]+O(\varepsilon^2)\ ,\\
v_{10}&= \varepsilon\big[c_0 \ket{000} +(c_0+c_1)  \ket{100} + 
    (c_0+c_1) \ket{110}+ c_1\ket{111}\big]+O(\varepsilon^2)\ ;\\
v_{11}&=  \ket{111}+O(\varepsilon)\ .
\end{align*}
If $b_0+b_1\ne0$ and $c_0+c_1\ne0$ (which happens almost
surely\footnote{Note that whenever we say something happens almost surely
means not only that it happens with probability 1 but also that the set of
perturbations which may not satisfy the statement form a closed algebraic
variety of dimension strictly lower than the set of all possible
perturbations.}), this can be transformed into an alternative set
spanning $S_3(\veps)$,
\begin{equation}
\label{eq:ghz-parent-set}
\big\{\ket{000} + O(\varepsilon),\ket{001} + \ket{011}+O(\varepsilon),
\ket{100} +  \ket{110} + O(\varepsilon),\ket{111}+O(\varepsilon)\big\}\ .
\end{equation}
The parent Hamiltonian of the perturbed MPS $\ket{M(C_{\varepsilon,P})}$ is
thus 
$H_{P,\varepsilon}=\sum h_{P,\varepsilon}$, with each $h_{P,\veps}$ acting
locally on three consecutive sites, and projecting onto $S_3(\veps)^\perp$.

In order to obtain the uncle Hamiltonian we finally need to take the limit
$\veps\rightarrow0$. Then, the four states in
Eq.~(\ref{eq:ghz-parent-set}) become orthogonal, and the family
$h_{P,\varepsilon}$ converges to the projection onto the orthogonal
complement of 
\[
\spanned\{\ket{000},\ket{0{+}1},\ket{1{+}0},\ket{111}\}.
\]
Here, $\ket{0\!+\!1}\equiv\ket0\ket+\ket1$, with
$\ket+=(\ket0+\ket1)/\sqrt2$.  Thus, the uncle Hamiltonian has local terms
\begin{equation} \label{eq:ghz-uncle}
h'_P=\I-\big[\ket{000}\bra{000}+\ket{111}\bra{111}+\ket{0{+}1}\bra{0{+}1}
    +\ket{1{+}0}\bra{1{+}0}\big]
\end{equation}
Note that this limit does not depend on the perturbation $P$ (as long as
$b_0+b_1\neq 0$ and $c_0+c_1\neq 0$), and will be called $h'$ or
$h'_{i-1,i,i+1}$ in the following.  The parent Hamiltonian $W'$ is
obtained as the sum $H'=\sum h'\equiv\sum_i h'_{i-1,i,i+1}$.

\subsection{Ground space of the uncle Hamiltonian}

What is the ground space of the uncle Hamiltonian? Since
\[
\ker(h')=\spann \{\ket{000},\ket{111},\ket{0+1},\ket{1+0}\}
\supset
\spann \{\ket{000},\ket{111}\}=\ker h\ ,
\]
the ground states $\ket{0\dots0}$ and $\ket{1\dots1}$ of the GHZ parent
Hamiltonian are also ground states of the uncle; in particular, the uncle
is frustration free. However, $h'$ allows for additional ground states.
Indeed, if we consider the ground state space of $h'$ acting on $m$
consecutive sites (with open boundaries), the ground space is
\begin{equation}
\spanned \big\{ \ket{0\dots0}, \ket{1\dots 1}, \sum
\ket{0\cdots 01\cdots 1}, \sum\ket{1\cdots 10\cdots
0}\big\}\subseteq  (\C^2)^{\otimes m} ,  
\end{equation}
where the sums run over all positions of the ``boundary wall'' $01$ and
$10$, respectively. Yet, when closing the boundaries,  the additional
states $\sum\ket{0\cdots 01\cdots 1}$ and $\sum\ket{1\cdots 10\cdots 0}$
stop being in the intersection of the kernels, and the ground space of the
uncle Hamiltonian coincides with the ground space of the parent
Hamiltonian. Intuitively, with periodic boundary conditions the boundary
walls need to come in pairs, and it is impossible to give both of them
momentum zero as they meet which is not in the ground space of $h'$.

\subsection{Spectrum of the uncle Hamiltonian}
Let us now show that the uncle Hamiltonian of the GHZ state is gapless.
To this end, we consider the unnormalized states
\begin{equation}\label{lowenergystatesGHZfinitechain}
\ket{\phi_{N}}=\sum_{\substack{-N \le i \le -1\\1 \le j \le N}}  
\ket{\phi_{i,j}}
\end{equation}
with 
\[
\ket{\phi_{i,j}}=\ket{0^{-N-1}0\cdots 0^i1\cdots 10^j\cdots 0^{N+1}}
\]
on a chain of length $2N+3$, where the superscripts indicate the position
of the corresponding qubit.  (This notation will be used throughout this
work.) These states are orthogonal to the ground space,  and
$\braket{\phi_N}{\phi_N}=N^2$.  Further, $\ket{\phi_N}$ is \emph{almost} a
zero momentum state of the two boundaries: It is a ground state of all
terms in $H'$ except $h_{-1,0,1}$, $h_{N,N+1,-N-1}$, and
$h_{N+1,-N-1,-N}$, and by counting the violating configurations, we find
that $\bra{\phi_N}H' \ket{\phi_N}=O(N)$.  Hence, for the energies of these
states we have
\begin{equation}
\frac{\bra{\phi_N}H' \ket{\phi_N}}{\braket{\phi_N}{\phi_N}}=O(1/N).
\end{equation}
This implies that (on a chain of length $2N+3$)
$H'$ has at least one eigenvalue $\lambda_N\le
\frac{\bra{\varphi_N}H' \ket{\varphi_N}}{\braket{\varphi_N}{\varphi_N}}=
O(1/N)$, i.e., the  family  of uncle Hamiltonians is gapless. 

Does $H'$ have a continuous spectrum? One idea to prove so would be to to
construct momentum eigenstates with an energy scaling like
$\Theta(k^2/N^2)$ ($\Theta$ denotes the exact scaling rather than an upper
bound).  To this end, we give the sum in
Eq.~(\ref{lowenergystatesGHZfinitechain}) a momentum,
\begin{equation}\label{eq:GHZ-momentumstates}
\ket{\phi_{N,k}}=\sum_{\substack{-N< m<-1\\1< n < N}}  
e^{2\pi i m k/N}\ket{\phi_{m,n}}\ .
\end{equation}
It is straightforward to see that in addition to the $O(N)$ contribution from
before, the $N-1$ terms $h{-N-1,-N,-N+1},\dots,h_{-2,-1,0}$ all give a
contribution $\Theta(k^2N)$, resulting in an energy
$\frac{\bra{\varphi_N}H' \ket{\varphi_N}}{\braket{\varphi_N}{\varphi_N}}=
\Theta(k^2/N^2)$.  

Unfortunately, the existence of states with energy $\Theta(k^2/N^2)$ does
not allow to conclude that the spectrum of $H'$ is dense: The existence
of a state with energy $E$ only implies the existence of eigenvalues
$\lambda_1\le E$ and $\lambda_2\ge E$, but tells us nothing about their
exact value. (The reason this worked for the gapless excitations was that
$H'\ge0$, and that $\ket{\phi_N}$ was orthogonal to the ground space.)
It is, however, indeed possible to show that $H'$ has a continuous
spectrum.  For the GHZ example, this can be done by mapping it to the XY
model as discussed in Sec.~\ref{sec:GHZ-and-XY}, and in Sec.~\ref{sec:uncle-spectrum}, we give a
proof that uncle Hamiltonians of arbitrary MPS have a continuous spectrum
which works directly in the thermodynamic limit.

\subsection{Gapless uncles for unique ground
states\label{sec:uncle-for-000}}

Can we obtain uncle Hamiltonians with similar properties in the case of
MPS which are unique ground states? Lemma~\ref{lemma:span-continuity}
tells us that this cannot happen as long as the MPS tensors are injective,
which is always the case as long as such an MPS is in its standard form:
In that case, the uncle Hamiltonian is equal to the parent Hamiltonian.
However, as we will demonstrate in the following, interesting uncle
Hamiltonians can be obtained by choosing a different MPS representation.

Consider a qubit chain $(\mathbb C^2)^{\otimes N}$, and a state
$\ket{M(A)}=\ket{0\dots0}$.  Clearly, this is a unique ground state of a
gapped local Hamiltonian, with standard MPS representation $A_0=(1)$,
$A_1=(0)$.  However, we can write the same state with bond dimension $2$
and 
\[
A_0=\left(\begin{matrix} 1 & 0 \\ 0 & 1\end{matrix}\right)\ 
\mbox{\ and\ } 
A_1=\left(\begin{matrix} 0 & 0 \\ 0 & 0\end{matrix}\right)\ 
\]
We can now perturb $A$ with a perturbation $P$ as in
Eq.~\eqref{eq:P-GHZ}, with $a_i=d_i$ and $b_i=c_i$, $i=0,1$. A calculation similar to the one for the GHZ state
shows that in the limit of a vanishing perturbation, the ground space on
three sites is $S_3=\spann\{\ket{000},\ket{W}\}$, where
$\ket{W}=(\ket{001}+\ket{010}+\ket{001})/\sqrt{3}$ (as long as $b_1\ne0$
or $c_1\ne0$); the uncle Hamiltonian $h'$ is the projector onto
$S_3^\perp$,
\begin{equation}
\label{eq:000-uncle-3sites}
h'=\I - \big[\ket{000}\bra{000}+\ket{W}\bra{W}\big]\ .
\end{equation}
For an open chain of length $n$, $H'=\sum h'$ has the two
ground states
\begin{align*}
\ket{0_n}&=\ket{00\cdots 0} \ \ \mathrm{   and} 
\\
 \ket{\mathrm{W}_n}&=\ket{100\cdots 0}+
    \ket{010\cdots0}+\ldots+\ket{000\cdots 01}\ .
\end{align*}
Different from the GHZ case, the extra state $\ket{W_N}$ does not
disappear from the kernel when closing the boundaries -- the uncle Hamiltonian on a chain
with periodic boundaries has a two-dimensional ground space
$\spann\{\ket{0_N},\ket{W_N}\}$. Note, however, that the thermodynamic
limit of $\ket{0_N}$ and $\ket{W_N}$ is the same, and thus, the ground
space collapses to the original one in the thermodynamic limit.

Again, we can construct gapless excitations by considering the states
\begin{equation}
\ket{\phi_n}=\sum_{i\ne j} \ket{0\cdots 01^i00\cdots 01^j0\cdots 0}\ .
\end{equation}
As before, they are orthogonal to the ground space, and their energy is 
$O(1/N)$. Alternatively, we could have also choosen a W state with
momentum, 
\[
 \ket{\varphi_{N,k}}=\sum_j e^\frac{2\pi i kj}{N}
\ket{0\cdots 001^j00\cdots 0}, \ k\ne 0\ ,
\]
which are also orthogonal to the ground space, and have energy $O(k^2/N^2)$.

Again, the spectrum is dense in the thermodynamic limit.  This can once
more be verified by a mapping to the XY model or, directly in the
thermodynamic limit, using the methods described in Sec.~\ref{sec:uncle-spectrum}.

\subsection{Relation to the XY model\label{sec:GHZ-and-XY}}

Both the uncle Hamiltonian for the GHZ state and for the $\ket{0\dots0}$
state are closely related to the XY model (or, equivalently, to
non-interacting fermions), which can be used to immediately infer that
they are gapless models with continuous spectra. Let us first consider the
uncle Hamiltonian of the GHZ state, Eq.~(\ref{eq:ghz-uncle}): It can be
rewritten as 
\begin{equation}
    \label{eq:xy-model}
h'=-\tfrac14\big[\I\otimes Z\otimes Z + 
    Z\otimes Z\otimes \I + \I\otimes X\otimes \I -Z
\otimes X\otimes Z\big] + \tfrac12\I\otimes\I\otimes\I\ .
\end{equation}
This is exactly the Hamiltonian discussed in Eq. (11) from \cite{wolf:phasetransitions-mps} at
$g=0$\footnote{In fact, the construction in Eq.~(10) of \cite{wolf:phasetransitions-mps} is, up
to a gauge transformation, equivalent to the uncle construction, with
$\varepsilon=\sqrt{g}$.}. It can be solved either by transforming it to
non-interacting fermions, or by a duality transformation to the XY model
\cite{peschel:entanglemententropy-dualitymap} \footnote{The partial isometry $T:\ket{i_1,\dots, i_N}\mapsto \ket{i_1+i_2,\dots,i_N+i_1}$ from the eigenspace of $X^{\otimes N}$ associated to the value $1$  to the even $Z$ parity space induces the duality mapping $X_i\mapsto T\circ X_i\circ T^{-1}=X_{i-1}\otimes X_{i}$ and $Z_i \otimes
Z_{i+1}\mapsto T\circ(Z_i \otimes
Z_{i+1})\circ T^{-1}= Z_i$.}.  The resulting
Hamiltonian is 
\[
H_{XY}=-\tfrac14\sum_{i}\big[\, 
    X_i\otimes X_{i+1}+Y_i\otimes Y_{i+1} +2 Z_i\,\big]
    +\mathrm{const.}
\]
Indeed, this point in the XY model, which can be solved exactly by mapping
it to non-interacting fermions~\cite{lieb:XY}, is known to be gapless with
a continuous spectrum.

Let us now turn to the uncle Hamiltonian (\ref{eq:000-uncle-3sites}) for
the $\ket{0\dots0}$ state.  Let us first replace the uncle Hamiltonian
with a simpler one with the same spectral properties. Namely, let
\begin{equation}
\label{eq:000-uncle-2sites}
\tilde h'=\I-\ket{00}\bra{00}-\ket{\Phi^+}\bra{\Phi^+}\ ,
\end{equation}
with $\ket{\Phi^+}=(\ket{01}+\ket{10})/\sqrt{2}$. We have that
\[
\tfrac12 h'\le 
\tfrac12 (\tilde h'_{12}+\tilde h'_{23}) \le  h'\ ,
\]
which implies that for any finite chain, the ordered eigenvalues
$\lambda_i$ of $H'=\sum h'$ and $\tilde \lambda_i$ of $\tilde H'=\sum
\tilde h'$ are related by
\[
\tfrac12 \lambda_i\le \tilde \lambda_i \le \lambda_i\ .
\]
I.e., if we want to determine essential spectral properties of $H'$, such
as whether it has a continuous spectrum, we can equally well study $\tilde
H'$. Since $\tilde h'$ can be rewritten as
\[
\tilde h' =-\tfrac14\big[X\otimes X + Y\otimes Y + Z\otimes \I +
    \I\otimes Z\big] + \tfrac12 \I\otimes \I\ ,
\]
this yet again gives rise to the same point of the XY model,
Eq.~(\ref{eq:xy-model}), proving that the
Hamiltonians~Eqs.~(\ref{eq:000-uncle-3sites}) and
(\ref{eq:000-uncle-2sites}) have a continuous spectrum.

\section{ \label{sec:main}
Properties of the uncle Hamiltonian}

In this section, we will see that the observations we made for the GHZ
uncle Hamiltonian generalize to uncle Hamiltonians of arbitrary MPS with
degenerate ground states (under some generic conditions): Their ground
state space is equal to the ground state space of the parent Hamiltonian,
they are gapless, and they have a continuous spectrum. This section will
also contain the proofs which have been omitted for the special case of
the GHZ state.

For simplicity, we will focus here on the case where the MPS tensor $C_i$
in its standard form, Theorem~\ref{thm:std-form-inj}, has two blocks,
$C_i=A_i\oplus B_i$, but the same procedure can be followed in the general
case: The results are completely analogous in case of multiple different
blocks, but there are some differerences if there are blocks with a
multiplicity larger than one.  We will comment on this particular case in
Section~\ref{sec:injectiveMPS}.

Thus, in this section we will be dealing with an MPS $\ket{M(C)}$,
\[
C=\left( \begin{array}{cc} A & 0 \\
0 & B
\end{array} \right)\ ,
\]
where both $A$ and $B$ are injective.  We will choose $A$ and $B$ in the
normal form of Theorem~\ref{thm:std-form-inj}.
The \emph{parent} Hamiltonian of this MPS consists of local projectors
$\Pi[C\cont C]$ with kernels 
\[
\spann\{C\cont C\} = \spann\{A\cont A\}+\spann\{B\cont B\}\ ,
\]
where the two-dimensional ground state space is
spanned by $\ket{M(A)}$ and $\ket{M(B)}$.

We will need two basic lemmas, which we show next.

\begin{lemma}[Consequences of injectivity]
\label{lemma:inj-consequences}
The following three properties are equivalent, 
$1\Leftrightarrow 2\Leftrightarrow 3$:
\begin{enumerate}
\item $A$ is injective.
\item For any $X$, there exists an $\ket a$ such that 
\begin{equation}
    \label{eq:injective-aA-X}
\sum_i \braket{a}{i}A_i 
\equiv\ 
\raisebox{0.6em}{\figgg{0.18}{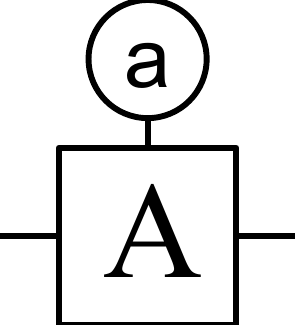}}
\ = \ 
\figgg{0.18}{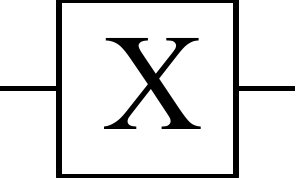}
\ \equiv  X\ .
\end{equation}
\item There exists a tensor $A^{-1}$ such that
\begin{equation}
    \label{eq:injective-Ainv}
\sum_i (A_i)_{\alpha\beta}
    ( (A^{-1})_i)_{\alpha'\beta'} \equiv\ 
\figgg{0.18}{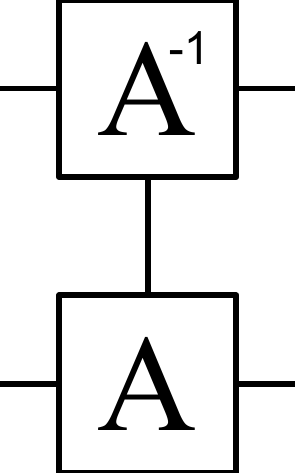} \ = \
\figgg{0.18}{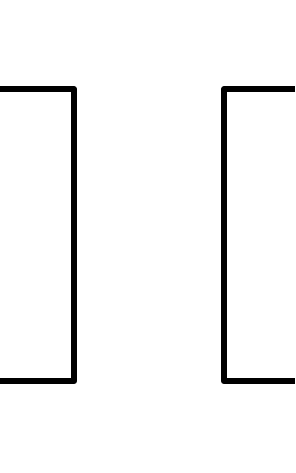} \ = \delta_{\alpha\alpha'}\delta_{\beta\beta'}
\end{equation}
(``left inverse to $A$'').
\end{enumerate}
Also, the following three are equivalent, 
$4\Leftrightarrow 5\Leftrightarrow 6$:
\begin{enumerate}
\setcounter{enumi}{3}
\item $\left(\begin{smallmatrix}A & R\\L & B\end{smallmatrix}\right)$
    is injective
\item For any $X$ of the dimensions of $A$, there exists an $\ket{a}$ such
that (\ref{eq:injective-aA-X}) holds, and additionally
\[
\raisebox{0.6em}{\figgg{0.18}{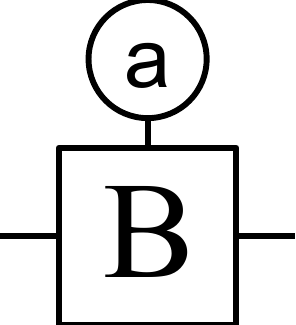}}
\ \equiv\ 
\raisebox{0.6em}{\figgg{0.18}{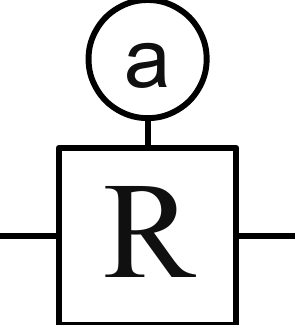}}
\ \equiv\ 
\raisebox{0.6em}{\figgg{0.18}{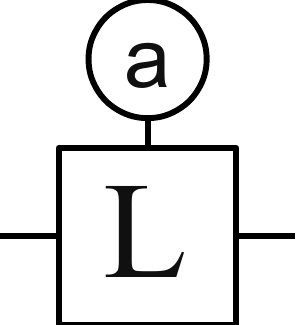}}
\ \equiv 0\ ,
\]
and the corresponding statement holds for the other three blocks.
\item There exists a tensor $A^{-1}$ such that~(\ref{eq:injective-Ainv})
holds, and additionally 
\[
\figgg{0.18}{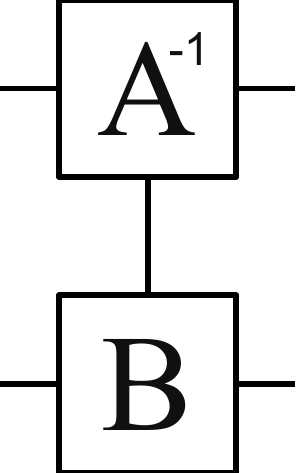}
\ \equiv\ 
\figgg{0.18}{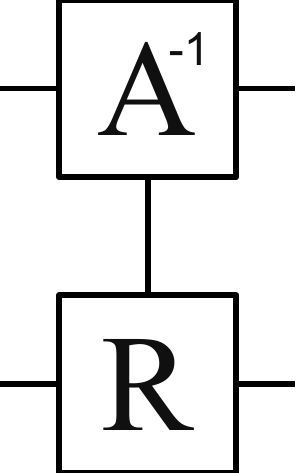}
\ \equiv\ 
\figgg{0.18}{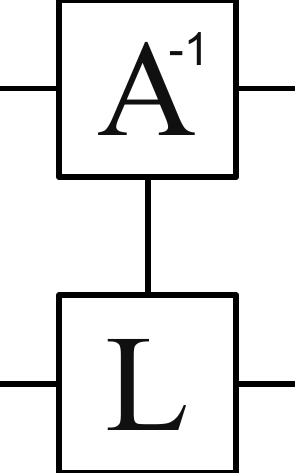}
\ \equiv 0\ ,
\]
and the corresponding statement holds for the other three blocks.
\end{enumerate}
\end{lemma}

\begin{proof}
$1\Leftrightarrow 2$ since by definition, injectivity means that the $A_i$
span the whole matrix algebra. $2\Rightarrow 3$ by defining
$\ket{a_{\alpha\beta}}$ such that\ in (\ref{eq:injective-aA-X}),
$X=\ket{\alpha}\bra{\beta}$, and choosing $((A^{-1})_i)_{\alpha\beta} =
\braket{a_{\alpha\beta}}{i}$, and $3\Rightarrow 2$ by setting
$\braket{a}{i}=\tr[(A^{-1})_i X^T]$.
$4\Rightarrow 5$ by considering equivalence between 1 and 2 and the matrix \[\tilde{X}=\left( \begin{array}{cc}
X & 0\\
0 & 0 
\end{array} \right), \]
or the corresponding matrices for the other blocks.

$5\Rightarrow 4$ since any matrix can be block-decomposed as 
\[\left( \begin{array}{cc}
X & Z\\
W & Y 
\end{array} \right),\] and for these blocks there exist vectors $\ket{a_X}, \ket{b_Y}, \ket{r_Z}, \ket{l_W}$ that give rise to $X$, $Y$, $Z$ and $W$ when applied to $A$, $B$, $R$ and $L$ respectively, and 0 when applied to the other blocks. Thus we can consider the sum $\ket{a_X}+\ket{b_Y}+\ket{r_Z}+\ket{l_W}$ to satisfy condition 2, and therefore injectivity of tensor in 4.

$5\Rightarrow 6$ by defining
$\ket{a_{\alpha\beta}}$ such that in (\ref{eq:injective-aA-X})
$X=\ket{\alpha}\bra{\beta}$ if both indices correspond to the $A$ block or 0 otherwise, and choosing $((A^{-1})_i)_{\alpha\beta} =
\braket{a_{\alpha\beta}}{i}$,  and $6\Rightarrow 5$ by setting
$\braket{a}{i}=\tr[(A^{-1})_i \tilde{X}^T]$, with $\tilde{X}=\left( \begin{array}{cc}
X & 0\\
0 & 0 
\end{array} \right)$.
\hspace*{\fill}\qed
\end{proof}

\begin{lemma}[Gauge transformations for span]
\label{lemma:span-gaugeinvariance}
Let $\mc L:\mc M_D\rightarrow \mc M_D$ be an invertible map on $D\times D$
matrices.  Then, $\spann\{T\}=\spann\{\mc L(T)\}$, and equally
$\Pi[T]=\Pi[\mc L(T)]$, where the natural action of $\mc L$ on three-index tensors $(T_i)_{\alpha\beta}$ is given by
$[\mc L(T)]_i=\mc L(T_i)$.
\end{lemma}

\begin{proof}
This follows from $\trace[\mc L(T_i)X]=\trace[T_i(\mc L^*(X^*))^*]$, for every $X\in \mc M_D$.
\end{proof}

\subsection{Form of the uncle Hamiltonian}

We will now determine the general form of the uncle Hamiltonian
(including whether the limit in its definition exists).

\begin{theorem}\label{thm:unclehamiltonian}
Let $\ket{M(C)}$, be an MPS with block-injective $C=A\oplus B$, 
 and let
\[
P=\left( \begin{array}{cc}
P^A & R \\
L & P^B
\end{array} \right)
\]
be an arbitrary tensor, such that the ``uncle tensor''
\begin{equation}
\label{eq:uncletensor-2site}
U=\left( \begin{array}{cc}
    A\cont A & A \cont  R+R\cont B \\
    B\cont L+L\cont A & B\cont B
    \end{array} \right)\ .
\end{equation}
is injective.  Then, the uncle Hamiltonian induced by $P$ exists and is a
sum of local terms $h_P' = \Pi[U]$.
\end{theorem}

\begin{proof}
Consider a perturbation $C_\varepsilon = C + \varepsilon P$ of the MPS
$\ket{M(C)}$.  We have that
\begin{equation}
C_\varepsilon \cont C_\varepsilon=
    \left( \begin{array}{cc}
	A\cont A + O(\varepsilon)  & 
	\varepsilon (A\cont R+R\cont B) + O(\varepsilon ^2) \\
	\varepsilon(B\cont L+L\cont A) + O(\varepsilon ^2) & 
	B\cont B + O(\varepsilon)
    \end{array} \right)\ ;
\end{equation}
clearly, for $\veps\ne0$ the map $\mc L_\veps$ which multiplies the
off-diagonal blocks by $1/\veps$ is invertible and thus (following 
Lemma~\ref{lemma:span-gaugeinvariance})
\[
\Pi[C_\veps\cont C_\veps] =
\Pi[\mc L_\veps(C_\veps\cont C_\veps)]
=\Pi[U+O(\veps)]\ .
\]
Following Lemma~\ref{lemma:span-continuity}, the limit $\lim_{\varepsilon \to 0}
\Pi[U+O(\veps)]$ exists whenever $U$ is injective, and equals $\Pi[U]$.
\hspace*{\fill}\qed
\end{proof}

The required injectivity of the tensor in Eq.~(\ref{eq:uncletensor-2site})
follows in particular from the following condition on the perturbation.
\begin{definition}[Injective perturbation]
A perturbation $P$ of an MPS $\ket{M(C)}$ (in the notation of
Theorem~\ref{thm:gs-subspace-parent}) is called \emph{injective} if
\begin{equation}
\label{eq:inj-condition-tensor}
\left(\begin{matrix}A&R\\L&B\end{matrix}\right)
\end{equation}
is an injective tensor.
\end{definition}

\begin{lemma}
\label{lemma:inj-perturbation}
If a perturbation $P$ is injective, then the resulting ``uncle tensor''
$U$, Eq.~(\ref{eq:uncletensor-2site}), is injective.
\end{lemma}
\begin{proof}

Let us consider condition 5 from Lemma \ref{lemma:inj-consequences}, and any $X$ of the dimensions of $A\cont A$. Since the perturbation tensor is injective, there exist vectors $\ket{a}$ and $\ket{a'}$ such that 
\[
\raisebox{0.6em}{\figgg{0.18}{Aa}}
\ = \ 
\figgg{0.18}{X}
\ \equiv  X\ ,\  
\raisebox{0.6em}{\figgg{0.18}{Ba}}
\ \equiv\ 
\raisebox{0.6em}{\figgg{0.18}{Ra}}
\ \equiv\ 
\raisebox{0.6em}{\figgg{0.18}{La}}
\ \equiv 0,
\]
\[
\raisebox{0.6em}{\figgg{0.18}{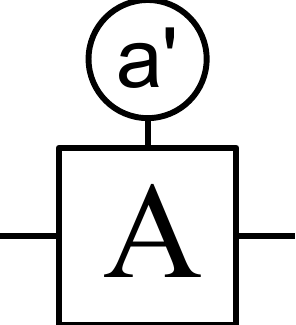}}
\ = \I \ ,\ \mathrm{and}\   
\raisebox{0.6em}{\figgg{0.18}{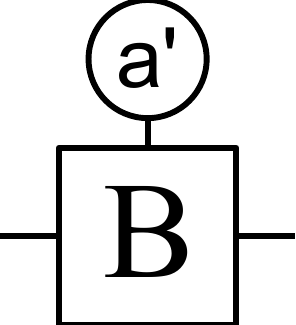}}
\ \equiv\ 
\raisebox{0.6em}{\figgg{0.18}{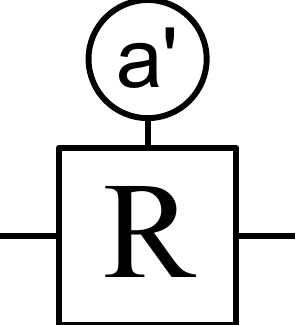}}
\ \equiv\ 
\raisebox{0.6em}{\figgg{0.18}{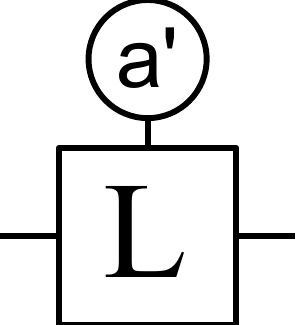}}
\ \equiv 0\ .
\]

The product $\ket{a}\otimes\ket{a'}$ yields then for the uncle tensor
\[
\raisebox{0.6em}{\figgg{0.18}{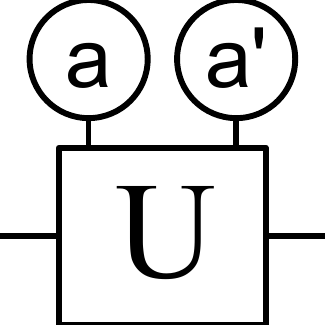}}\ = \ 
\left( \begin{array}{cc}
X & 0 \\
0 & 0
\end{array} \right).
\]

A similar reasoning can be followed for the rest of blocks, thus satisfying condition 5 from Lemma \ref{lemma:inj-consequences}.
\hspace*{\fill}\qed
\end{proof}

A perturbation $P$ is generically injective if $d\ge D^2$, and we will
assume injective perturbations in the following.  Note that unlike for the
$\ket{\mathrm{GHZ}}$ state, the uncle Hamiltonian does in general depend on the
perturbation (though only on its off-diagonal blocks $R$ and
$L$).\footnote{
If $D\le d< D^2$ and $U$ [Eq.~(\ref{eq:uncletensor-2site})]
is injective one can construct both parent and uncle Hamiltonians
from projectors onto the span of three consecutive sites -- this
is what we did for the GHZ example.  Then, the uncle 
Hamiltonian is the projector associated
to 
\[
\left( \begin{array}{c@{\quad}c}
A\cont A\cont A & A\cont A\cont R+A\cont R\cont B+R\cont B\cont B\\
B\cont B\cont L+B\cont L\cont A+L\cont A\cont A & B\cont B\cont B
\end{array} \right)\ .
\]
While we will restrict to injective perturbations for clarity, the same
steps can be followed assuming only injectivity of $U$.
}

\subsection{Ground space of the uncle Hamiltonian}

In the following, we study the ground state space of the uncle
Hamiltonian. Throughout, we will restrict to injective perturbations.
\begin{theorem}
\label{thm:gs-subspace-uncle}
Let $P$ be an injective perturbation of an MPS $\ket{M(C)}$, $C=A\oplus
B$.  Then, the ground space of the uncle Hamiltonian $H'_P$ is spanned by
$\ket{M(A)}$ and $\ket{M(B)}$, and thus equal to the ground space of the
parent Hamiltonian.
\end{theorem}

\begin{proof}
The parent Hamiltonian is frustration free, i.e., its ground states
minimize the energy of each local term. The ground space is thus
\[
\ker (H) = \bigcap \ker (h)\ .
\]
Since $\ker(h) \subset \ker(h'_{P})$, it
follows that 
\[
\ker (H'_P) = \bigcap \ker ( h'_{P})
\supset \ker(H)\ ,
\]
i.e., the uncle Hamiltonian is frustration free, and any ground state of
the parent is also a ground state of the uncle.

In order to classify all states in $\ker (H'_P) = \bigcap \ker (
h'_{P})$, we will follow the same steps as for the proof of
the ground space structure of the parent Hamiltonian~\cite{schuch:peps-sym}:
First, we will prove inductively how the ground space on a chain with open
boundaries, $\bigcap_{i=1}^k \ker{h'_{P,i,i+1}}$, grows -- the
\emph{intersection property}.  Then, we will show how the ground space
changes when we close the boundaries -- the \emph{closure property}.

Recall that the base of induction is the projector on two sites which follows from the uncle tensor in Eq.~(\ref{eq:uncletensor-2site}).

\begin{lemma}[Intersection property]
\label{lemma:intersection}
Given a chain of length $n$, let $S_n$ be the vector space
\begin{align}
S_n&=A_n+ B_n+ R_n+ L_n\ , \ \mathrm{ where}
    \label{eq:Sn-def} \\
A_n&=\Big\{\ \figgg{0.18}{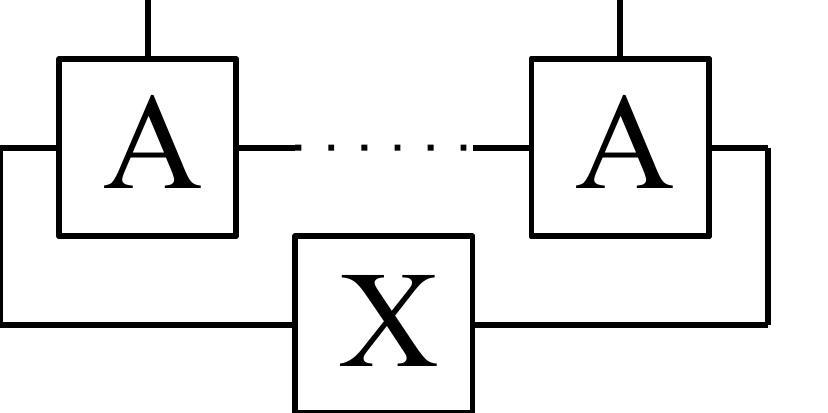} /X\in \mathcal{M}_l\Big\},
    \nonumber\\
B_n&=\Big\{\ \figgg{0.18}{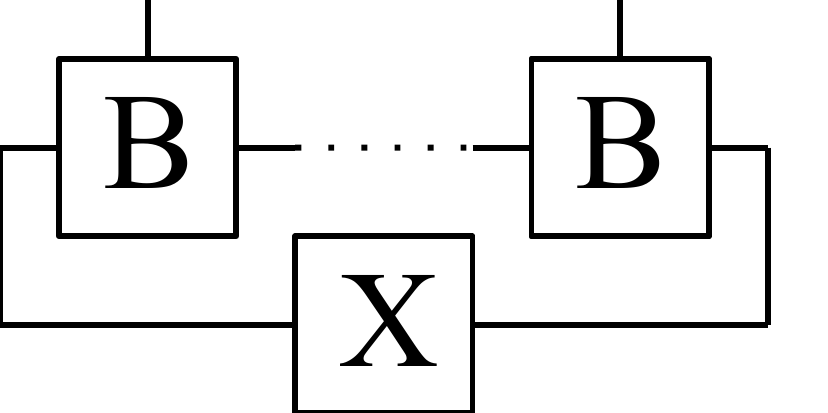} /X\in \mathcal{M}_m\Big\},
    \nonumber\\
R_n&=\Big\{\sum_{\mathrm{pos\,R}}\figgg{0.18}{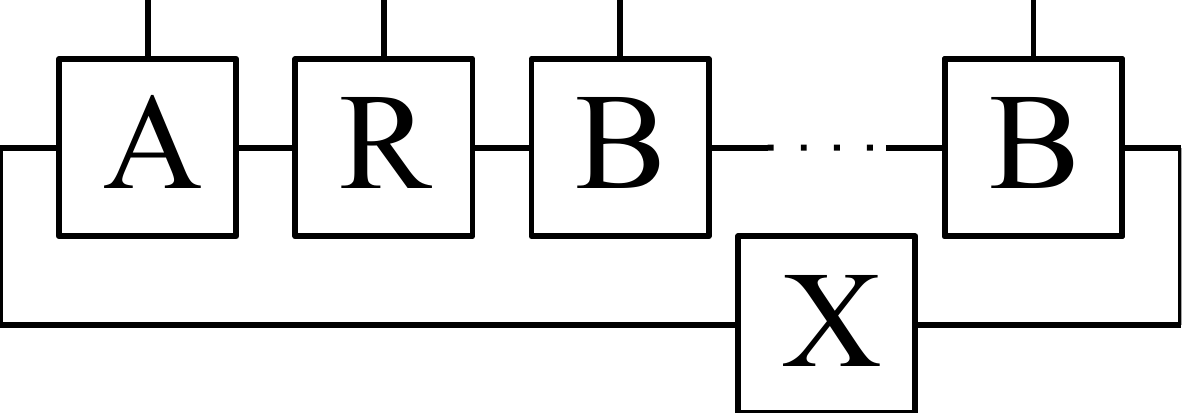}\  /X\in \mathcal{M}_{l\times m}\Big\},
    \nonumber\\
L_n&=\Big\{\sum_{\mathrm{pos\,L}}\figgg{0.18}{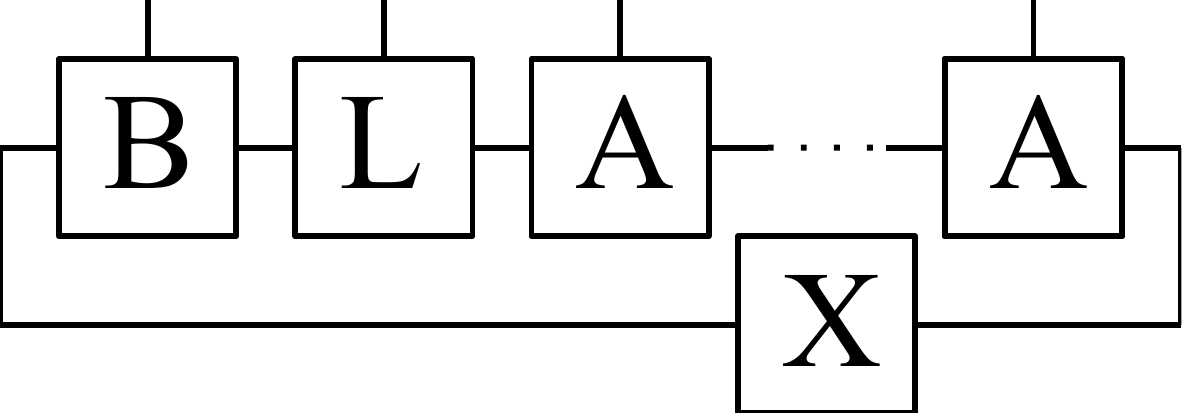}\  /X\in \mathcal{M}_{m\times l}\Big\},
    \nonumber
\end{align}
where the sums run over all possible positions of the $R$ or $L$, padded
with $A$'s and $B$' to the left and right as indicated. Further, let 
$\left( \begin{smallmatrix} A & R \\ L & B \end{smallmatrix} \right)$ be 
injective. Then, the \emph{intersection property} 
$S_n\otimes \C ^d \cap \C^d\otimes S_n= S_{n+1}$ holds. 
\end{lemma}
The proof of the Lemma is given in
Appendix~\ref{appendix:intersection-closure}.

Starting from $S_2=\ker(h'_{1,2})$, and using that 
\[
\bigcap_{i=1}^{k} \ker(h'_{i,i+1}) = 
\bigcap_{i=1}^{k-1} \ker(h'_{i,i+1}) \,\cap\,
\bigcap_{i=2}^{k} \ker(h'_{i,i+1})\ ,
\]
the lemma allows to inductively prove that the ground
space on $n$ consecutive sites with open boundaries is
$\bigcap_{i=1}^{n-1}\ker(h'_{i,i+1}) = S_n$:  It differs from the ground space of the
parent Hamiltonian by the presence of the ``zero momentum domain wall
states`` $R_n$ and $L_n$, analogous to the domain wall states for the GHZ
uncle. It remains to show that these states disappear from the kernel when closing the
boundaries.

\begin{lemma}[Closure property]
\label{lemma:closure} 
Consider a chain of length $N$, and let $S_\mathrm{left}=S_N$ defined on
sites $1,\dots,N$, and $S_\mathrm{right}=S_N$ defined on sites
$2,\dots,N,1$, using the definitions of the previous lemma. (I.e., for
$S_\mathrm{right}$ the ordering of sites is shifted cyclically by one.)
Then, 
\[
S_\mathrm{left}\cap S_\mathrm{right} = \spanned\left\{\ket{M(A)},\ket{M(B)}\right\}.
\]
\end{lemma}
The proof is again given in the Appendix~\ref{appendix:intersection-closure}.

The closure property shows that if we close the boundaries on a chain of
length $N$, we indeed recover the ground space of the parent Hamiltonian,
since
\[
\bigcap_{i=1}^{N} \ker(h'_{i,i+1}) = 
\bigcap_{i=1}^{N-1} \ker(h'_{i,i+1}) \,\cap\,
\bigcap_{i=2}^{N} \ker(h'_{i,i+1}) = S_\mathrm{left}\cap S_\mathrm{right}\ .
\]
Together, the two lemmas thus prove Theorem~\ref{thm:gs-subspace-uncle}.
\hspace*{\fill}\qed
\end{proof}

\subsection{Gaplessness of the uncle Hamiltonian}

One of the key properties of the parent Hamiltonian is that it exhibits a
spectral gap above the ground space \cite{nachtergaele:spectralgap}. On
the other hand, as we will prove in the following the uncle Hamiltonian is
generically gapless.

\begin{theorem}
\label{thm:gapless}
The uncle Hamiltonian $H_P'$ is gapless for almost every $P$.
\end{theorem}

In order to prove this result, we need two basic lemmas which follow directly from the normal form of MPS, and which concern the transfer operators, and another lemma in which we find some low energy states --not eigenstates-- which let us finally prove the theorem. Let us first show them.

\begin{lemma}\label{lemma:exptransferoperator}
Let $A$ denote an injective block of an MPS. Then
\begin{equation}
\label{eq:EAA-errorterm}
(E^A_A)^k\equiv
\left(\figgg{0.18}{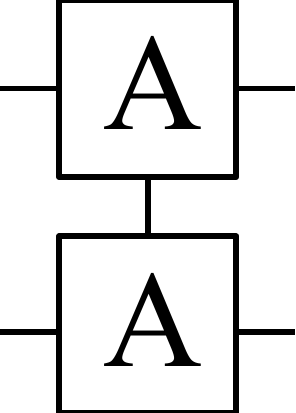}\right)^k = 
    \figgg{0.18}{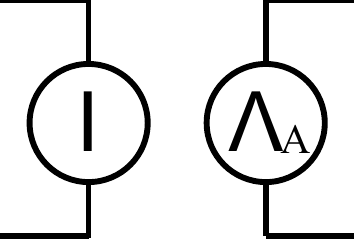} \ + O(e^{-k})\ .
\end{equation}
\end{lemma}

Note that $O(e^{-k})$ denotes a bound up to a constant
in the exponent. This notation will be used throughout this work.

\begin{lemma}\label{lemma:offdiagonal-damping}
Under the conditions of Theorem~\ref{thm:std-form-inj},  the spectral radius
$\rho(E^{A^i}_{A^j})<1$ for $i\ne j$.
\end{lemma}

\begin{proof}
Let us take $X$ such that $E^{A^i}_{A^j}(X)=\sum_k A^j_k X  (A^{i}_k)^*=\lambda X$. We have then, by using Cauchy-Schwarz inequality, that
\begin{gather*}
|\lambda|\trace(X\Lambda_{A^i}X^*)=|\sum_k\trace(A^j_k X  (A^{i}_k)^*\Lambda_{A^i}X^*)|=|\sum_k\trace(X(A^{i}_k)^*\sqrt{\Lambda_{A^i}} \sqrt{\Lambda_{A^i}}X^*A^j_k)|<
\\
<\left(\sum_k \trace(X(A_k^i)^*\Lambda_{A^i}A_k^i X^*)\right)^{1/2}
\left(\sum_k \trace((A_k^j)^*X\Lambda_{A^i}X^*A_k^j)\right)^{1/2}=\trace(X\Lambda_{A^i}X^*),
\end{gather*}
and therefore for any eigenvalue we have $|\lambda|<1$. Note that  the inequality is strict since we have that $\spanned\{\bra{k}A^i\ket{l}\}\cap \spanned\{\bra{m}A^j\ket{n}\}=\{0\}$ due to block injectivity. 
\hspace*{\fill}\qed
\end{proof}

\begin{lemma}\label{lemma:lowenergystates}
For a chain of length $6N+1$, let
\begin{equation} \label{eq:phiN-def}
\ket{\phi_{N}}=\sum_{\substack{-2N\le i\le -N \\ N\le j \le 2N}}
\ket{\zeta_{i,j}}\ ,
\end{equation}
where 
\begin{equation}
\ket{\zeta_{i,j}}\ =\ \figgg{0.18}{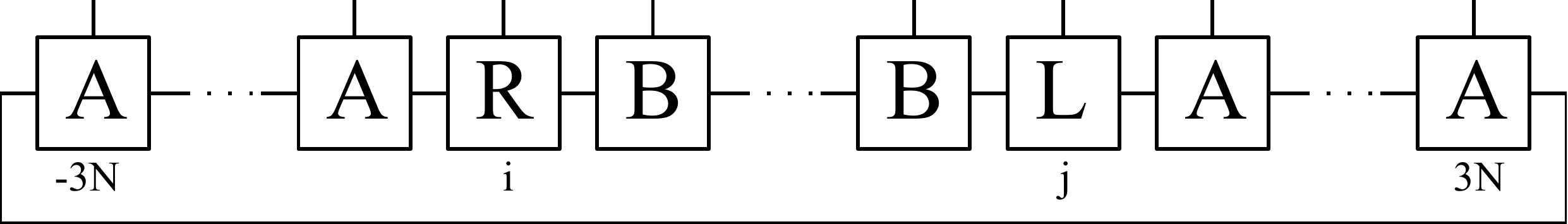}\quad.
\end{equation}
Then, for almost every $R$ and $L$ (and thus almost every $P$), the
following holds:
\begin{enumerate}
\item $\braket{\phi_N}{\phi_N}=\Theta(N^2)$.
\item $\braket{M(A)}{\phi_N}=O(e^{-N})$ and $\braket{M(B)}{\phi_N}=O(e^{-N})$.
\item $\bra{\phi_N}H'_P\ket{\phi_N}=O(N)$.
\end{enumerate}
Here, $\Theta(\cdot)$ denotes both lower and upper bounds on the scaling.
\end{lemma}
Note that $\ket{\phi_N}$ generalizes the GHZ ``boundary wall'' ansatz for
low energy states, Eq.~(\ref{lowenergystatesGHZfinitechain}). The range
for $i$ and $j$ in (\ref{eq:phiN-def}) is chosen such that $R$ and $L$
move over a region of size $N$ each, leaving two separating regions of
length $2N$ each which contain only $A$ or $B$ tensors, respectively.

\begin{proof}

\noindent
\emph{1) $\braket{\phi_N}{\phi_N}=\Theta(N^2)$:} Results from Lemma \ref{lemma:exptransferoperator} allow us to approximate 
\begin{equation}
\braket{\zeta_{i,j}}{\zeta_{k,l}}=
\figgg{0.18}{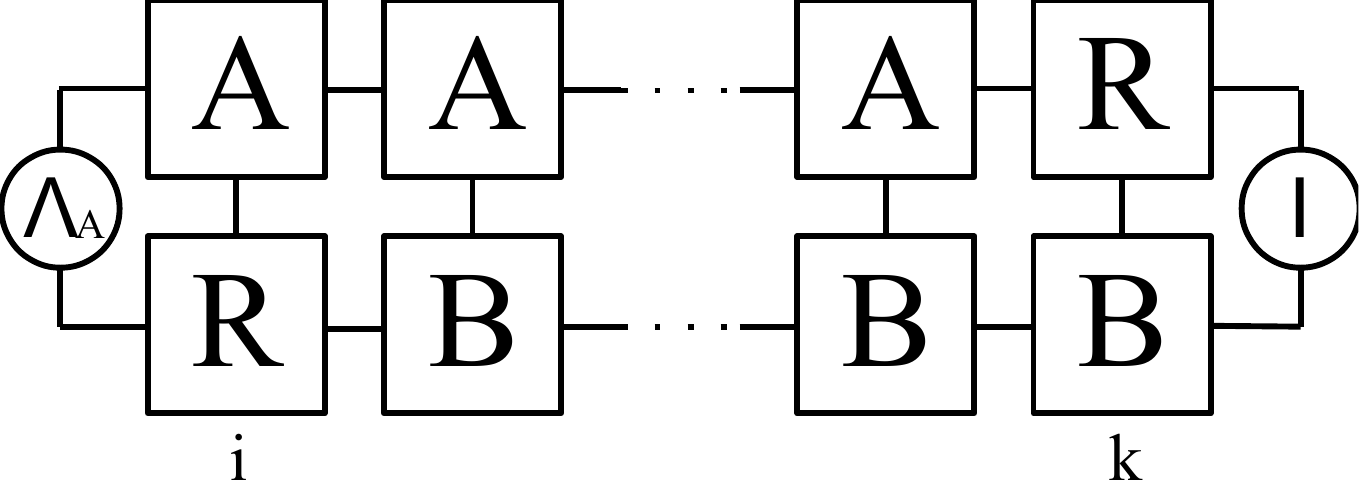} \  \figgg{0.18}{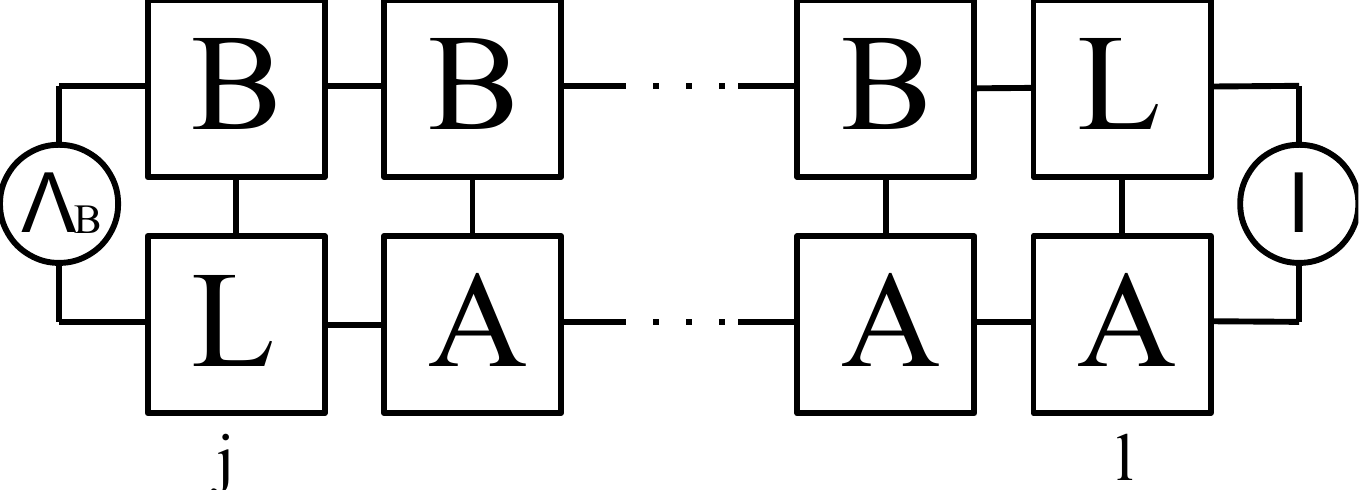}
+O(e^{-2N})\ ,
\end{equation}
for $i<k$ and $j<l$, and correspondingly for the other cases. [Note that
we have used $\rho(E_B^A)=\rho(E_A^B)\le 1$,
Lemma~\ref{lemma:offdiagonal-damping}, together with
Eq.~(\ref{eq:EAA-errorterm}) to bound the error term.] It follows that
$\braket{\phi_N}{\phi_N} = \Xi_R \Xi_L+O(e^{-2N})$, with 
\begin{align*}
\Xi_R := &
\left({N}  +1\right)\figgg{0.18}{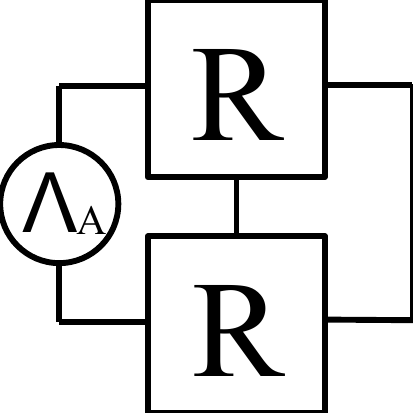} \ + \ 
    \figgg{0.18}{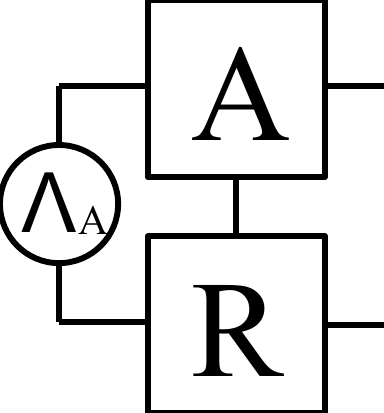} \left(\sum_{n=0}^N 
    \left({N}-n\right)\figgg{0.18}{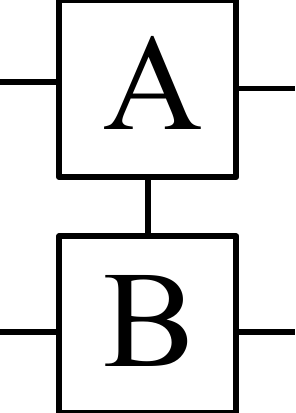}^n \right) \figgg{0.18}{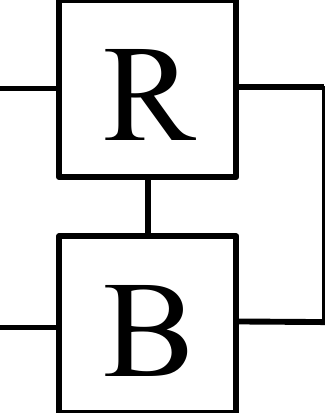} 
    \quad +
\\[1em]
    &  + \ \figgg{0.18}{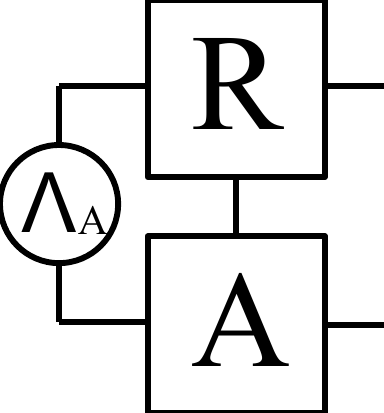} 
    \left(\sum_{n=0}^N \left({N}-n\right)\figgg{0.18}{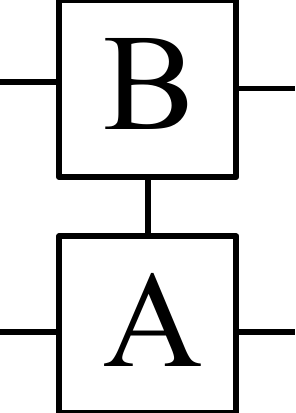}^n \right) 
    \figgg{0.18}{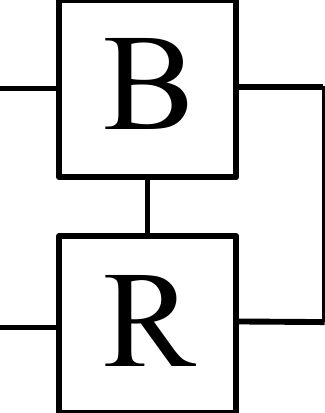}\quad ,
\end{align*}
and correspondingly for $\Xi_L$.  Using 
\[
\sum_{n=0}^N \left(N-n\right)\figgg{0.18}{AB}^n
= \frac{N\I - \frac{E^A_B -(E^A_B)^{N+1}}{\I-E^A_B}}{\I-E^A_B} 
=\frac{N \I }{\I-E^A_B}+O(1)\ ,
\]
we find that $\Xi_R=C_RN+O(1)$, with 
\[
C_R =\ \figgg{0.18}{sum1} \ + \figgg{0.18}{sum21}{(\I-E^A_B)}^{-1}
\figgg{0.18}{sum22} \ + \figgg{0.18}{sum31}{(\I-E^B_A)}^{-1}
\figgg{0.18}{sum32}\ .
\]
$C_R$ is a quadratic function in $R$ which does not vanish identically 
(e.g., there exists an $R$ for which $\Pi_R\Pi_A=0$ and the first term is
non-zero).  Thus, the $R$ for which $C_R=0$ form an algebraic variety of
smaller dimension, and $C_R\ne0$ for almost all $R$.\footnote{If
$\Pi_A\Pi_B=0$, i.e., $E^A_B=0$, such as for the GHZ state, one can prove
that $C_R\ne0$ for any injective perturbation.}

The same argument can be used to see that $\Xi_L=C_LN+O(1)$ where
$C_L\ne0$ for almost all $L$, and thus,
$\braket{\phi_N}{\phi_N}=\Theta(N^2)$ for almost every perturbation as
claimed.

\vspace{1em}\noindent \emph{2) $\braket{M(A)}{\phi_N}=O(e^{-N})$ and 
$\braket{M(B)}{\phi_N}=O(e^{-N})$:}
In the scalar product
\begin{equation}
\braket{M(A)}{\phi_{N}}=\sum_{\substack{-2N\le i\le -N \\ N\le j \le 2N}}
\ \figgg{0.18}{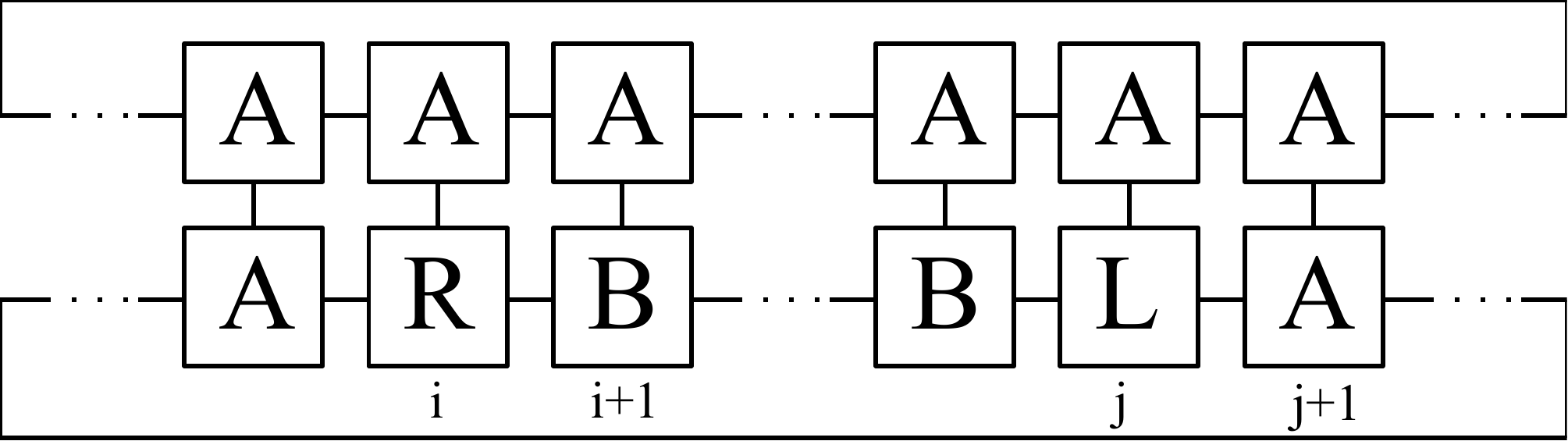}\quad,
\end{equation}
every summand contains $(E^A_B)^{2N}$. Using
Lemma~\ref{lemma:offdiagonal-damping} and the fact that there are only
$O(N^2)$ summands, $\braket{M(A)}{\phi_N}=O(e^{-N})$ follows, and
analogously for $\braket{M(B)}{\phi_N}$.

\vspace{1em}\noindent \emph{3) $\bra{\phi_N}H'_P\ket{\phi_N}=O(N)$:}
The only terms in $H'_P$ which give a non-zero energy are
$h'_{-2N-1,-2N}$, $h'_{-N,-N+1}$, $h'_{N-1,N}$, and $h'_{2N,2N+1}$.
For each of them, $N+1$ summands in \eqref{eq:phiN-def} contribute, and
thus, $\bra{\phi_{N}} H'_P\ket{\phi_{N}}=O(N)$.
\hspace*{\fill}\qed
\end{proof}

\begin{proof}[of Theorem~\ref{thm:gapless}]
For the normalized states
$\ket{\hat\phi_N}:=\ket{\phi_N}/\|\ket{\phi_N}\|$ on a chain of length
$6N+1$ (with $\ket{\phi_N}$ of Lemma~\ref{lemma:lowenergystates}), we have
that $\ket{\hat\phi_N}$ tends to be orthogonal to the ground space of
$H'_P$ and $\bra{\hat\phi_N}H'_P\ket{\hat\phi_N}\rightarrow 0$ as
$N\rightarrow\infty$.  Together with simple spectral decomposition
arguments, this implies the existence of a sequence $\delta_N\rightarrow
0$ such that $H'_P$ (on $6N+1$ sites) has at least one eigenvalue in the
interval $(0,\delta_N)$.  
\hspace*{\fill}\qed
\end{proof}

\subsection{The spectrum of the uncle Hamiltonians is $\R^+$}\label{sec:uncle-spectrum}
In order to study more properties of the spectra of the uncle Hamiltonians
we need to move on to the thermodynamic limit. A formal description of it
via GNS-representations with respect to ground states can be found in Appendix~\ref{appendix:GNS}. The spectrum in the
thermodynamic limit can be found to be the whole positive real line and
the spectra of the finite sized chains can be proven to tend to be dense
in the positive real line.

Through the GNS-representation the problem is translated into studying the action of $H'_P$ on the space
\begin{equation*}
S=\bigcup_{i\le j} S_{i,j}, \hbox{ where } S_{i,j}=\left\{\figgg{0.18}{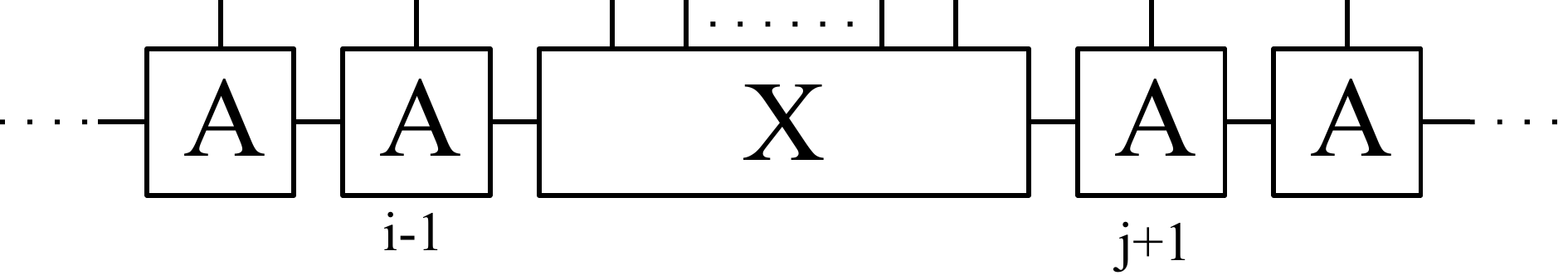},\ X\right\},
\end{equation*}
where $X$ can be any tensor of the corresponding dimensions. The spectrum one must study is that of its unique self-adjoint extension\footnote{We ommit the subscript $P$ for the thermodynamic limit Hamiltonians.} $H'_\w:\bar{S}\lra\bar{S}$, with $\w=\w_A=\figgg{0.18}{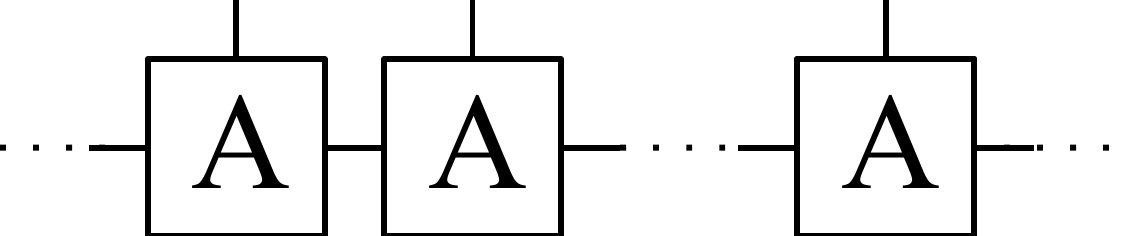}$ the ground state used for the representation. There exists a unique self-adjoint extension because $S$ is a dense set of analytic vectors for $H'_P|_S$, and therefore it is essentially self-adjoint \cite{reedsimon}.

We must note that the choice of either $\w_A$ or $\w_B$ is irrelevant, due to the symmetric role blocks $A$ and $B$ have in the problem.

In first place, we must show that $H'_\w$ is gapless. A family of states related to those we used previously for finite chains in Lemma \ref{lemma:lowenergystates} and Theorem \ref{thm:gapless} let us show the absence of gap:
\begin{equation} \label{eq:phiNlimit-def}
\ket{\phi_{N}}=\sum_{\substack{-2N\le i\le -N \\ N\le j \le 2N}}
\ket{\zeta_{i,j}}\ ,
\end{equation}
where 
\begin{equation}
\ket{\zeta_{i,j}}\ =\ \figgg{0.18}{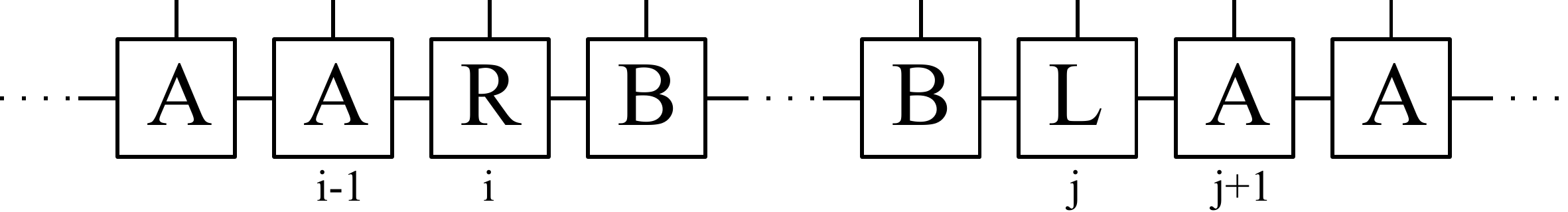}\quad.
\end{equation}

\begin{theorem}
\label{thm:limitunclegapless}
The uncle Hamiltonian $H'_P$ is gapless in the thermodynamic limit for almost every $P$.
\end{theorem}

\begin{proof}
We consider the operator $H'_\w$, which models the thermodynamic limit of $H'_P$ with boundary conditions described by the $A$ tensor --i.e., taking $\w=\w_A$.

The bounds from Lemma \ref{lemma:lowenergystates} also apply to the states in (\ref{eq:phiNlimit-def}), and use of the spectral theorem for unbounded operators (Appendix~\ref{appendix:spectraltheorem}) leads to the fact that $H'_\w$ is gapless.
\hspace*{\fill}\qed
\end{proof}

This last result shows the existence of a sequence of elements in $\sigma(H'_\w)$ tending to 0
%, which we may call $\{\mu_i\}_{i=1}^\infy$
. They will allow us to prove that the spectrum is the entire positive real line, for which we need first the following result.

\begin{proposition} \label{summingeigenvalues}

If some values $\{\lambda_1, \dots , \lambda_n\}$  lie in the spectrum of $H'_\w$ then the sum $\sum_i \lambda_i$ also lies in the spectrum of $H'_\w$.
\end{proposition}

The proof of this result, which can be also found  in Appendix~\ref{appendix:GNS}, stems on the fact that one can find states in $S$, and not necessarily in $\bar S$, `evidencing' a given value is a spectral value. These states can be `concatenated' to prove that any sum of spectral values known to exist is also a spectral value.

\begin{theorem} \label{thm:spectrumuncleghz} The spectrum of the uncle Hamiltonian $H'_P$ in the thermodynamic limit is the whole positive real line for almost every perturbation $P$. 
\end{theorem}

\begin{proof}
The set of finite sums of any sequence of real numbers tending to 0 is dense in $\R^+$. Since there exists a sequence of elements in $\sigma(H'_\w)$ tending to 0 --as it can be deduced from Theorem \ref{thm:limitunclegapless}-- and any finite sum of these elements lies also in $\sigma(H'_\w)$, which is closed, this last spectrum must be equal to $\R^+$. Therefore the spectrum of $H'_P$ is the whole positive real line in the thermodynamic limit.
\hspace*{\fill}\qed
\end{proof}

\subsection*{Spectra for finite chains} After this discussion on thermodynamic limit Hamiltonians we need to go back to the finite chains, and study how the spectra of the uncle Hamiltonians constructed on finite spin chains tend to be dense in $\R^+$ as the size of the chain grows. 

Given $i<j$, $S_{i,j}$ can be easily mapped to any finite chain of lenght $2N+1$ for $N>\max \{i,j\}$ via
\[
\begin{array}{cccc}
e_{N}:&S_{i,j}&\rightarrow& \mathcal{H}_{2{N+1}} \\
 & \figgg{0.17}{GNSvector4} & \mapsto & \figgg{0.18}{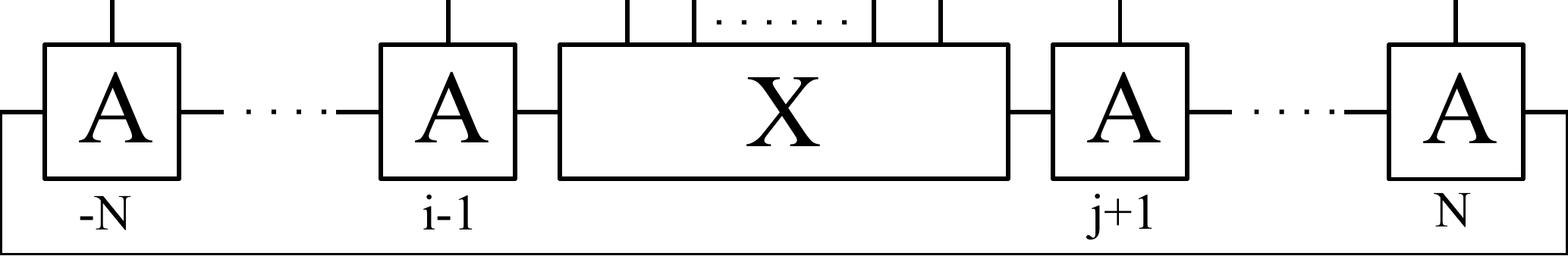}
\end{array},
\]
and this family of maps capture important information since they tend to be isometric embeddings.

\begin{lemma}\label{lemma:almostembeddings}
For fixed $i$ and $j$, $e_N$ is an isometry up to a correction $O(e^{-N})$.
% $S_{i,j}$ and $\varepsilon>0$, there exists a value $N_0$ such that for every $M\ge N_0$ $e_M$ is almost an isometry up to an error $\varepsilon$.
\end{lemma}
\begin{proof}
It follows from Lemma \ref{lemma:exptransferoperator}.
\hspace*{\fill}\qed
\end{proof}

Note that in the case of $G$-isometric MPS \cite{schuch:peps-sym}, these maps are isometries.

Finally, given a spectral value $\lambda$ in the thermodynamic limit spectrum, and a vector $\ket{\varphi}$ in some $S_{i,j}$ `close enough' to evidence $\lambda$ is in this spectrum, one can find that $e_N(\ket{\varphi})$ is also close to evidence $\lambda$ is in the spectrum for the finite chain of length $2N+1$. Clearly we cannot conclude $\lambda$ is this spectrum, but we can find some spectral value not far from it.

One can follow the same procedure for several spectral values $j\lambda$, $j=1,\dots,n$, and by choosing small values of $\lambda$ and carefully handling all these `closenesses', we can find our final result for spectra for finite sized chains, whose proof is detailed also in Appendix~\ref{appendix:GNS}.

\begin{theorem}\label{thm:finitesizespectra}
The spectra of the uncle Hamiltonians for finite size chains tend to be dense in the positive real line.
\end{theorem}

\section{Uncle Hamiltonians for injective spin chains}\label{sec:injectiveMPS}

In this last section, just for completeness of the picture, we will sketch how one can also construct examples of gapless Hamiltonians for injective MPS. As we have seen, in the injective case the parent Hamiltonian construction is robust against perturbations in the matrices defining the MPS. Therefore, in order to construct such examples, one needs to leave the canonical MPS representation. Following the ideas of Sec. ~\ref{sec:uncle-for-000}, we will get such examples by considering particular perturbations of MPS representations with repeated blocks.

To this end, let us start
from an MPS $\ket{M(A)}$ with injective tensor $A$, and let
\[
C=\left( \begin{array}{cc} A & 0\\ 0 & A \end{array} \right)\ .
\]
Then, $\ket{M(C)}=\ket{M(A)}$ denotes another MPS description of the same
state.  We can now consider a perturbation $C+\varepsilon\left(
\begin{array}{cc} P & R\\ R & P \end{array} \right)$ and construct the
corresponding uncle Hamiltonian $H'$. Note that not any perturbation would lead to the same type of result. The tensor $C$ has additional symmetries since both diagonal blocks are the same, and therefore some symmetries are also needed for the perturbation.

The following is the result for the
second type of perturbation.

\begin{theorem}

Let $A$ be the injective tensor description of a given MPS, and let us consider this MPS as described by the non-injective tensor 
\begin{equation*}
\left(
\begin{array}{cc}
A & 0 \\
0 & A 
\end{array}
\right).
\end{equation*} Given a perturbation tensor $C=\left(
\begin{array}{cc} P & R\\ R & P \end{array} \right)$ such that $(\begin{array} {ccc} A & R  \end{array})$ is injective the ground space of the uncle Hamiltonian $H'$  for finite chains is spanned by
\begin{equation}
\raisebox{0.4em}{\figgg{0.18}{AA_A2}} \  \mathrm{ and }  \sum_{\mathrm{pos R}} \figgg{0.18}{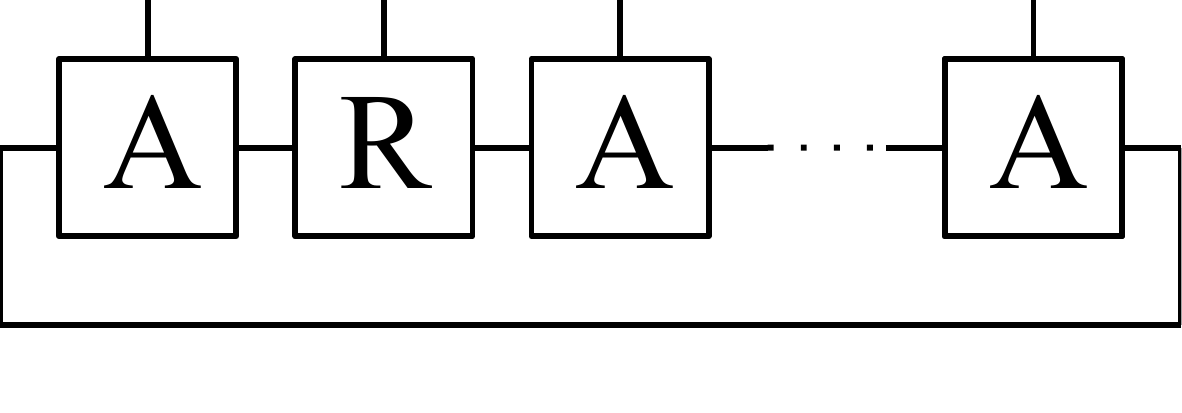}\ .
\end{equation}

This uncle Hamiltonian is local, frustration free, and gapless. The spectra of the finite chains tend to be dense in $\R^+$. In the thermodynamic limit the ground space collapses to a one-dimensional space, which is exactly the ground space of the thermodynamic limit of the parent Hamiltonian, and the spectrum is the whole positive real line $\R^+$.

\end{theorem}

One can check this result following essentially the steps from the preceeding sections. In a chain with length $N$, one can also consider as low energy states the states with momentum
$\ket{\varphi_k}=\sum_j e^{2jk\pi i/N} \ket{\zeta_j}$, with
\begin{equation}
\ket{\zeta_j}=\figgg{0.18}{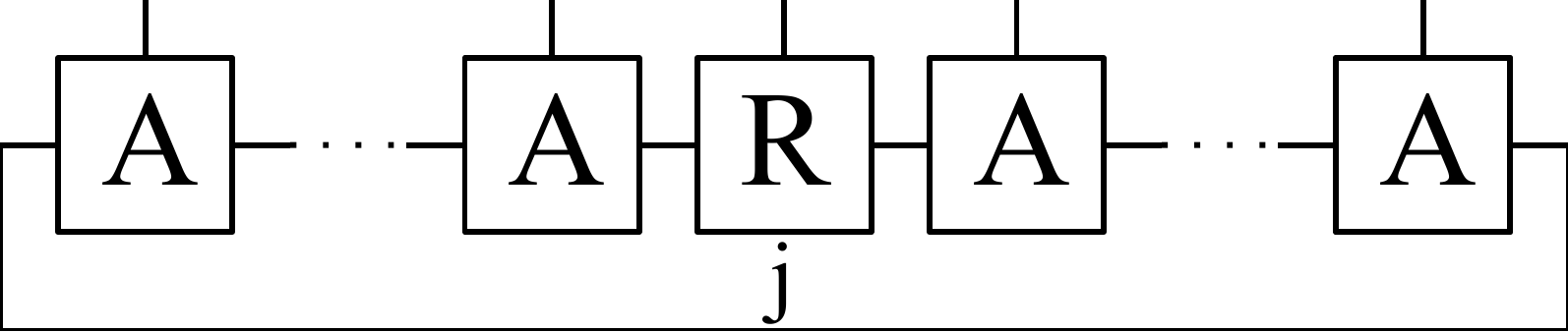}\ .
\end{equation}

The calculations related to states with momentum can be reduced considerably, since $\braket{\varphi_k}{\varphi_k}=N\braket{\zeta^1}{\varphi_k}$ and $\bra{\varphi_k}H'\ket{\varphi_k}=
N\bra{\zeta_1}H'\ket{\varphi_k}$.

When $N$ is large enough (and odd), we have the approximation
\begin{gather*}
\braket{\varphi_k}{\varphi_k}=N\braket{\varphi^1}{\varphi_k}\sim \\
\sim N\left( \figgg{0.18}{sum1} +\figgg{0.18}{sum21} \left(\sum_{n=0}^{\frac{N-1}{2}} e^{kn\frac{2\pi i}{N}} \figgg{0.18}{AA2}^n \right) \figgg{0.18}{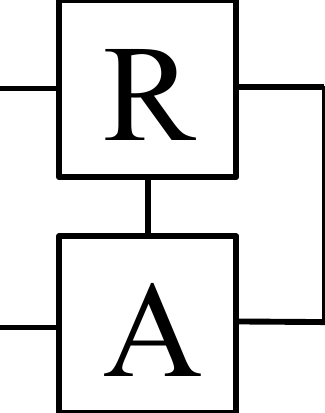}  + \figgg{0.18}{sum31} \left(\sum_{n=0}^{\frac{N-1}{2}} e^{kn\frac{2\pi i}{N}} \figgg{0.18}{AA2}^n \right) \figgg{0.18}{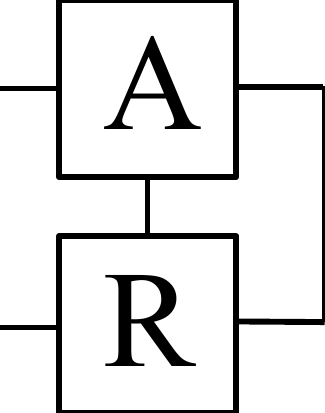}\right)=\\
 =
N\left( \figgg{0.18}{sum1} +\figgg{0.18}{sum21} \frac{(e^{k\frac{2\pi i}{N}}E^A_A)^{\frac{N+1}{2}}-\I}{e^{k\frac{2\pi i}{N}}E^A_A-\I}\figgg{0.18}{sum22inj}  + \figgg{0.18}{sum31}  \frac{(e^{k\frac{2\pi i}{N}}E^A_A)^{\frac{N+1}{2}}-\I}{e^{k\frac{2\pi i}{N}}E^A_A-\I}\figgg{0.18}{sum32inj}\right)% \to 
%\\ \to
%N\left( \figgg{0.18}{sum1} +\figgg{0.18}{sum21} \frac{e^{k\frac{(N+1)\pi i}{N}}\figgg{0.18}{ILA}-\I}{e^{k\frac{2\pi i}{N}}E^A_A-\I}\figgg{0.18}{sum22inj} +\right.
%\left. - \figgg{0.18}{sum31}  \frac{e^{k\frac{(N+1)\pi i}{N}}\figgg{0.18}{ILA}-\I}{e^{k\frac{2\pi i}{N}}E^A_A-\I\right)\figgg{0.18}{sum32inj}},
\end{gather*}

\noindent where $E^A_A=\figgg{0.18}{AA2}$. The expression multiplying $N$ is either divergent or convergent to a constant different from 0 for almost every $R$. Therefore, $\braket{\varphi_k}{\varphi_k}=\Theta(N)$.

The value of $\bra{\varphi_k}H'\ket{\varphi_k}$ can also be approximated as
\begin{gather*}
\bra{\varphi_k}H'\ket{\varphi_k}=
N\bra{\zeta_1}H'\ket{\varphi_k}\sim \\
\sim 
N\left(\figgg{0.18}{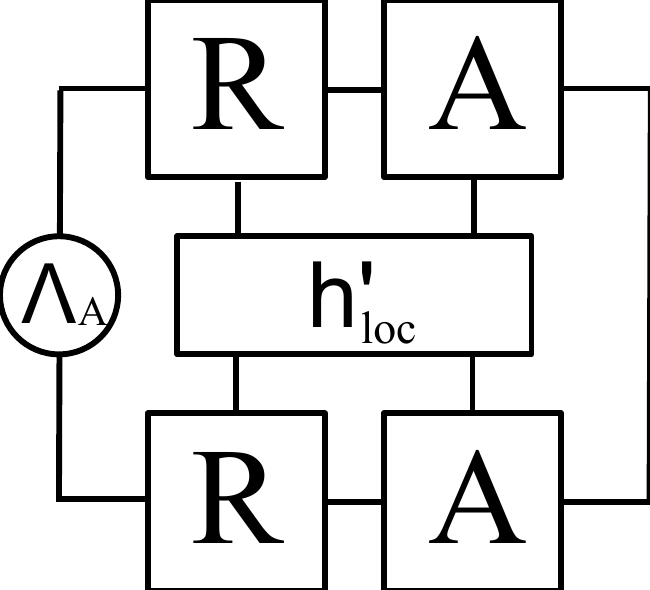}+
e^{k\frac{2\pi i}{N}}\figgg{0.18}{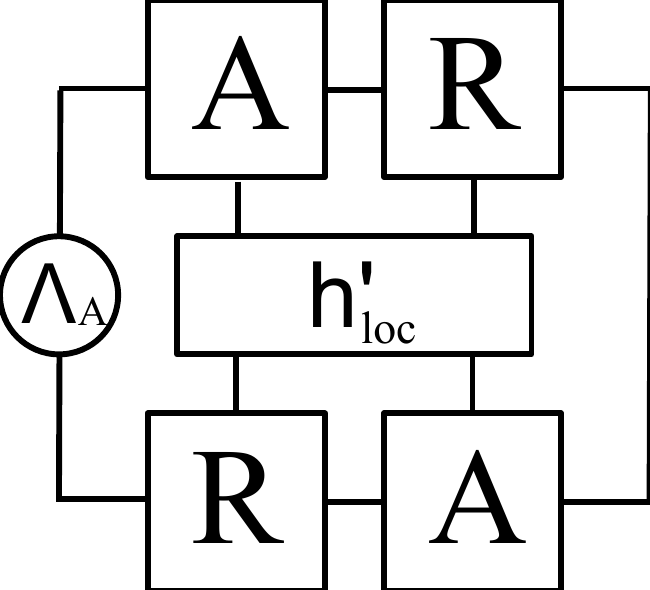}+
e^{-k\frac{2\pi i}{N}}\figgg{0.18}{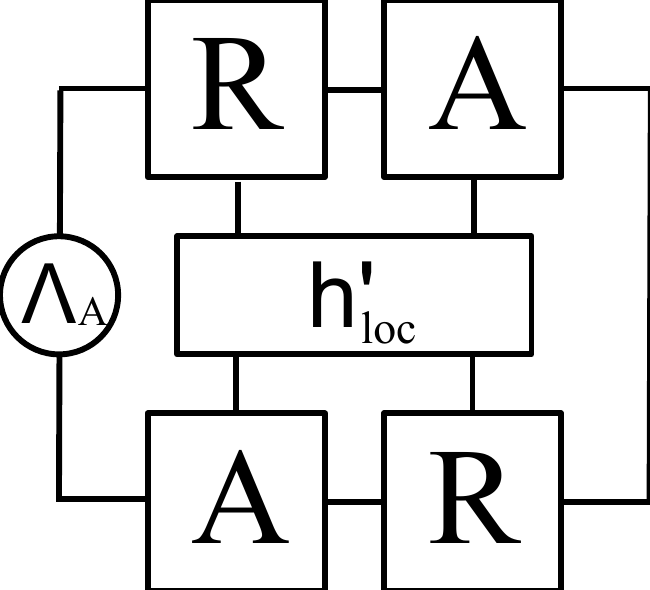}+
\figgg{0.18}{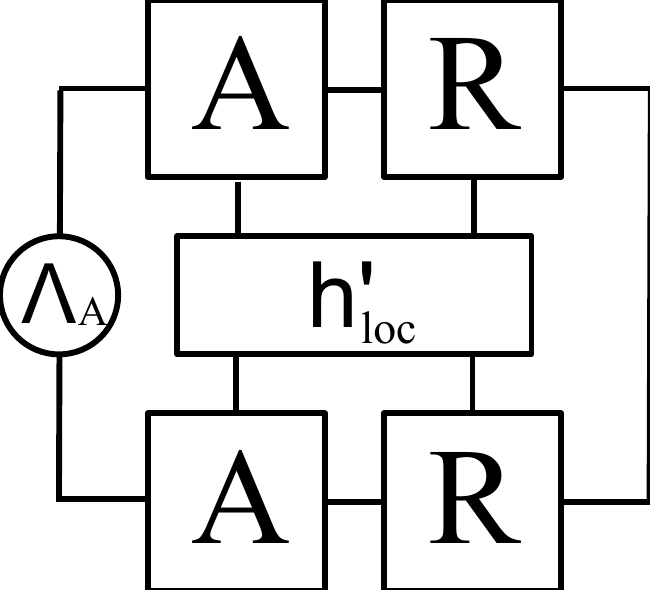}\right)=O(1/N)
\end{gather*}

Hence, there can be found low energy states from this family of states. They are orthogonal to the ground space, and can be used to follow the steps in the previous sections in order to prove that the uncle Hamiltonian is gapless and has spectrum $\R^+$ for most perturbations also for injective MPSs.

\section{Conclusions}

In this work we have shown how to construct new Hamiltonians for Matrix Product States, which we called 'uncle' Hamiltonians. These uncle Hamiltonians  share some of the properties of the parent Hamiltonian (frustration freeness, ground space) but have a completely different spectrum. Instead of having an energy gap above the zero energy space, they are gapless and their spectrum is $[0,\infty)$. We have shown how these uncle Hamiltonians are obtained by doing  linear perturbations on the matrices defining the MPS, considering the parent Hamiltonian of the perturbed MPS and then making the perturbation tend to zero. When the MPS are block-injective, the perturbations leading to an uncle Hamiltonian are essentially all that do not respect the block structure. The uncle Hamiltonian construction might have implications in the problem of classifying quantum phases of matter, due to the role the spectral gap has in this problem, and can provide new tools in order to get more precise classifications. The distinction between {\it good} directions for the perturbation (those preserving the block structure for which we recover the original gapped parent Hamiltonian) and {\it bad} directions (those incompatible with the symmetry and then leading to an uncle) seems to be of upmost importance in the analysis of the robustness of topological quantum phases in 2D \cite{fernandez:toric-code,robustness-2D}. Notice that the topological character of a PEPS -the 2D generalization of an MPS- is given exactly by the existence of such a discrete symmetry \cite{schuch:peps-sym}.

\section*{Acknowledgements}
We specially thank Bruno Nachtergaele for his help.  This work has been
partially funded by the Spanish grants MTM2011-26912
and QUITEMAD, the European projects QUEVADIS and CHIST-ERA CQC, the Alexander
von Humboldt Foundation, the Gordon and Betty Moore Foundation through
Caltech's Center for the Physics of Information, the NSF Grant No.\
PHY-0803371, and the ARO Grant No.\ W911NF-09-1-0442.  We also acknowledge
the hospitality of the Centro de Ciencias de Benasque Pedro Pascual and of
the Perimeter Institute, where part of this work was carried out.

\providecommand{\bysame}{\leavevmode\hbox
to3em{\hrulefill}\thinspace} \providecommand{\andname}{and }

\appendix

\rule{\textwidth}{0.2mm}

\bigskip
\bigskip

\section{Proof of the intersection and closure properties
\label{appendix:intersection-closure}}

\begin{proof}[Intersection property, Lemma~\ref{lemma:intersection}]
We start by proving $S_2\otimes \C ^d \cap \C^d\otimes S_2= S_3$. The
proof will straightforwardly generalize to $S_k\otimes \C ^d \cap
\C^d\otimes S_k= S_{k+1}$, $k>2$.  

Let us first show that $S_2\otimes \C ^d \cap \C^d\otimes S_2 \supset
S_3$.  To this end, let $\ket\varphi\in S_3$, i.e., there exist $X$, $Y$,
$Z$, and $W$ such that
\begin{equation}\label{eq:vector}
\ket{\varphi}= \ \figgg{0.15}{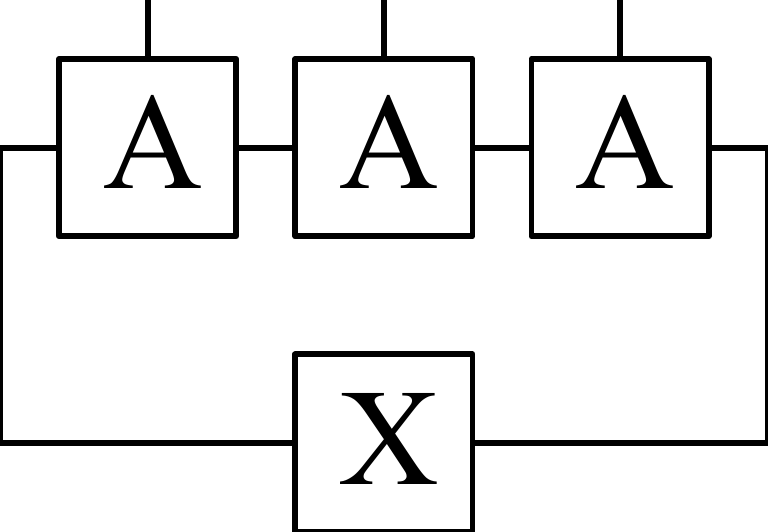}\ + \
\ \figgg{0.15}{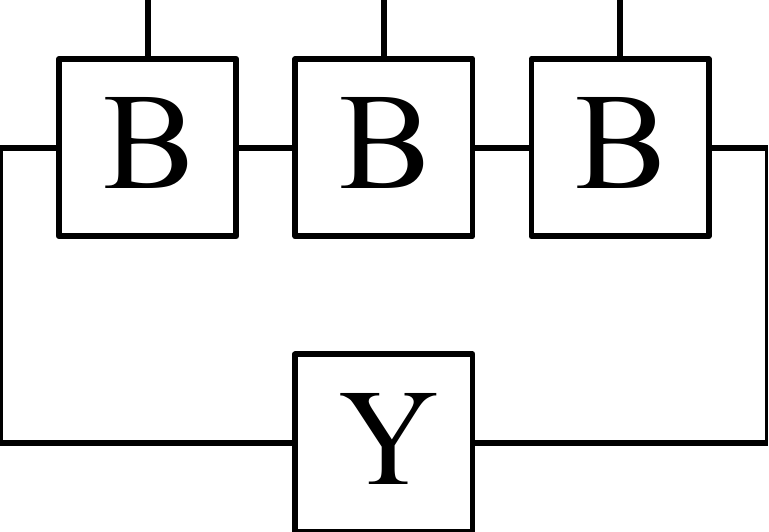} \ +\ 
    \sum _{\mathrm{pos\,R}}  \figgg{0.15}{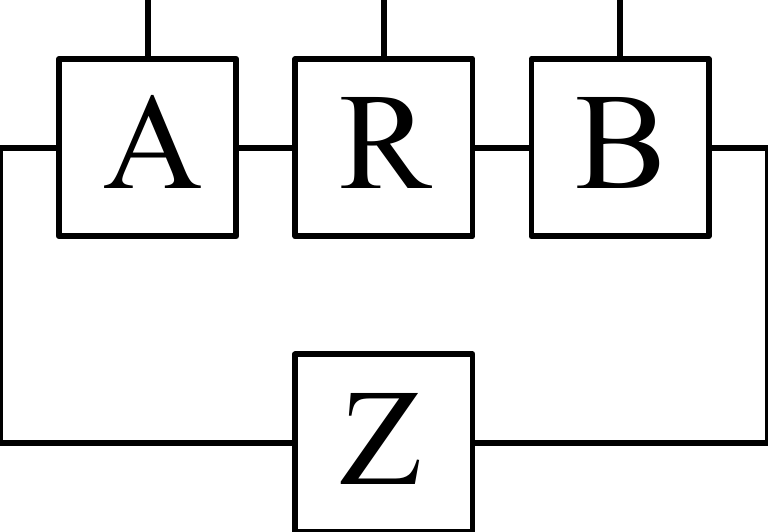}\ +\
 \sum_{\mathrm{pos\,L}} \figgg{0.15}{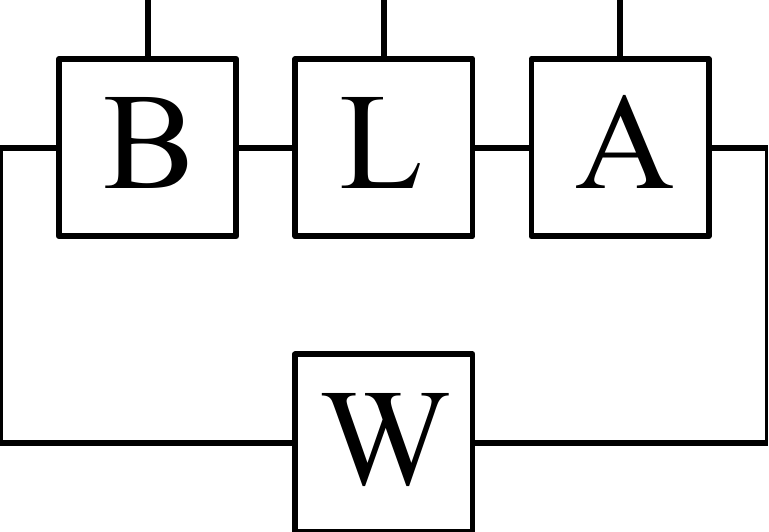}\ ,
\end{equation}
where the sums run over the three possible positions of $R$ and $L$,
respectively. If we now define
\begin{equation}\label{eq:Xprime-decomposition}
\figgg{0.15}{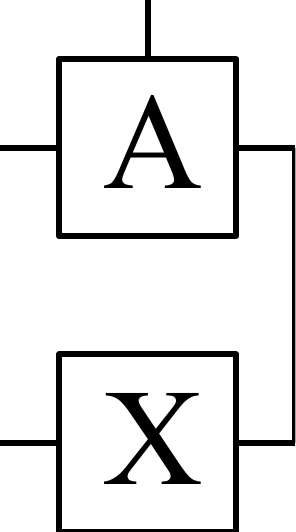}\ +\ \figgg{0.15}{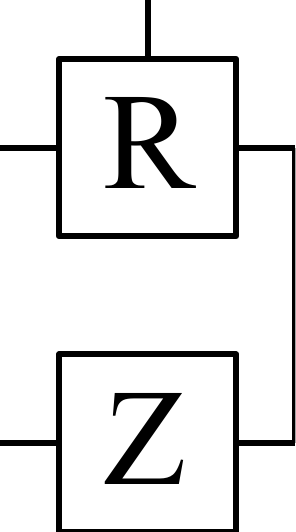}=
    \figgg{0.15}{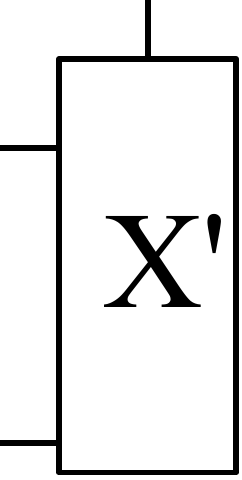}\  ,
\end{equation}
we have that
\[
\figgg{0.15}{AAAX}\ +\ \figgg{0.15}{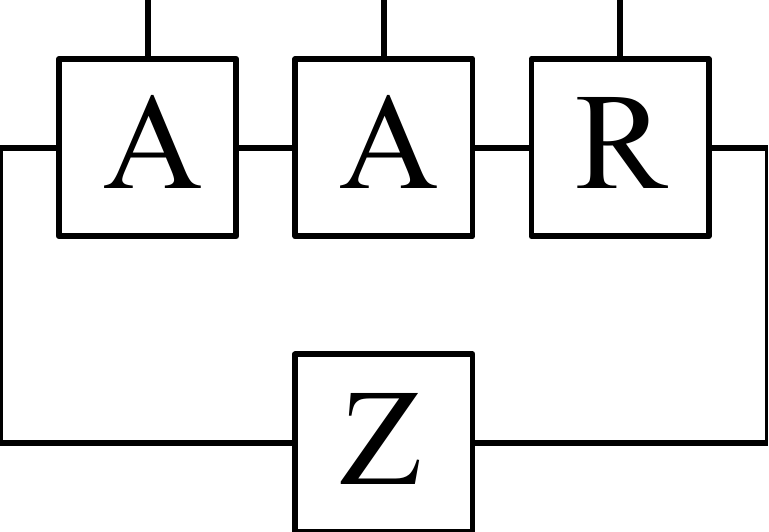}=
    \figgg{0.15}{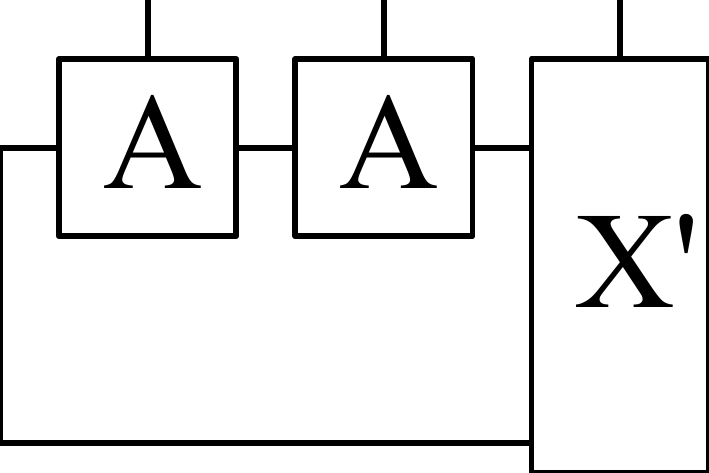}\ \in A_2\otimes \C^d\ ,
\]
and similarly
\begin{gather*}
\figg{0.15}{Dib/BBBY}\ + \  \figg{0.15}{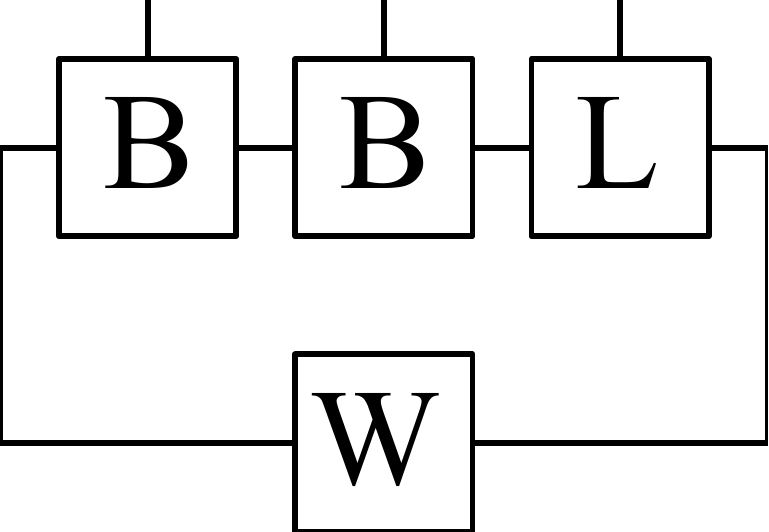}\in B_2 \otimes \C^d\ ,
\\[1em]
\displaystyle\sum_{\mathrm{pos\, R} \neq 3} \figg{0.15}{Dib/ARBZ}\ \in
R_2\otimes \C^d\ ,
\mbox{ and } \displaystyle\sum_{\mathrm{pos\, L} \neq 3}
\figg{0.15}{Dib/BLAW}\ \in L_2\otimes
\C^d\ ,
\end{gather*}
showing that $\ket{\varphi}\in S_2\otimes \C^d$. Similarly, one can show
that $\ket{\varphi}\in \C^d\otimes S_2$, proving that
$S_2\otimes \C ^d \cap \C^d\otimes S_2 \supset S_3$.

Let us now show that conversely, $S_2\otimes \C ^d \cap \C^d\otimes S_2
\subset S_3$. To this end, let $\ket\varphi \in S_2\otimes \C ^d \cap
\C^d\otimes S_2$, i.e., there exist tensors
$X$, $Y$, etc., such that 
\begin{equation}
\label{eq:appA:vectorinintersection}
\begin{aligned} 
\ket{\varphi} 
&=
\figgg{0.15}{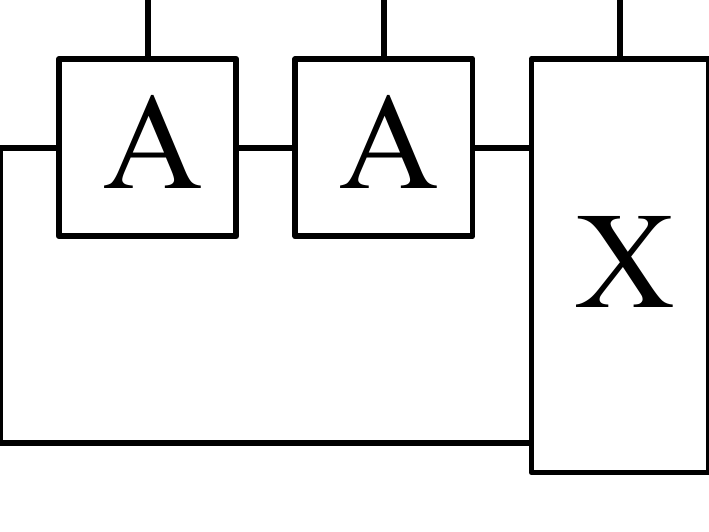}\ + \ \figgg{0.15}{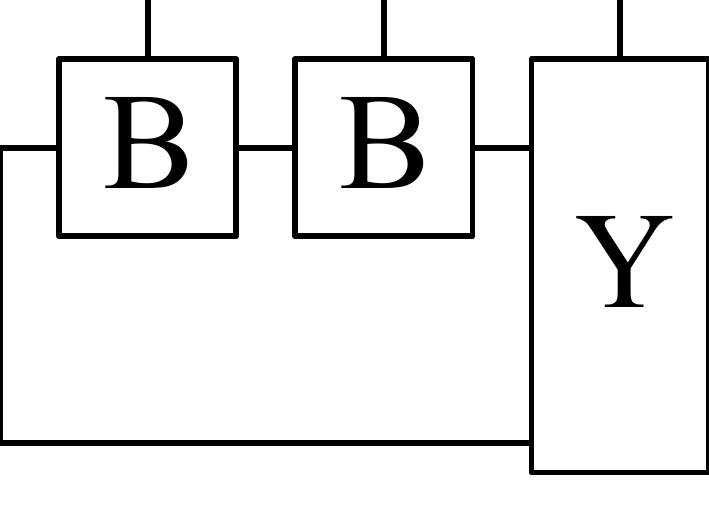}\ +\ 
	 \figgg{0.15} {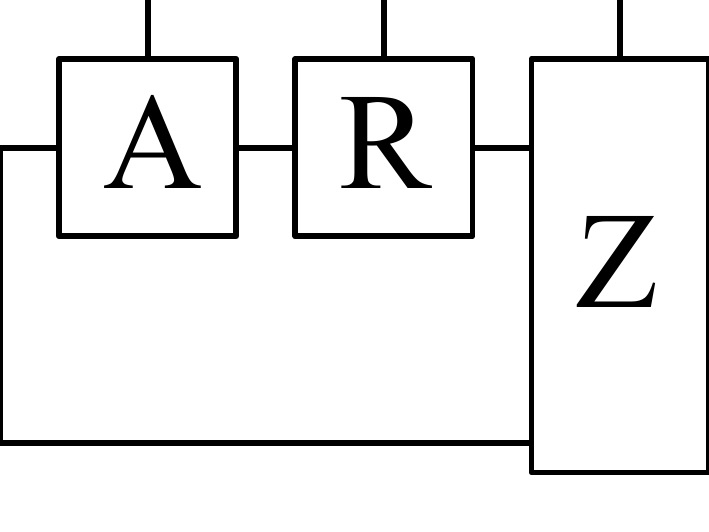}\ + \cdots\\[1ex]
&\hspace*{1.9cm}\dots+\ \figgg{0.15} {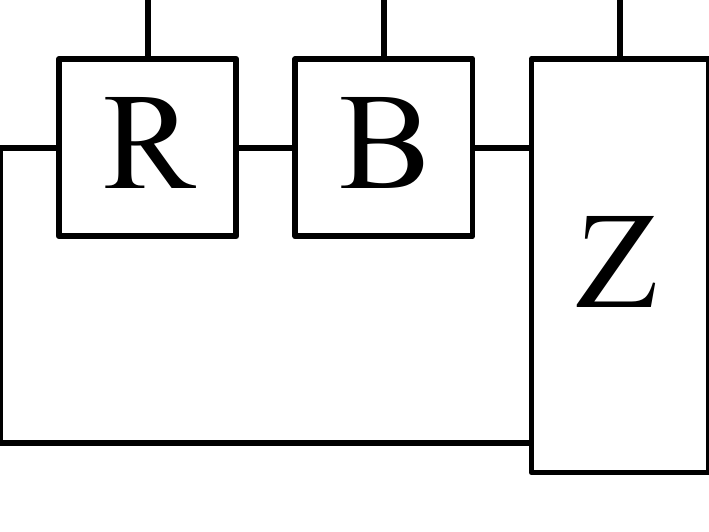} 
	+ \figgg{0.15}{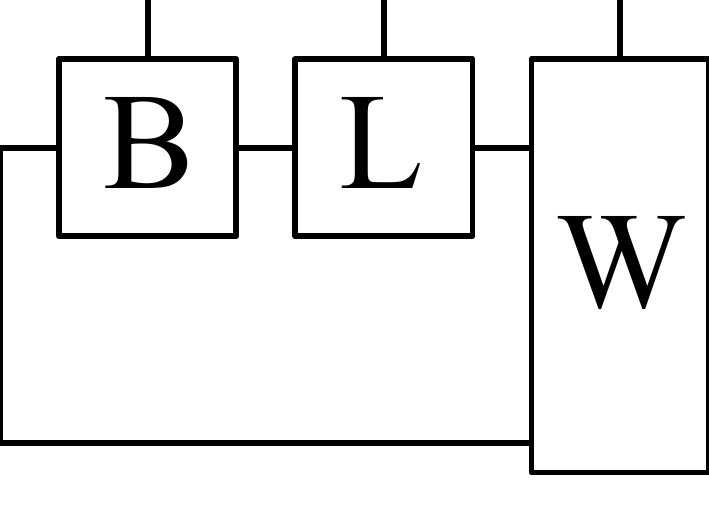} + \figgg{0.15}{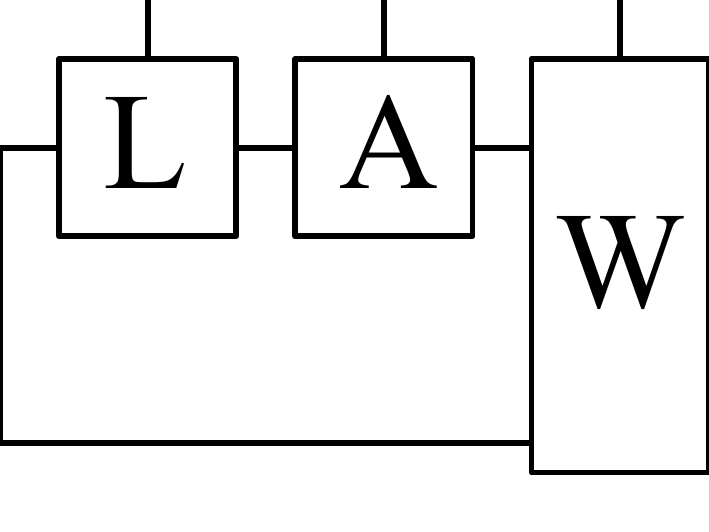}\\[1ex]
&= \figgg{0.15}{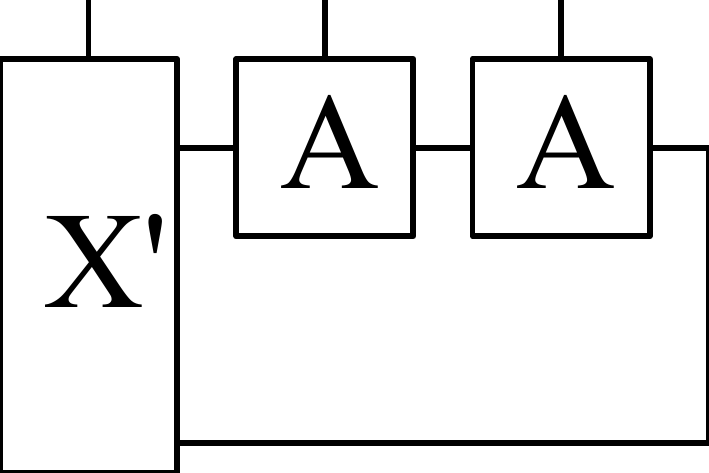}\ +\ \figgg{0.15}{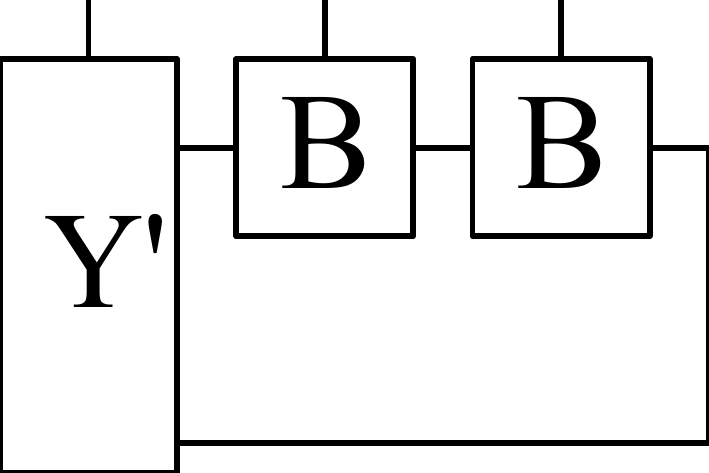}\ + \
	 \figgg{0.15} {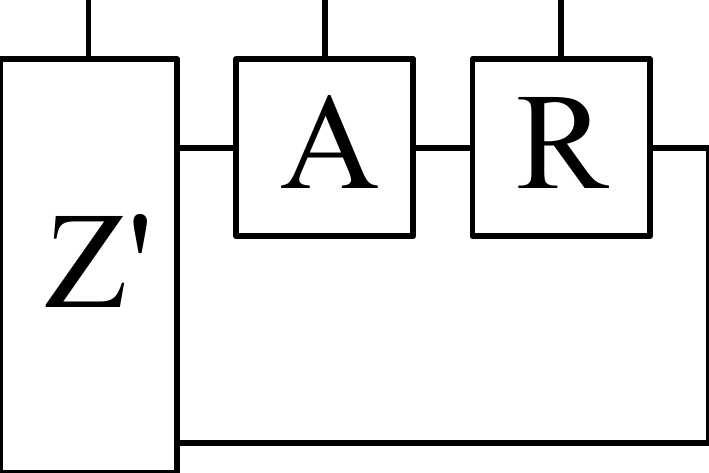}\ +\cdots\\[1ex]
&\hspace*{1.9cm}\dots+\ \figgg{0.15} {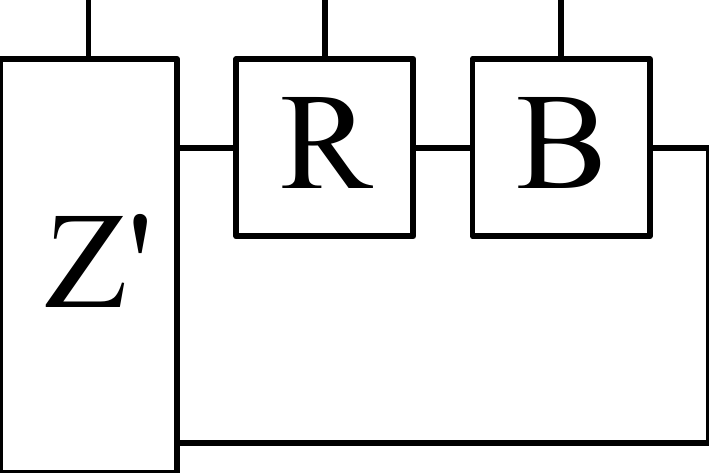} 
	\ +\ \figgg{0.15}{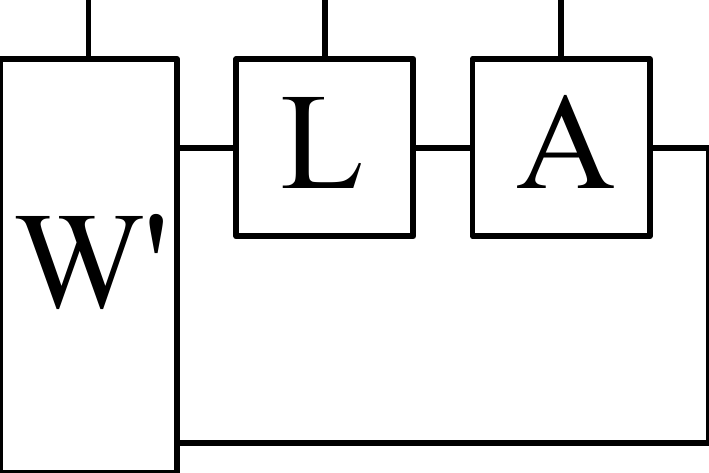}\ +\ \figgg{0.15}{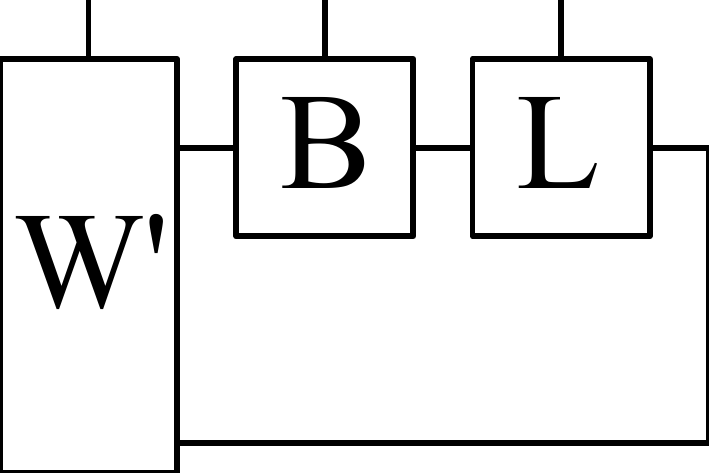}\ . 
\end{aligned}
\end{equation}
We want to show that $\ket{\varphi}$ is of the form (\ref{eq:vector}),
i.e., we need to show that $X$ has a decomposition such as in
(\ref{eq:Xprime-decomposition}), and so on. To this end, we will make
heavy use of Lemma~\ref{lemma:inj-consequences}. In particular,
injectivity of $\left(\begin{smallmatrix}A&R\\L&B\end{smallmatrix}\right)$
implies the existence of a tensor $R^{-1}$ left-inverse to $R$, which at
the same time annihilates any of the other tensors $A$, $B$, and $L$, as
well as the existence of a vector $\ket{b}$ satisfying condition 5 in Lemma \ref{lemma:inj-consequences} for $\figgg{0.15}{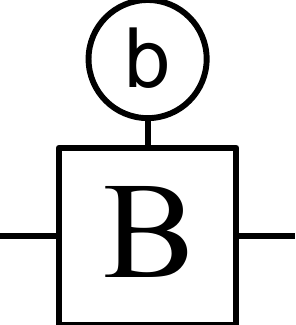}=\figgg{0.15}{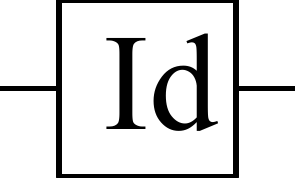}$. We now apply $R^{-1}$ to the second
site in Eq.~(\ref{eq:appA:vectorinintersection}), which cancels all terms
except one at each side of (\ref{eq:appA:vectorinintersection}).
Additionally, we project the third site onto $\ket{b}$ and obtain
\begin{equation}\label{eq2}
\figgg{0.15}{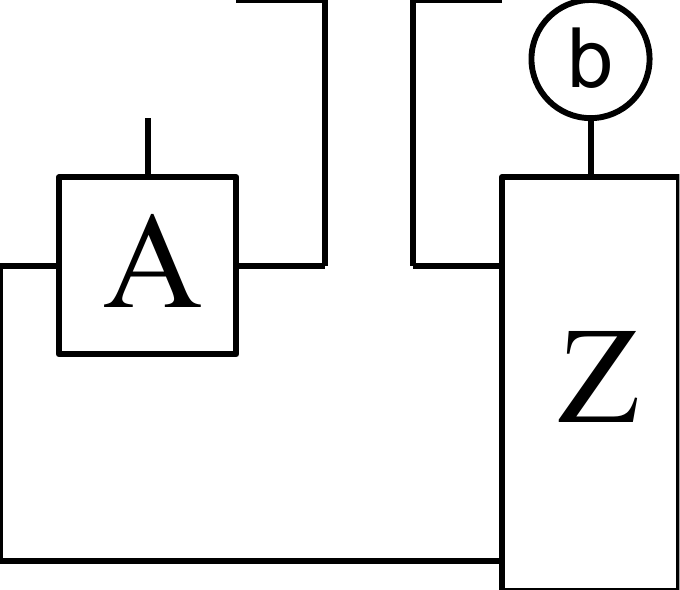}=\figgg{0.15}{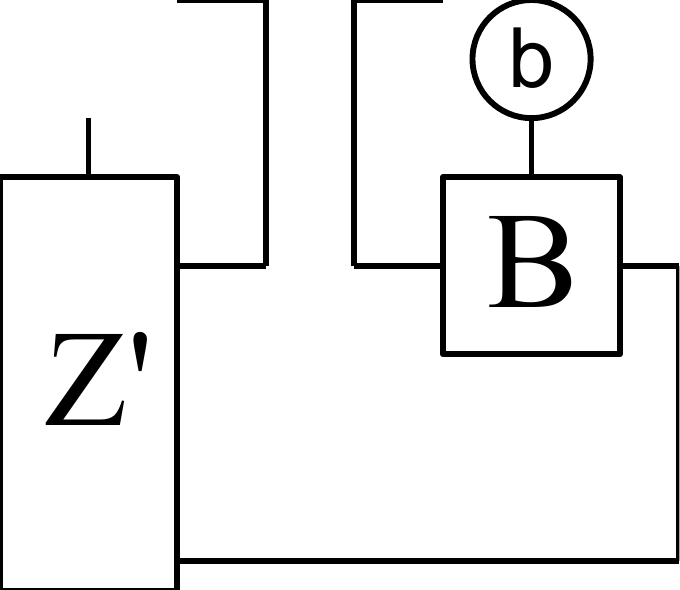}=\figgg{0.15}{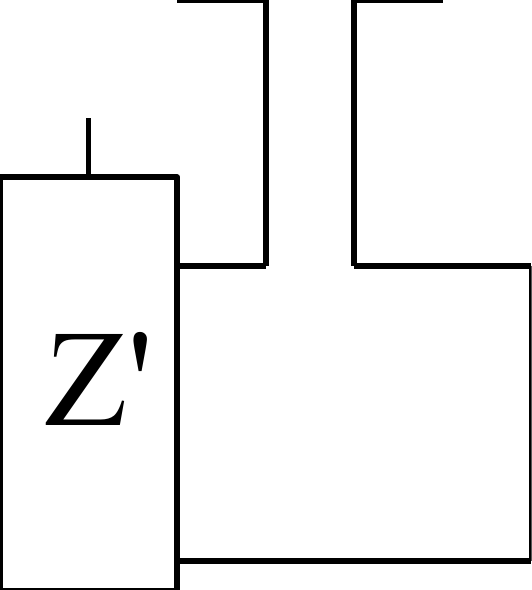}
\end{equation}
By defining $\figgg{0.15}{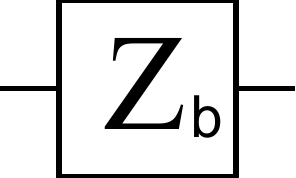}=\figgg{0.15}{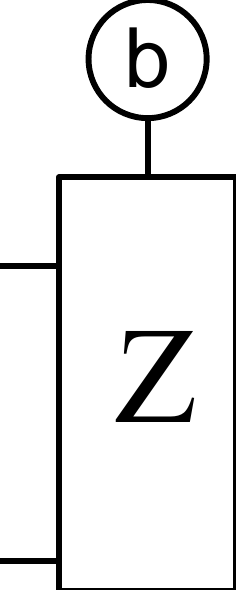}$\ , we therefore obtain
that 
\begin{equation}\label{eq:appA:Zprime-A-Zb}
\figgg{0.15}{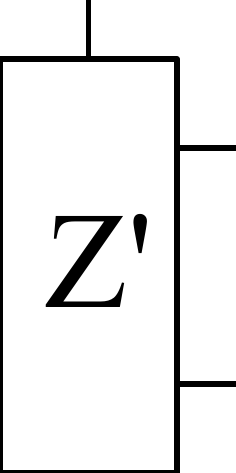}=\figgg{0.15}{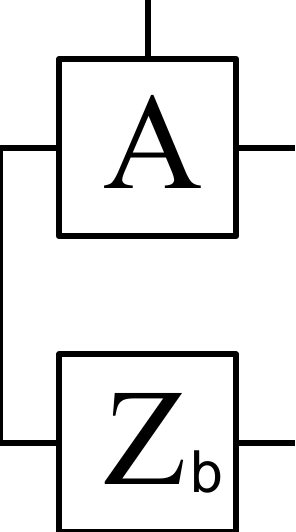}\ .
\end{equation}
Similarly, we can apply $A^{-1}$ at the second site and a 
vector $\ket{a}$ verifying condition 5 in Lemma \ref{lemma:inj-consequences} such that $\figgg{0.15}{Aa}=\figgg{0.15}{Id}$ 
at the first site to see that 
\begin{equation}\label{eq:appA:Z-B-Zaprime}
\figgg{0.15}{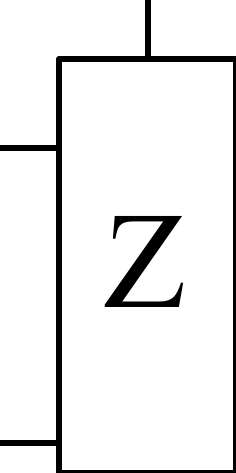}=
\figgg{0.15}{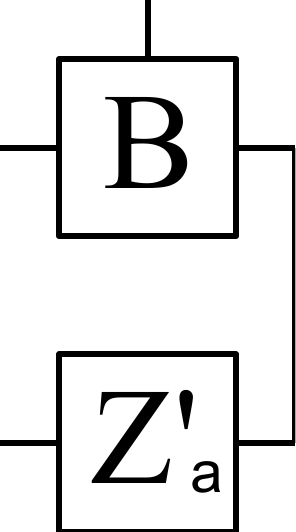}\ ,
\end{equation}
where $Z'_a$ is defined correspondingly.  Finally, we can project
Eq.~(\ref{eq:appA:Z-B-Zaprime} onto $\ket{b}$ to infer that 
\begin{equation}\label{eq:appA:Zb-eq-Zaprime}
\figgg{0.15}{Z_b}=\figgg{0.15}{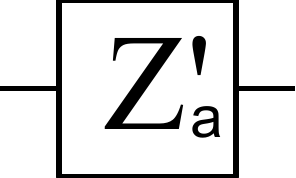}\ .
\end{equation}
Corresponding expressions for the form of $W$ and $W'$ can be derived using
$L^{-1}$.

% This shows that the summands
% \[
% \figgg{0.15} {ARbigZ} \ +\ \figgg{0.15} {RBbigZ}
% \]
% can be decomposed as 
% \[
% \figgg{0.15}{ARBZ_b} \ +\ \figgg{0.15}{RBBZ_b}\ ;
% \]
% in order to obtain a vector in $R_3$, we must
% recover the lacking summand 
% \[
% \figgg{0.15}{AARZ_b}\ .
% \]

Now let us return to the identity (\ref{eq:appA:vectorinintersection}) and
apply $A^{-1}$ to the second site, which yields
\begin{equation*}
\figgg{0.15}{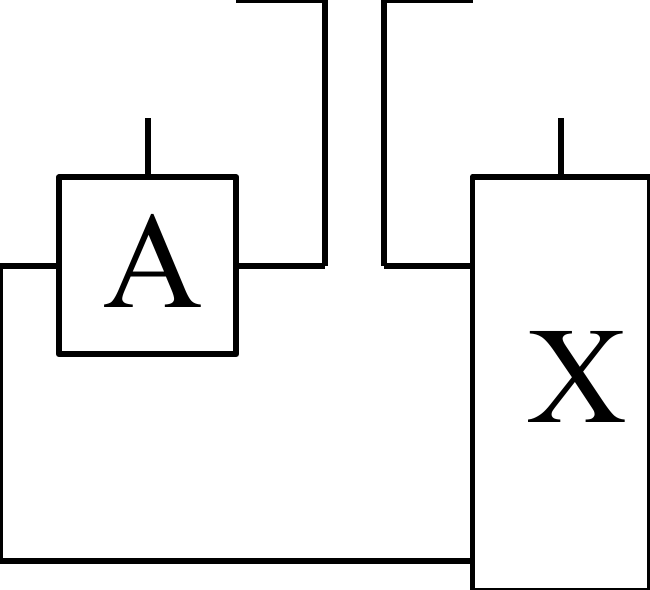} \ +\ \figgg{0.15}{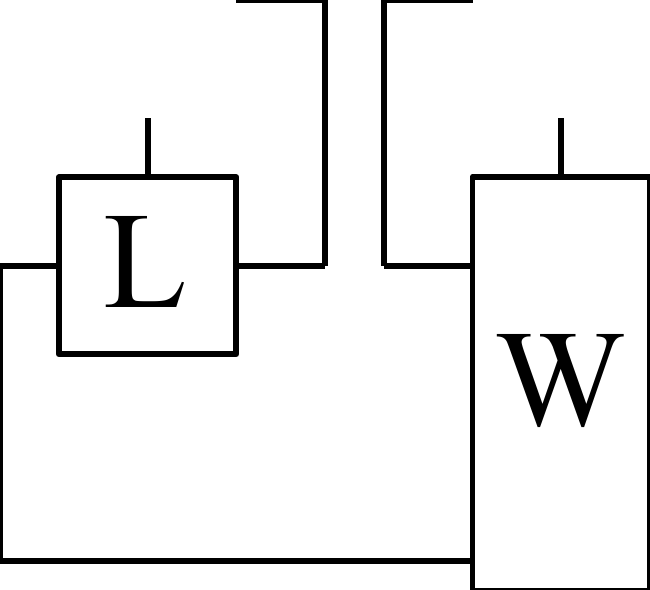} 
\ =\  
\figgg{0.15}{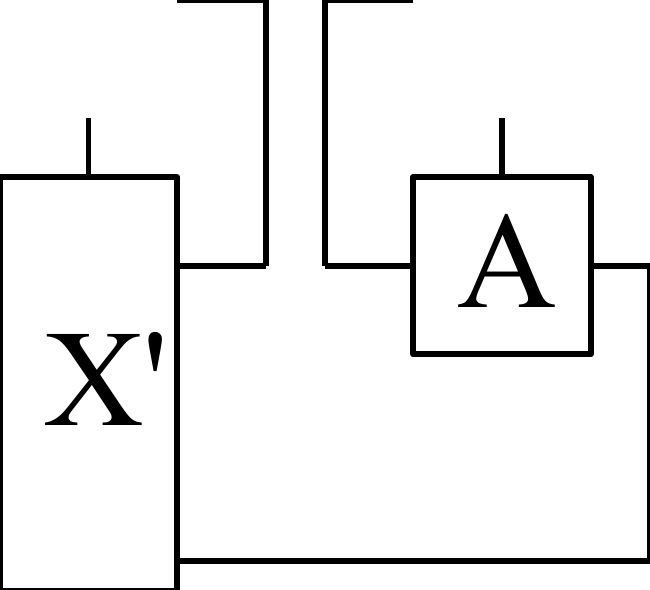}\ +\ \figgg{0.15}{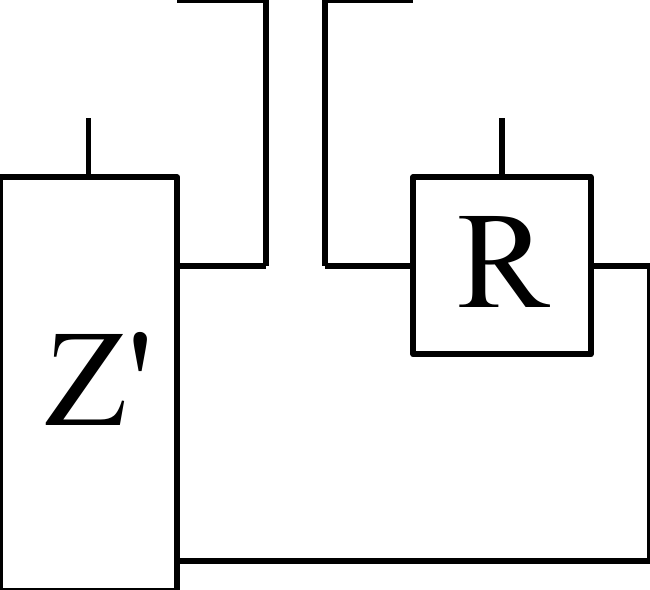}\ .
\end{equation*}
Using the analogue of
Eqs.~(\ref{eq:appA:Zprime-A-Zb}--\ref{eq:appA:Zb-eq-Zaprime}) for $W$,
this is equivalent to
\begin{equation*}
 \figgg{0.15}{AibigX}\ +\ \figgg{0.15}{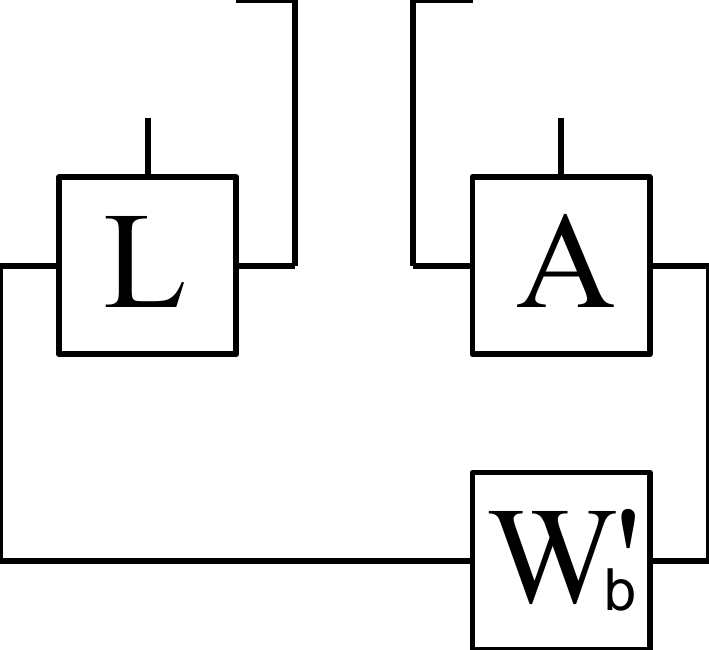}
\ =\ 
\figgg{0.15}{bigXpiA}\ +\ \figgg{0.15}{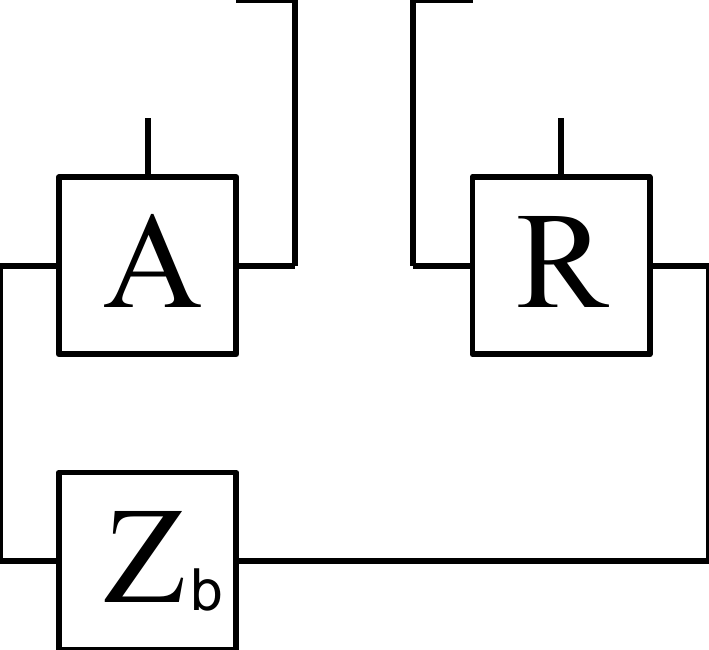}\ .
\end{equation*}
Now, we apply $\ket{a}$ to the first site and obtain 
\begin{equation}\label{eq:appA:X-decomposition}
\figgg{0.15}{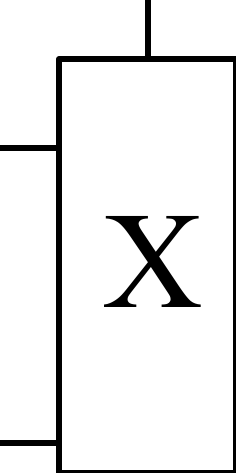}
\ =\ 
\figgg{0.15}{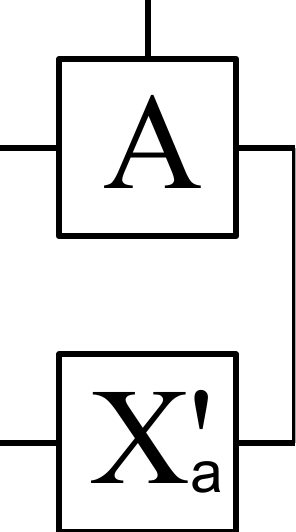}\ 
%-\ \figgg{0.15}{AWp_bL_a}
\ +\ \figgg{0.15}{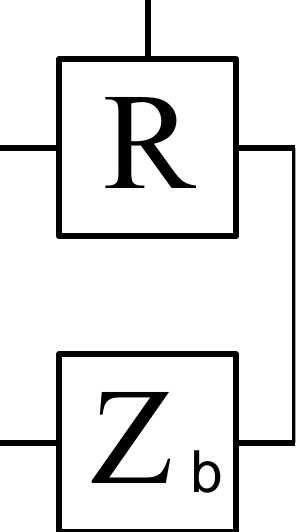}\ ,
\end{equation}
where we have defined $X_a'$ 
%and $L_a$ 
accordingly.

Combining Eqs.~(\ref{eq:appA:Z-B-Zaprime}), (\ref{eq:appA:Zb-eq-Zaprime}),
and (\ref{eq:appA:X-decomposition}), we find that
\begin{align*}
&\figgg{0.15}{AAbigX}\ +\ \figgg{0.15} {ARbigZ}\ +\ \figgg{0.15} {RBbigZ}=
\\
&=\figgg{0.15}{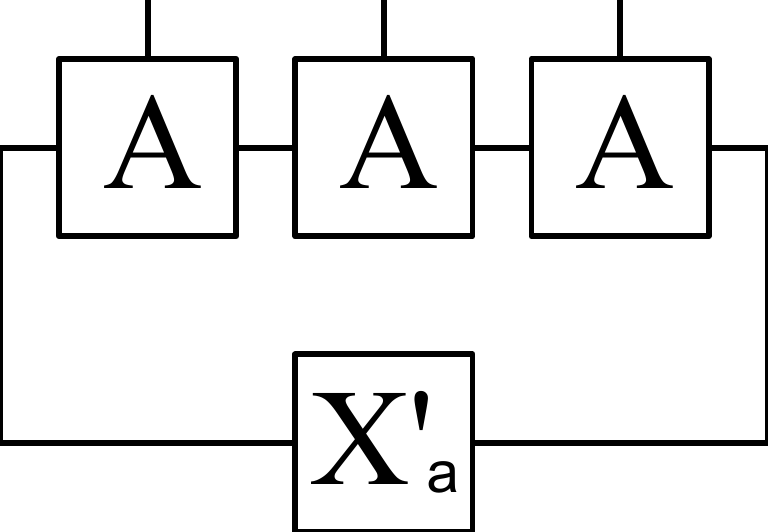}\ 
%-\ \figgg{0.15}{AAAWp_bK_a}
\ +\ \figgg{0.15}{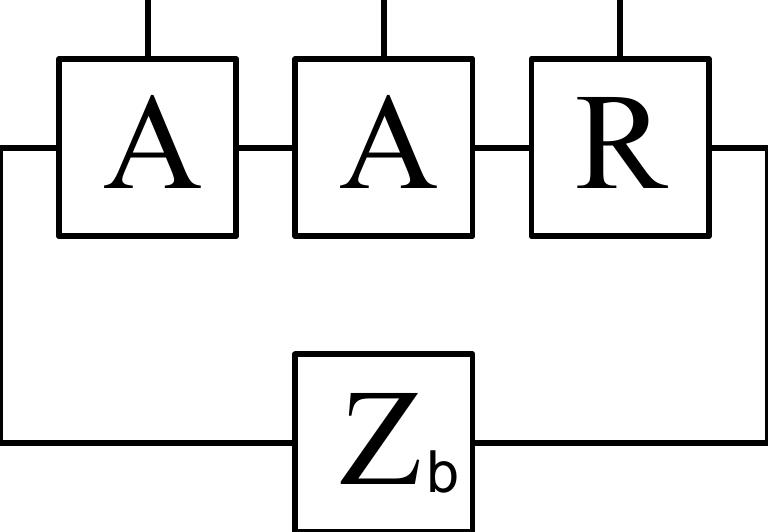}\ +\ 
\sum_{\mathrm{pos\, R}\neq 3}\figgg{0.15} {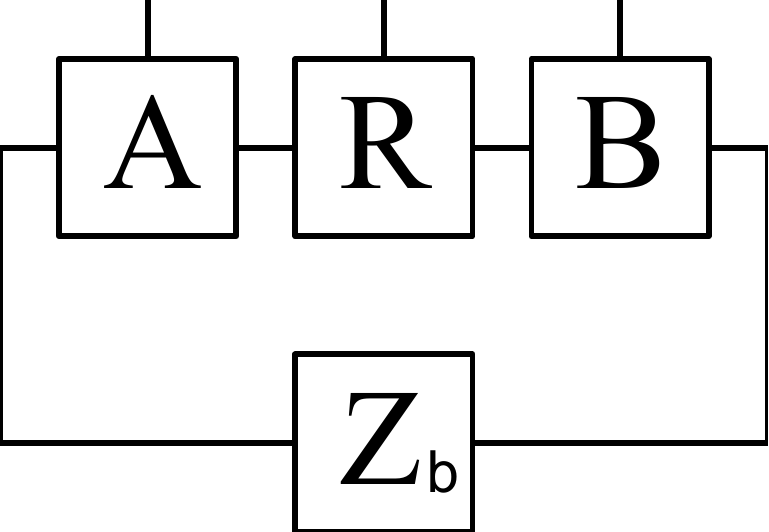}\ ,
\end{align*}
which shows that the l.h.s., which is half of the terms in
Eq.~(\ref{eq:vector}), is contained in $A_3+R_3$. In the same way it can
be shown that the sum of other three terms in Eq.~(\ref{eq:vector}) is
contained in $B_3+L_3$, which proves that $S_2\otimes \C ^d \cap
\C^d\otimes S_2 \subset S_3$.

The proof that $S_k\otimes \C ^d \cap \C^d\otimes S_k \subset S_{k+1}$ can
be carried out in the same fashion, using that the tensors 
\begin{align*}
&\figgg{0.15}{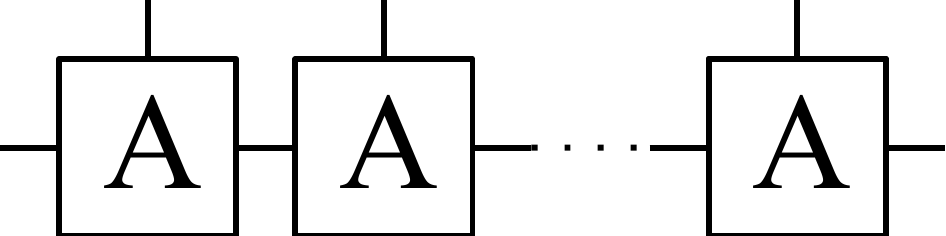}\ ,\quad \figgg{0.15}{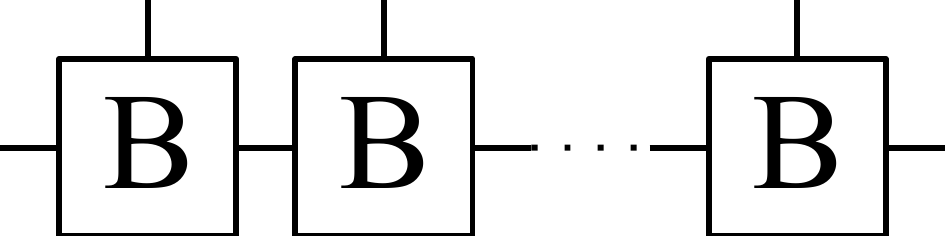}\ , \\[1ex]
&\sum_{\mathrm{pos\,R}}\figgg{0.15}{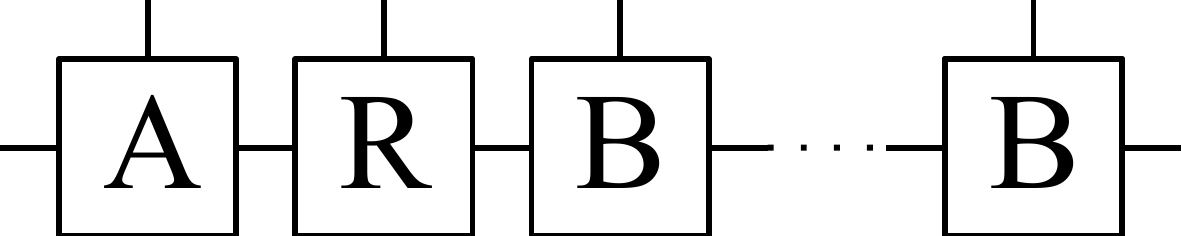} \ ,
\mbox{ and } \ \sum_{\mathrm{pos\,L}}\figgg{0.15}{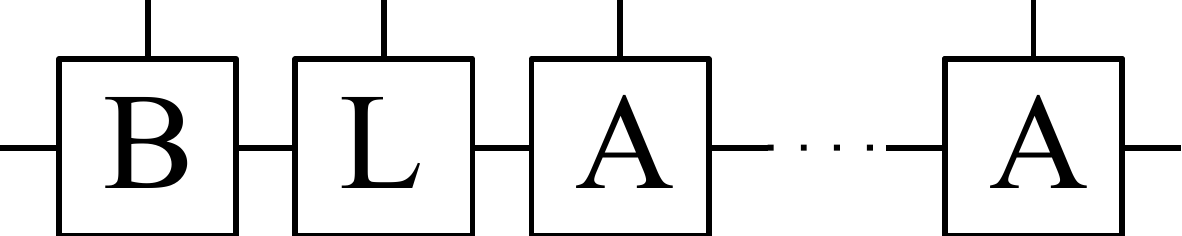}
\end{align*}
have inverses, since  
injectivity of $\left(\begin{smallmatrix} A & R \\
L & B \end{smallmatrix} \right)$ implies injectivity of
\begin{equation}
\left( \begin{array}{cc}
A\cont A\cont \cdots \cont A 
& 
\ \sum_{\mathrm{pos\,R}} A\cont R\cont B\cont  \cdots \cont B \\
\sum_{\mathrm{pos\,L}} B\cont L\cont A\cont  \cdots \cont A\ 
& 
B\cont B\cont  \cdots \cont B
\end{array} \right)\ ;
\end{equation}
this can be proven in analogy to Lemma~\ref{lemma:inj-perturbation}.
\hspace*{\fill}\qed
\end{proof}

\begin{proof}[Closure property, Lemma~\ref{lemma:closure}]
It is clear that 
\[S_\mathrm{left}\cap S_\mathrm{right}\supset \spanned\left\{\figgg{0.15}{AA_A2},\figgg{0.15}{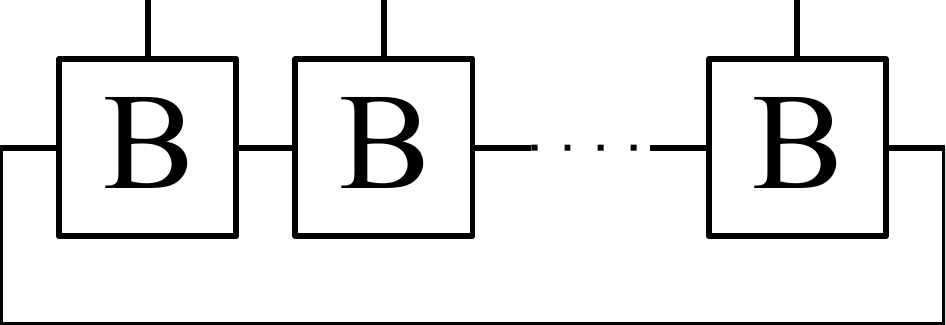}\right\}.
\]
To
show the converse, let $\ket{\varphi}\in S_\mathrm{left}\cap
S_\mathrm{right}$:
\begin{align*}
\ket{\varphi}\ =& \quad 
    \figgg{0.15}{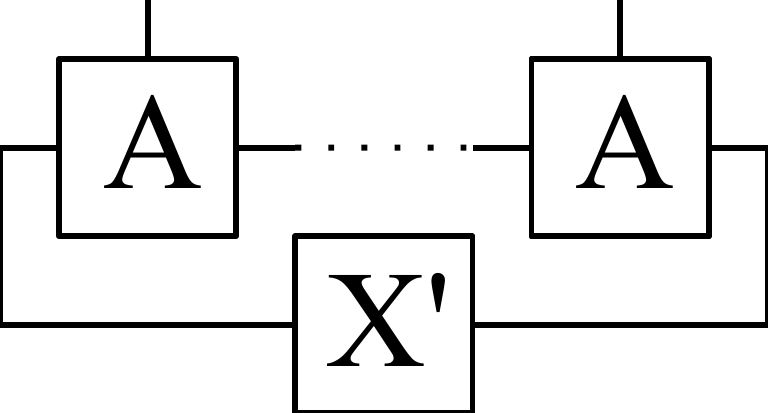}\ +\ 
    \figgg{0.15}{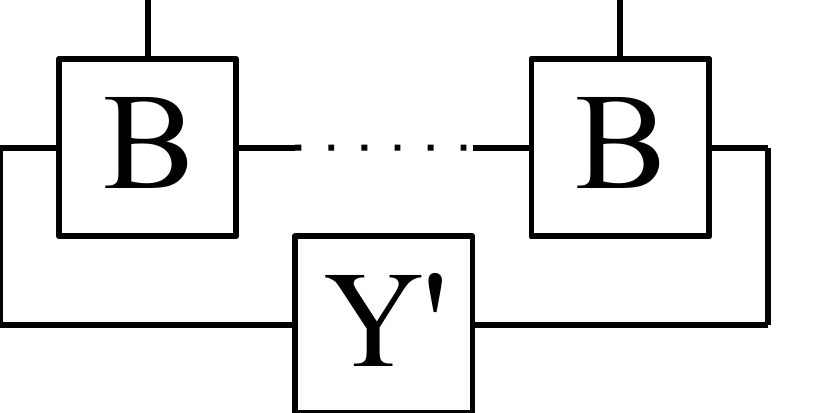}\  +\  
\\[1ex]
&\qquad
    +\sum_{\mathrm{pos R}} \figgg{0.15}{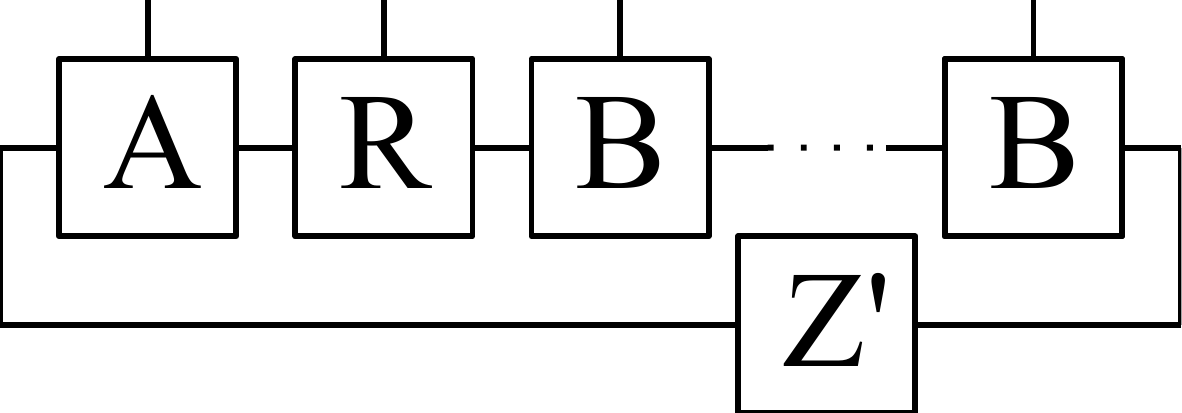}
\ +\ \sum_{\mathrm{pos L}}\figgg{0.15}{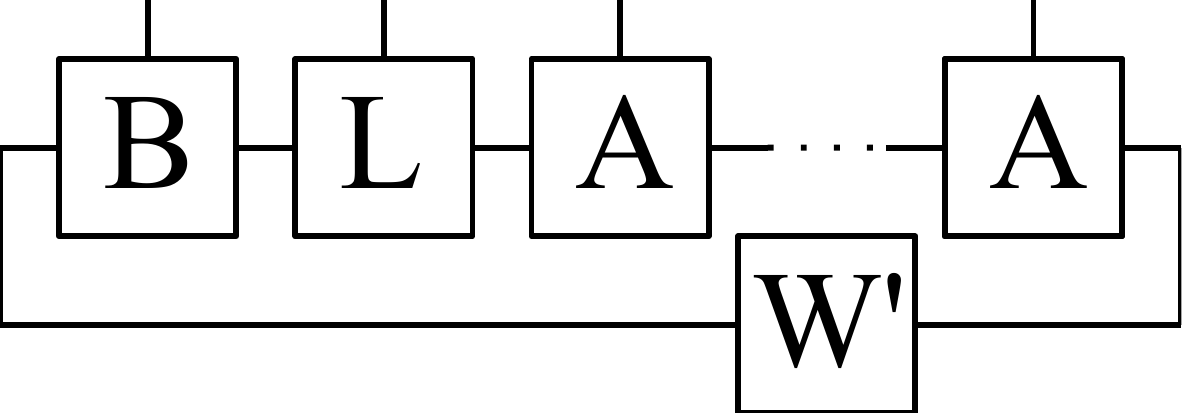}
\\[1ex]
=& \quad  \figgg{0.15}{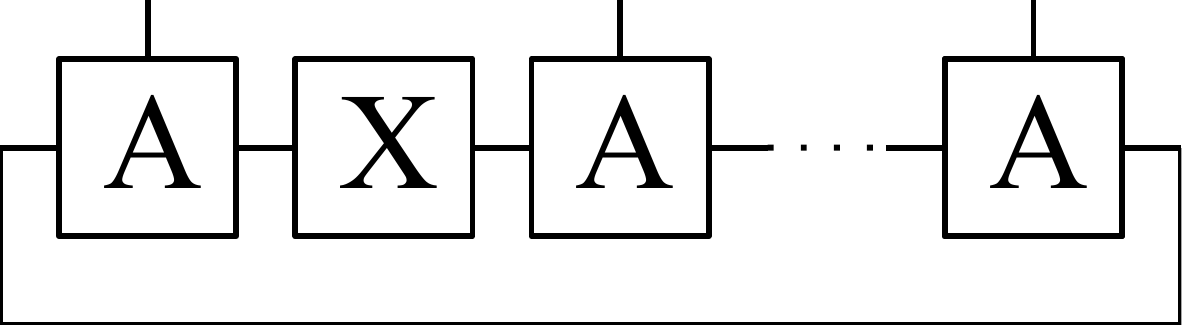}\ +\ \figgg{0.15}{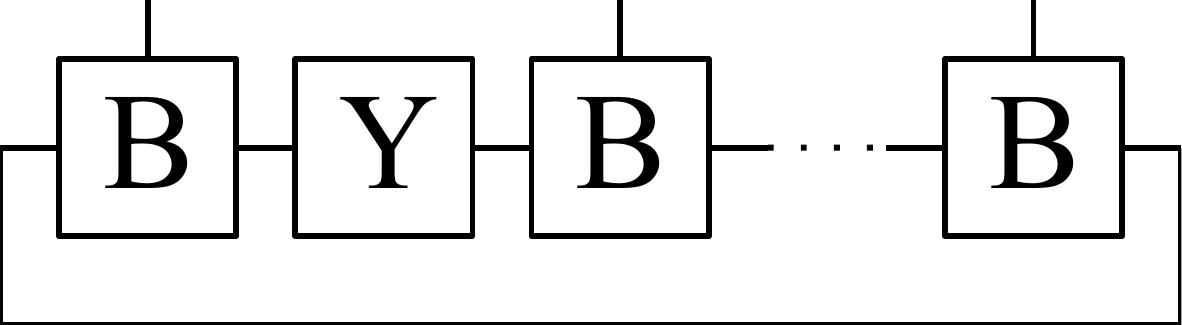}\  +\\[1ex]
&\qquad  + \sum_{\mathrm{pos\,R}} \figgg{0.15}{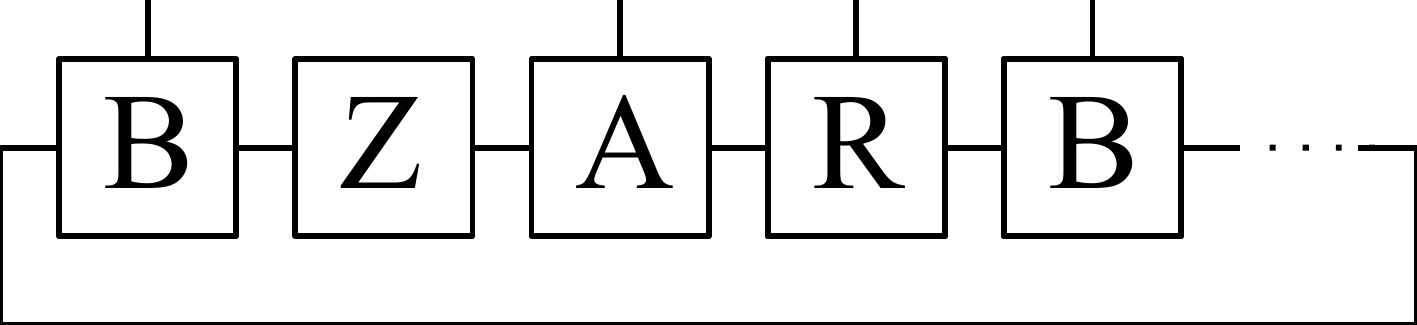}
\ +\ \sum_{\mathrm{pos\,L}}\figgg{0.15}{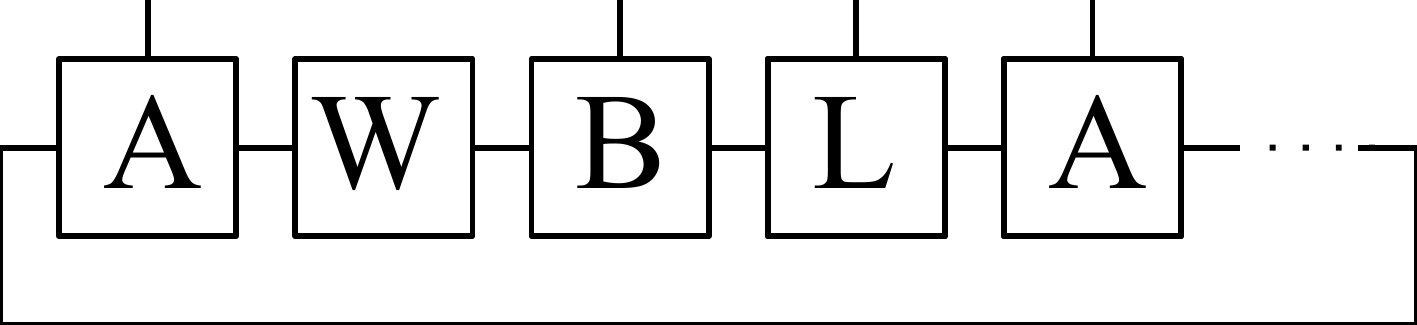}\ .
\end{align*}
By applying the inverse tensor corresponding to $\sum\limits_{\mathrm{pos\,
R}}\figgg{0.15}{ARB_B}$ at sites $2,\dots,N$, we get
\begin{equation*}
\figgg{0.15}{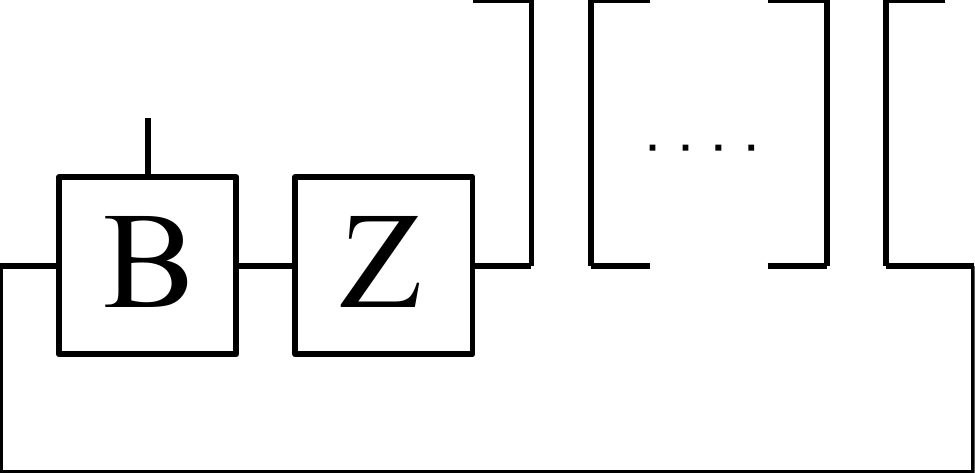}=\figgg{0.15}{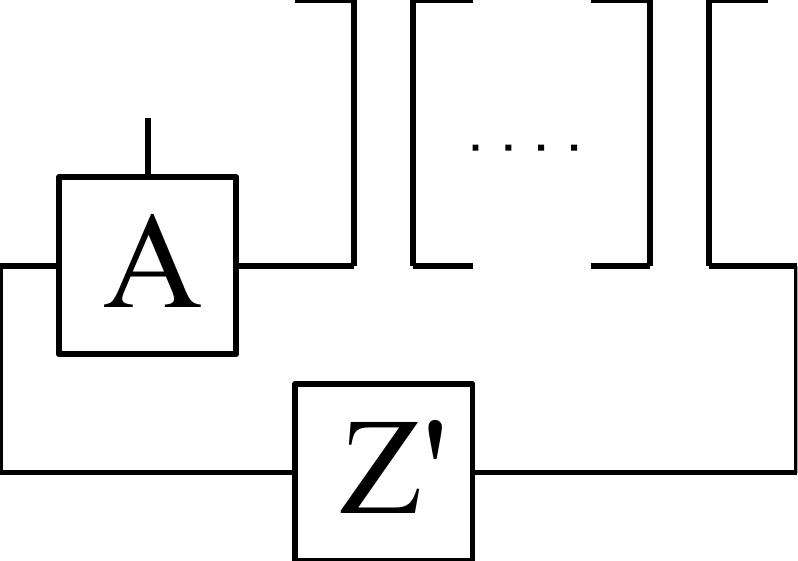}\ .
\end{equation*}
Now let $\ket a$ be such that 
$\raisebox{0.5em}{\figgg{0.15}{Aa}}=\figgg{0.15}{Id}$ and $\raisebox{0.5em}{\figgg{0.15}{Ba}}=0$
(Lemma~\ref{lemma:inj-consequences}): Projecting the first site onto $\ket
a$ yields $\figgg{0.15}{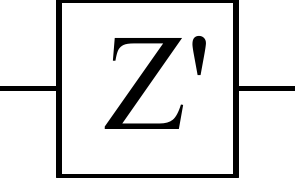}=0$; with a corresponding $\ket b$, we find
that $\figgg{0.15}{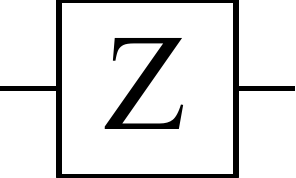}=0$. In the same way, we can prove that
$\figgg{0.15}{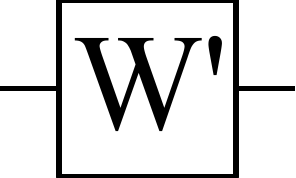}=\figgg{0.15}{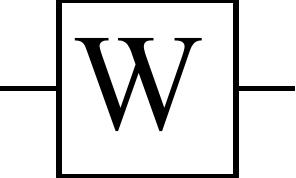}=0$.

Now, we can apply $A^{-1}$ to all sites to find that
$\figgg{0.15}{X}=\figgg{0.15}{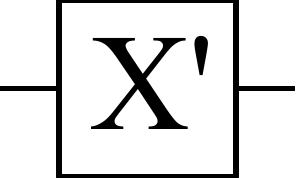}=\figgg{0.15}{Id}$, and $B^{-1}$ to
obtain $\figgg{0.15}{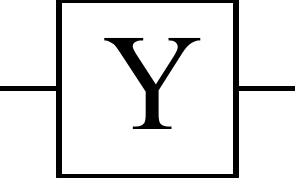}=\figgg{0.15}{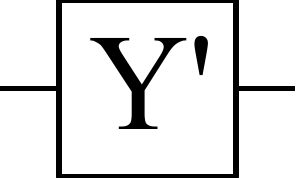}=\figgg{0.15}{Id}$ (a similar
proof can be found in \cite{schuch:peps-sym}), showing that 
$S_\mathrm{left}\cap S_\mathrm{right}\subset \spanned\left\{\ket{M(A)},\ket{M(B)}\right\}$.
\hspace*{\fill}\qed
\end{proof}

\section{The thermodynamic limit Hamiltonian via the GNS-representation}
\label{appendix:GNS}

The way to study the thermodynamic limit of finitely correlated spin chains is through the GNS-representation of the algebra of local observables \cite{nachtergaele:quantumspinsystems,bratteli:operatoralgebras}.

\begin{theorem}[Gelfand-Neimark-Segal Representation] \cite{conway}
Given a $\displaystyle C^*$-algebra\, $\mc  U$ with identity and a state $\w$ on it, there exists a essentialy unique, up to unitary equivalence, cyclic\footnote{A representation $(\mc H,\pi,e)$ of a $\displaystyle C^*$-algebra $\mc U$ is called cyclic if $\pi(U)e$ is dense in $\mc H$. Such vector $e$ is also called cyclic.} representation $(\mc H_\w,\pi_\w,\Omega_\w)$ such that $\w(A)=\braket{\Omega_\w}{\pi_\w(A)\Omega_\w}$ for all $A\in\mc U$. Consequently, $\|\Omega_\w\|^2=\|\w\|=1$. 
\end{theorem}

In order to construct it, a Hilbert space structure must be introduced via the state $\w$: $\braket{A}{B}=\w(A^*B)$ is the (possibly singular) inner product. The quotient of $\mathcal{A}$ by the subspace of elements such that $\braket{A}{A}=0$ is a pre-Hilbert space, which just needs to be completed to be the Hilbert space $\mathcal{H}_\w$ needed for the representation.

The class to which the identity in $\mathcal{A}$ belongs can be taken as the distinguised vector $\Omega_\w$.

The algebra we are dealing with is the algebra of local observables over an infinite spin chain:
\begin{equation}
\mathcal{A}=\bigcup _{i<j} \cdots \otimes \I \otimes \I \otimes \mathcal{A}_i \otimes \mathcal{A}_{i+1} \otimes \cdots \otimes \mathcal{A}_j \otimes \I \otimes \cdots,
\end{equation}

\noindent  where each $\mathcal{A}_k$ denotes the local algebra of observables at the respective site $k$. Since the dimension at each site is always the same in translationally invariant MPSs, this local algebra is the same for every site.

Let $\w_A$ be the state densely defined on local observables $O\in \mc A$ as 
$$\w_A(O)=\figgg{0.15}{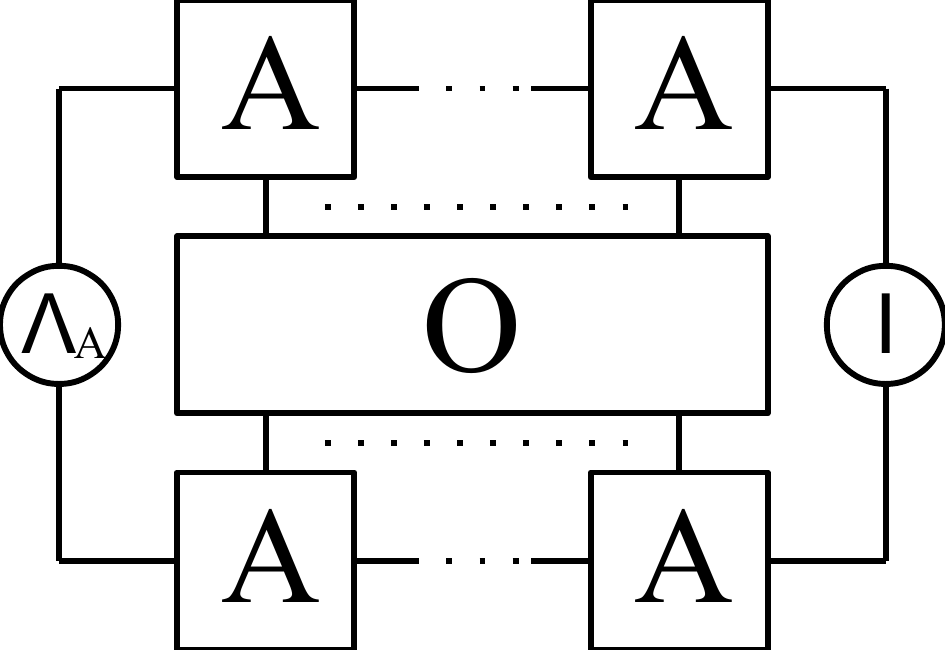},$$
and extended continuously to $\bar{\mc A}$.

This state can be considered as the limit of the unnormalized states $\ket{M(A)}=\figgg{0.18}{AA_A2}\ $,  since for any local observable $O$ we have $\w_A(O)$ as the limit
\begin{equation*}
\figgg{0.18}{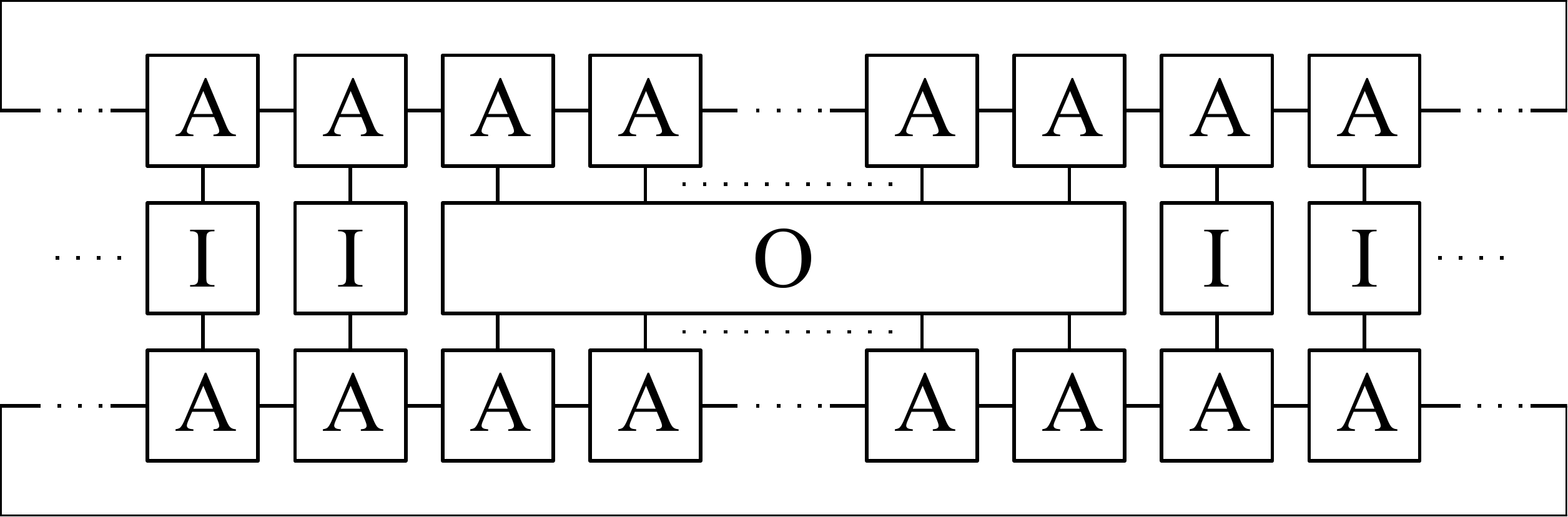}\to \figgg{0.18}{Limitobs2new}\ ,
\end{equation*}
\noindent because the many copies of the operator $\figgg{0.18}{AA2}$ tend to $\figgg{0.18}{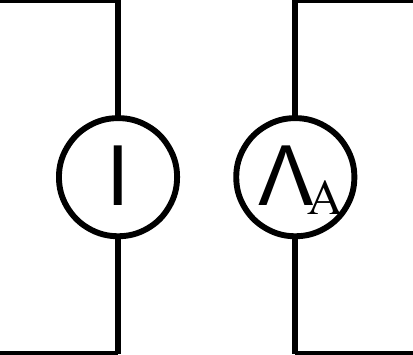}$, and also $\trace(\Lambda_A)=1$ -- which means $\w_A(\I)=1$, and therefore $\w_A$ is normalized.

Thus, it is natural to describe $\w_A$ as $$\w_A=\figgg{0.18}{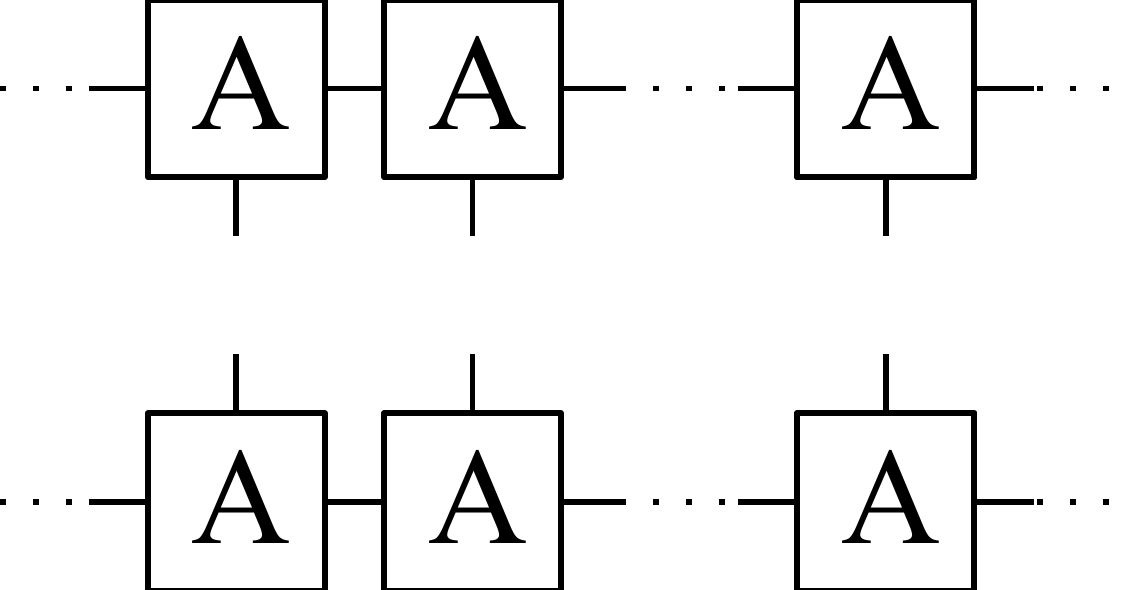}\ .$$ Similarly, we can consider the analogue state $\w_B$ for the $B$ block.

In order to study the spectral properties of the thermodynamic limit of the uncle Hamiltonian, we can take $\w_A$ as the state to which the representation is associated. The choice of either $\w_A$ or $\w_B$ is irrelevant in this case, since they play a similar role in the MPS. However, for general Hamiltonians, the spectrum may depend on the ground state taken for the representation.

In the first place we have to construct the quotient of the algebra $\mathcal{A}$ of local observables by the ideal of those observables $O$ such that $\w_A(O^*O)=0$:
\begin{equation*}
\figgg{0.18}{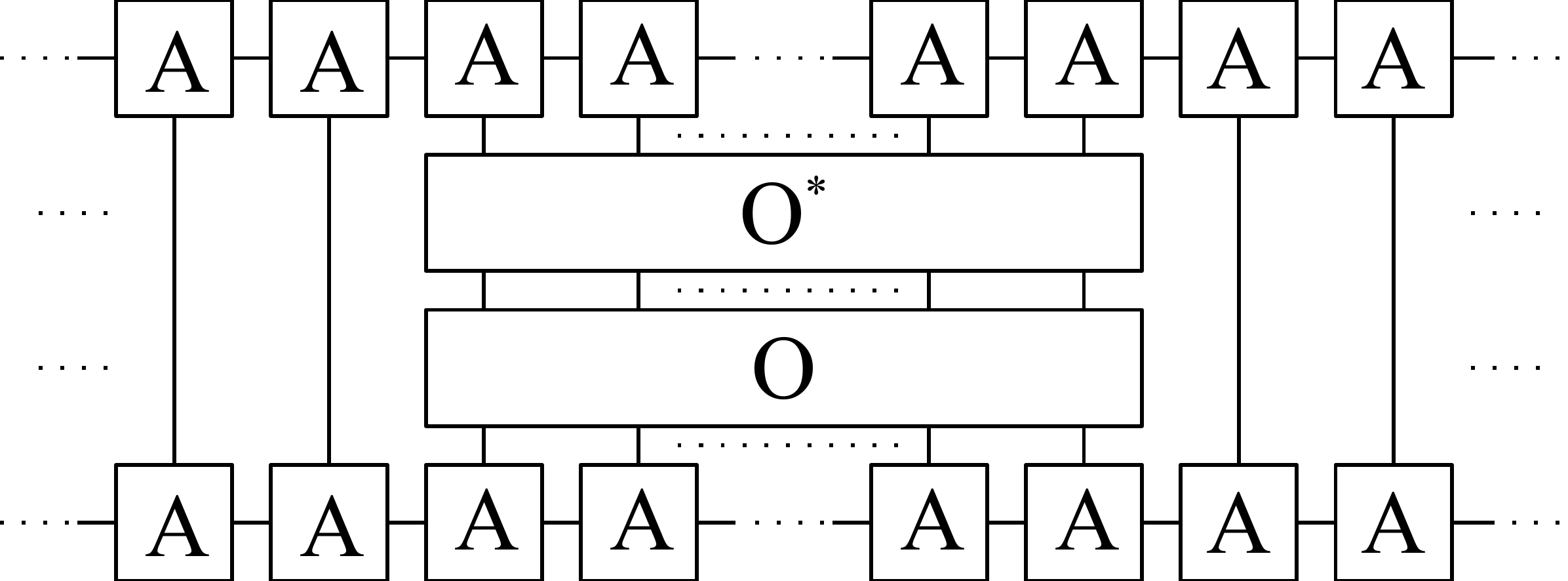}=0.
\end{equation*}

We can see this as the norm of the 'vector' $\figgg{0.18}{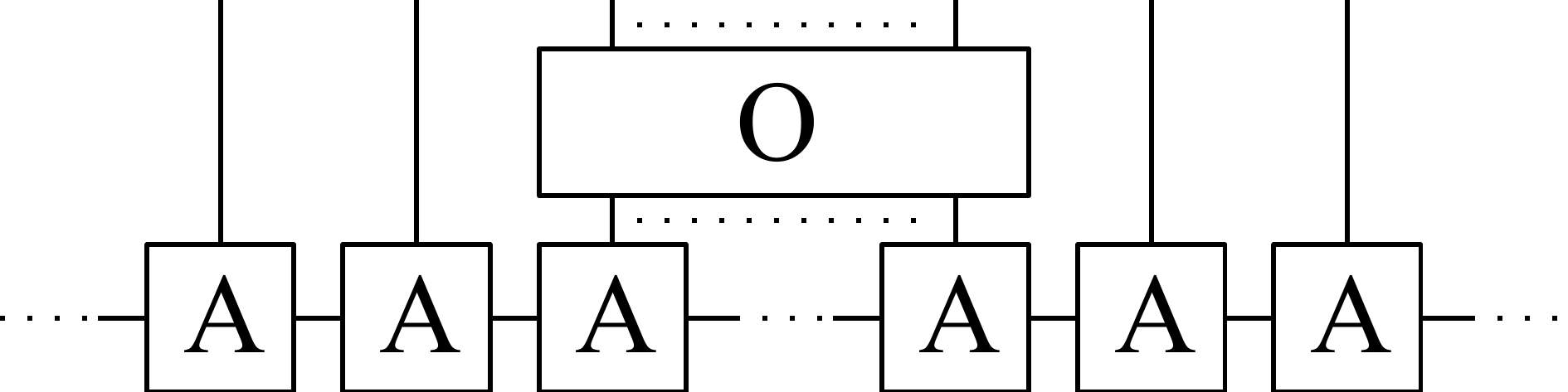}$, and the representatives of the equivalence classes in the quotient will have this form. Since any operator $O$ can be considered, the Hilbert space $\mc H_\w$ of the representation will be seen as the completion of
\begin{equation*}
S=\bigcup_{i\le j} S_{i,j}, \hbox{ where } S_{i,j}={\rm span}\left\{\figgg{0.18}{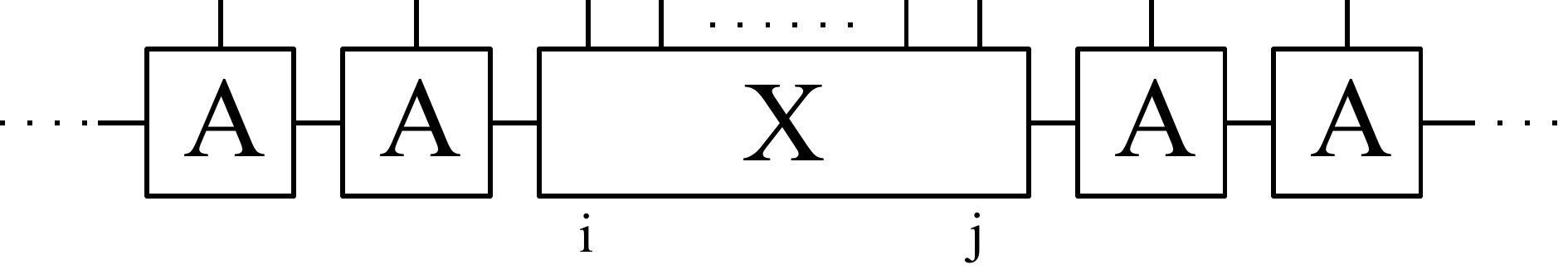},\ X\right\},
\end{equation*}
where different tensors $X$ can lead to the same state (tensors $X$ come from the contraction of $O$ observables with those $A$ tensors in the sites $O$ is acting on), and they just need to match the needed dimensions.

The local observables will be represented in $B(\mathcal{H}_\w)$ by themselves tensored with the identity, acting on $\bar{S}=\mathcal{H}_\w$.

Since, given any element of $S$, only finitely many local Hamiltonians $h'_P$ do not anihilate on it, the global uncle Hamiltonian is well defined on $S$, and therefore densely defined on $\bar{S}=\mc H_\w$. In the case  $H'_P|_S$ is essentially self-adjoint and has a unique self-adjoint extension, there is no other choice for this extension but to be the thermodynamic limit $H'_\w$ we must study.

For $H'_P|_S$ to have a unique self-adjoint extension, it suffices to have a dense set of analytic elements in its domain \cite{reedsimon}. A vector $\ket{\varphi}$ is analytic if there exists some $r>0$ such that
\[
\sum_{n=0}^\infty \frac{r^n}{n!}\|H_P'^n(\ket{\varphi})\|<\infty.
\]

If $\ket{\varphi}$ belongs to $S$ it must be in some $S_{-M,M}$. Recall that, for every $N$, we have that $\|H'_P|_{S_{-N,N}}\|\le 2N+2$ and $H_P'(S_{-N,N})\subseteq S_{-N-1,N+1}$, and consequently $\|H'^n(\ket{\varphi})\|\le \Pi_{k=1}^n(2M+2k)\|\ket{\varphi}\|$. From this we have
\[
\frac{r^n}{n!}\|H'^n(\ket{\varphi})\|\le{r^n}\Pi_{k=1}^n\frac{2M+2k}{k}\le r^n \Pi_{k=1}^n (2M+2)=r^n(2M+2)^n,
\]
which is summable for $r<1/(2M+2)$. 

Thus every vector in $S$ is an analytic vector for $H'_P|_S$, which is essentially self-adjoint, and therefore the thermodynamic limit Hamiltonian $H'_\w$ of the uncle Hamiltonian is the unique self-adjoint extension of $H'_P|_S$, whose spectrum we must study.

\subsection*{Proof of Proposition \ref{summingeigenvalues}}

In order to proof Proposition \ref{summingeigenvalues}, we first need the next simple fact, which we already commented in the main text.

\begin{proposition} \label{prop:finitelysupportedWeylseqs} A real value $\lambda\in \spec(H'_\w)$ iff there exists a sequence of normalized states   $\{\ket{\varphi_{\lambda,k} } \} _k\in S$ such that $\|(H'_\w -\lambda I ) (\ket{\varphi_{\lambda,k}}) \|\to 0$.

\end{proposition}

\begin{proof}
This follows from the fact that the residual spectrum of $H'_\w$ is empty --since it is self-adjoint--, the definition of a value lying in the point or continuous spectrum (Appendix~\ref{appendix:spectraltheorem}), and the fact that $H'_\w$ is the closure of $H'_P$ acting on $S$, that is, $\graph(H'_\w)=\overline{\graph(H'{|_S})}\subset \bar S\times\bar S$.
\hspace*{\fill}\qed
\end{proof}

\begin{proof} {\it (Proposition \ref{summingeigenvalues})}

We will state the proof for two values $a=\lambda_1$ and $b=\lambda_2$. The construction for the sum of more values is completely analogous.

For both $a$ and $b$ we can find some sequences of normal states $\{\ket{\varphi_{a,k} } \}_k$ and $\{\ket{\varphi_{b,k} } \}_k\subset S$ verifying the previous proposition, with $\|H'_\w(\ket{\varphi_{c,k}}) - \ket{\varphi_{c,k} }\|<1/k$ for $c=a,b$. We can assume, due to translationally invariance of $H'_P$, that the first sequence lies in $\bigcup_{i<-k}S_{i,-k}$, and the second one is contained in $\bigcup_{j>k}S_{k,j}$. These states would then have the form 
\[\ket{\varphi_{a,k} }=\figgg{0.18}{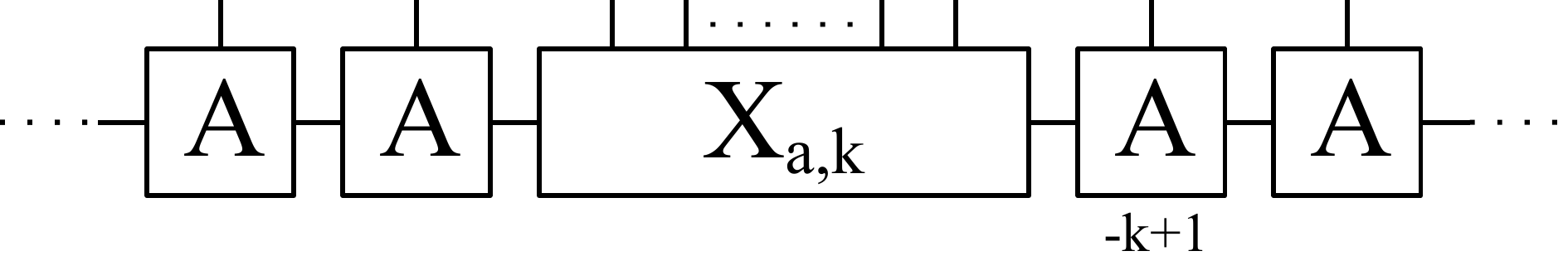}\ 
\]
\[
 \ket{\varphi_{b,k} }=\figgg{0.18}{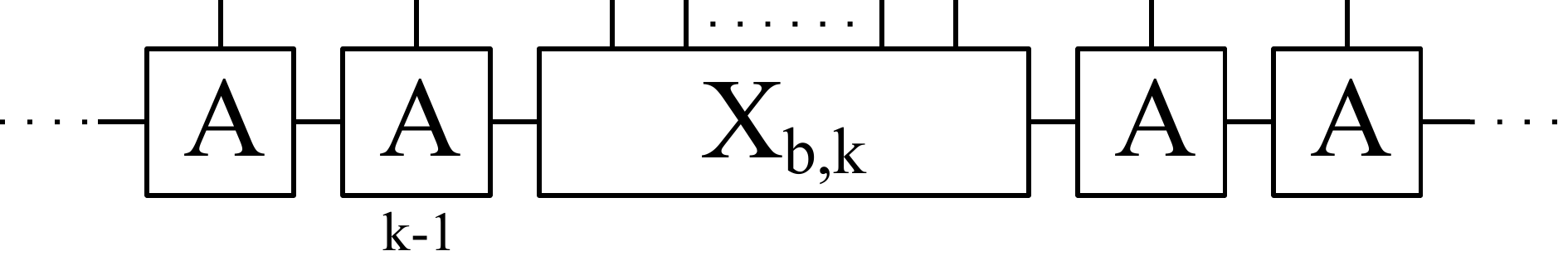}\ ,
\]
for some tensors  $X_{a,k}$, $X_{b,k}$. The normalization conditions would be
\[
\figgg{0.18}{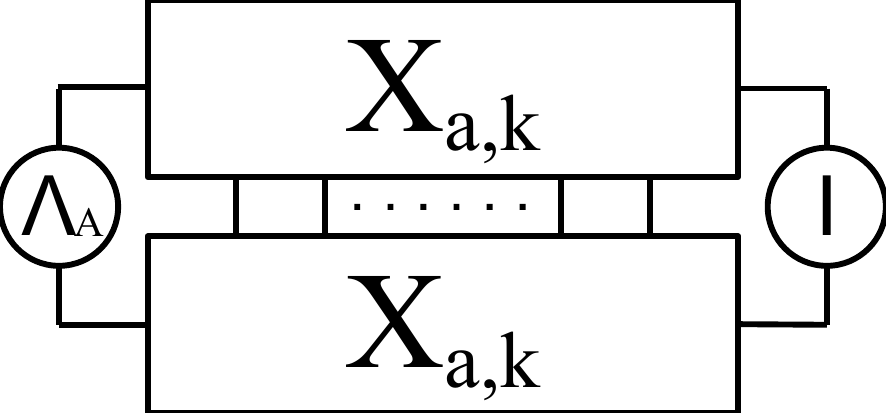}=1=\figgg{0.18}{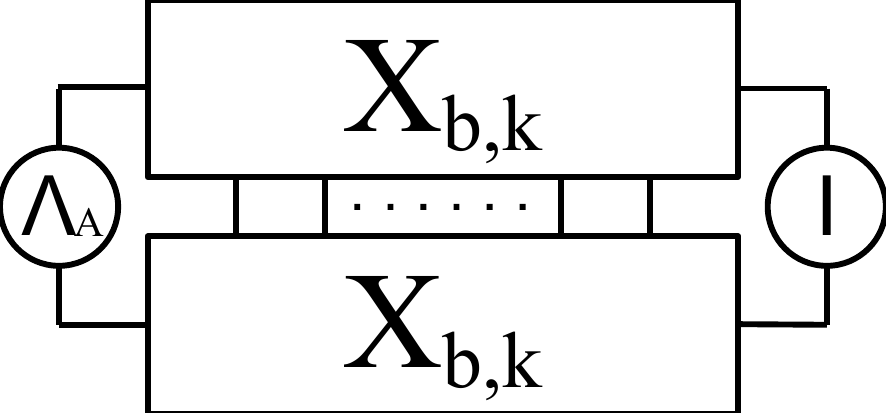}\ .
\]

From these states we can consider the 'concatenated' states
\begin{equation}
\ket{\Phi_k}=\figgg{0.18}{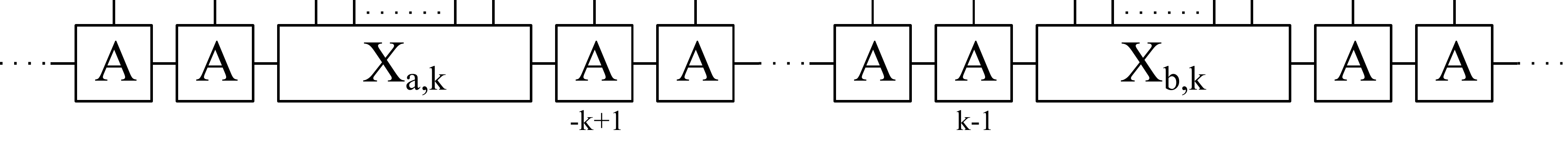}\ ,
\end{equation}
Note that the separation between the $X$ blocks is increasingly growing. 

And due to the structure of $H'_\w$ the image of 
$\bigcup_{i<-k} S_{i,-k}$ is contained in $\bigcup_{i<-k} S_{i,-k+1}$, and $H'_\w(\bigcup_{j>k}S_{k,j}) \subseteq \bigcup_{j>k}S_{k-1,j}$. Moreover, there exist tensors $X'_{a,k}$ and $X'_{b,k}$ such that
\[
H'_\w\left(\figgg{0.18}{conc}\right)=\figgg{0.18}{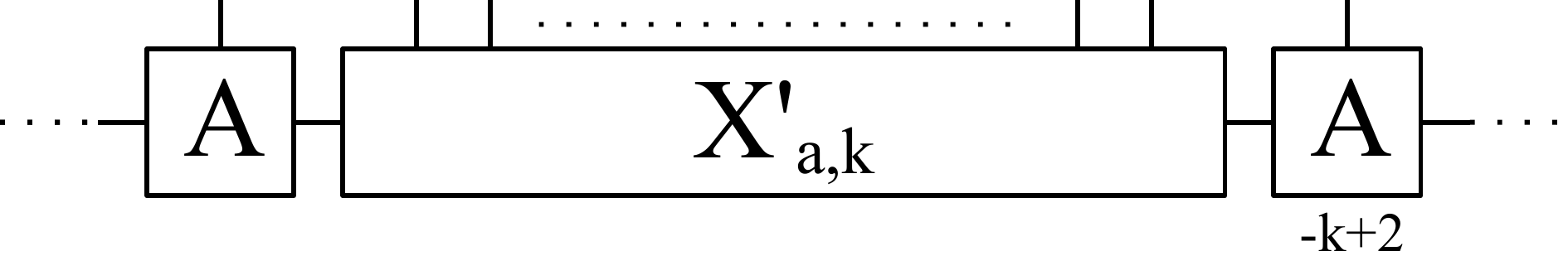}\  
\]
\[
H'_\w\left(\figgg{0.18}{cat}\right)=\figgg{0.18}{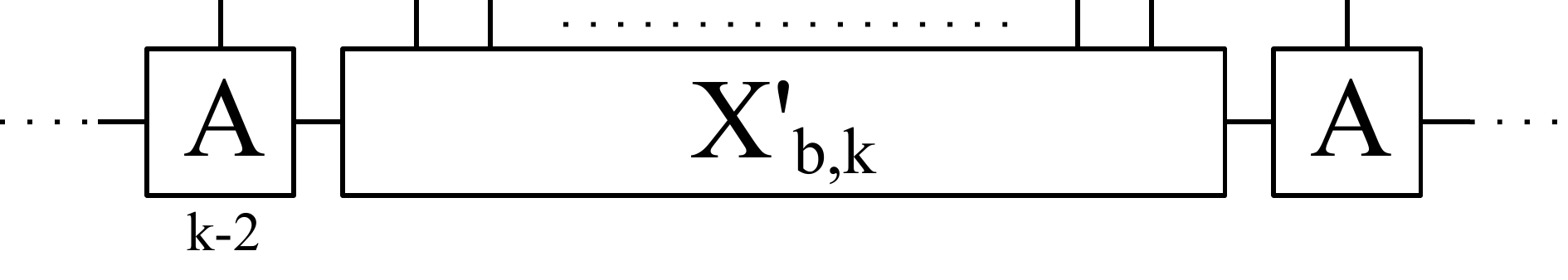}\ .
\]

These new tensors also allow us to describe the image of the concatenations:
\[
H'_\w(\figgg{0.18}{concat2})=
\]
\[
=\figgg{0.18}{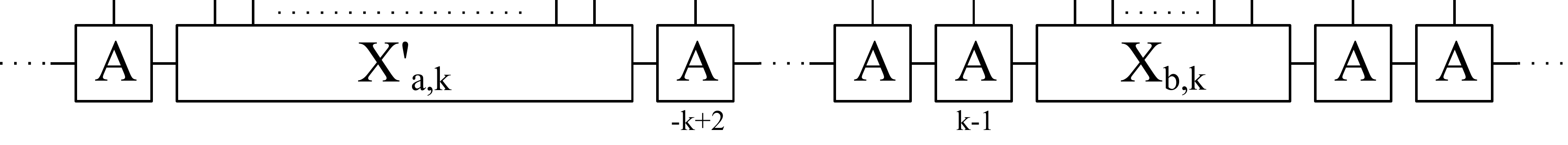}+
\]
\[
=\figgg{0.18}{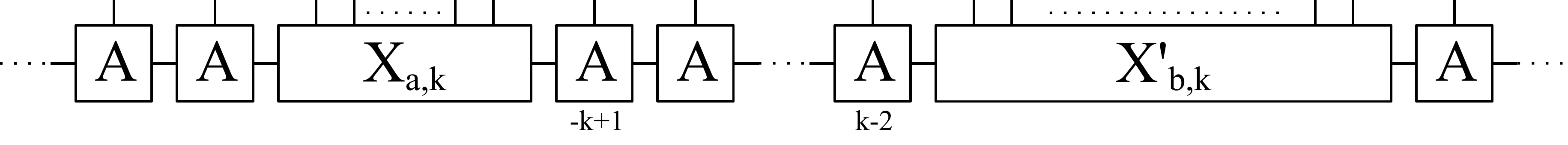}\ .
\]
Let us call $\ket{\Phi'_{k,a}}$ and $\ket{\Phi'_{k,b}}$ these two summands.

We then have that $\|H'_\w(\ket{\Phi_k}) -(a+b)\ket{\Phi_k}\|\le 
\|\ket{\Phi'_{k,a}} -a\ket{\Phi_k}\|+\|\ket{\Phi'_{k,b}} -b\ket{\Phi_k}\|$. 

We can derive a bound for the first summand:
\[
\|\ket{\Phi'_{k,a}} -a\ket{\Phi_k}\|^2=\braket{\Phi'_{k,a}}{\Phi'_{k,a}}+|a|^2\braket{\Phi_k}{\Phi_k}-2\Real(a\braket{\Phi'_{k,a}}{\Phi_{k}})=
\]
\[
\figgg{0.17}{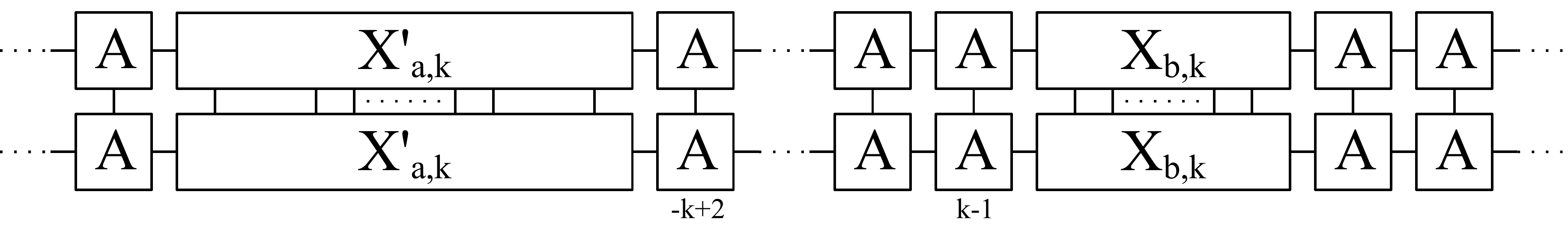} + |a|^2 (1+O(e^{-k}))-
\]
\[- 2 \Real \left(a
\figgg{0.17}{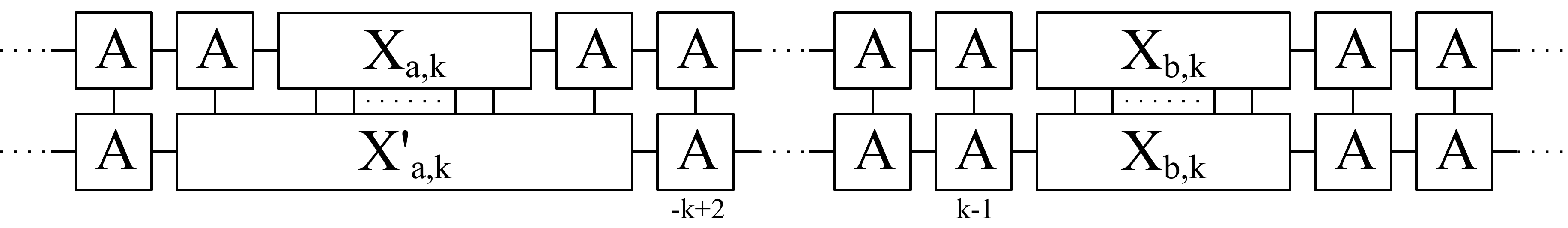} \right )=
\]
\[
= \figgg{0.18}{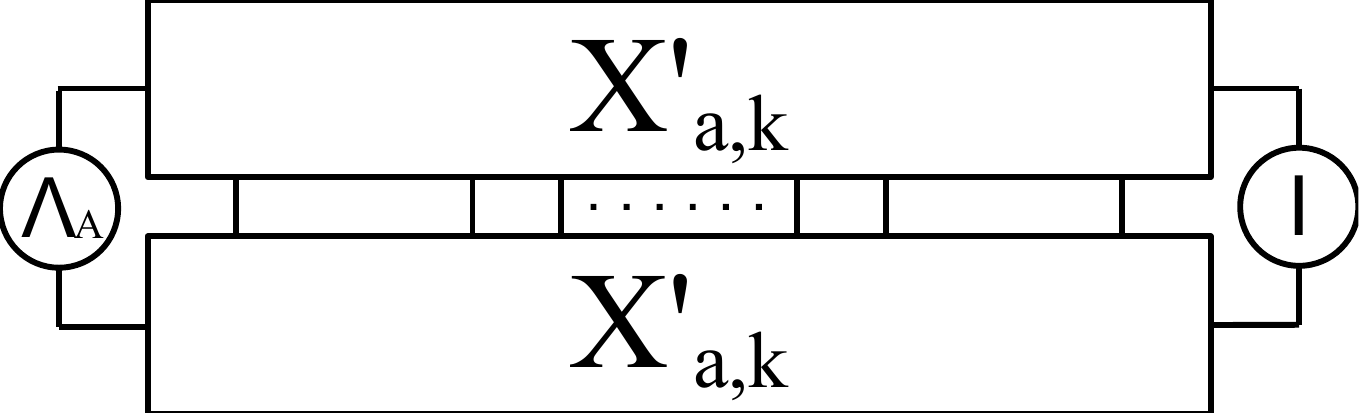}\ \figgg{0.18}{catcat} + |a|^2 -
\]
\[ - 2 \Real \left(a\ 
\figgg{0.18}{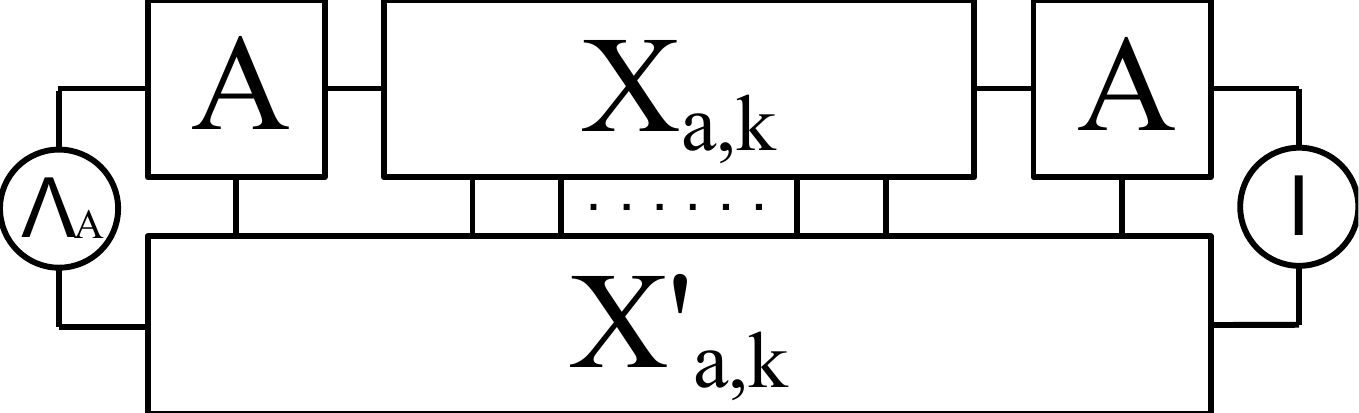}\ \figgg{0.18}{catcat} \right) +O(e^{-k})=
\]
\[ \figgg{0.18}{HconcHconc}\, + |a|^2- 2 \Real \left(a\ 
\figgg{0.18}{Hconcconc} \right)+O(e^{-k})=
\]
\[
=\| H'_\w(\ket{\varphi_{a,k}})-a\ket{\varphi_{a,k}}\|^2+O(e^{-k})<1/k^2 +O(e^{-k}),
\]
where $\Real(\cdot)$ denotes the real part.

A similar bound can be found for the second summand, and we can derive the bound
\[
\|H'_\w(\ket{\Phi_k}) -(a+b)\ket{\Phi_k}\|=O(1/k).
\]

We also have that $\|\ket{\Phi_k}\|\to 1$. Therefore, the sequence $\ket{\Phi_k}/\|\ket{\Phi_k}\|$ satisfy the conditions in Proposition \ref{prop:finitelysupportedWeylseqs} for $a+b=\lambda_1+\lambda_2$, and consequently this sum lies in the spectrum of $H'_\w$.

Longer concatenations would prove the result for any finite sum among values from $\{\lambda_1,\dots,\lambda_n\}$. Note that the bound we get depends on the number of values $\lambda_i$ we are summing,
\[
\|H'_\w(\ket{\Phi_k}) -(\sum_{i=1}^n\lambda_i)\ket{\Phi_k}\|=O(n/k).
\]

\hspace*{\fill}\qed
\end{proof}

\subsection*{Proof of Theorem \ref{thm:finitesizespectra}}

Since we will need to keep track of how close some vector evidencing the existence of elements in the spectrum is from being an eigenvector we introduce the following definition.

\begin{definition}\label{definition:almosteigenvecto}A normalized vector will be called an approximated eigenvector for an operator $A$ and a given value $\lambda$, and for an error $\varepsilon$ if $\|(A -\lambda I ) \ket{\varphi_{\lambda,k}} \|<\varepsilon$. In the case of non-normalized vectors they must satisfy 
\[
\frac{\|(A -\lambda I ) (\ket{\varphi_{\lambda,k}}) \|}{\|\ket{\varphi_{\lambda,k}}\|}<\varepsilon
\]
\end{definition}

\begin{lemma}\label{lemma:finiteappeigenvectors}
For any given values $\lambda\in\R^+$, $n\in\N$ and $\delta>0$, there exists a value $N_0$ such that we can find approximated eigenvectors for the values $j\lambda$, $j=1,\dots,n$, and for an error at most $\delta$ for every finite chain with length greater that $2N_0+1$.
\end{lemma}
\begin{proof}
Let us take $\delta'=\delta/4n$. For $\lambda$ and an error $\delta'$, a normalized approximated eigenvector  $\ket{\varphi_{\lambda,\delta}}=\figgg{0.18}{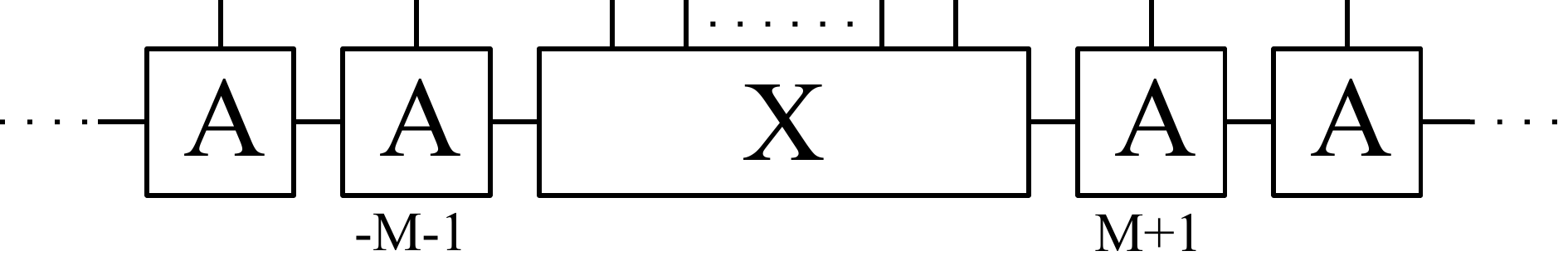}$ can be found in some $S_{-M,M}$ such that $\|(H'_\w -\lambda \I )\ket{\varphi_{\lambda,\delta}} \| < \delta'$. We can now find a value $r$ such that the following vectors are respectively approximated eigenvectors, not necessarily normalized, for the values $j\lambda$, $j=2,\dots,n$, and an error at most
$j\delta'<\delta/2$:
\begin{align*}
\ket{\varphi_{\lambda,\delta}^{(2)}}&=\figgg{0.18}{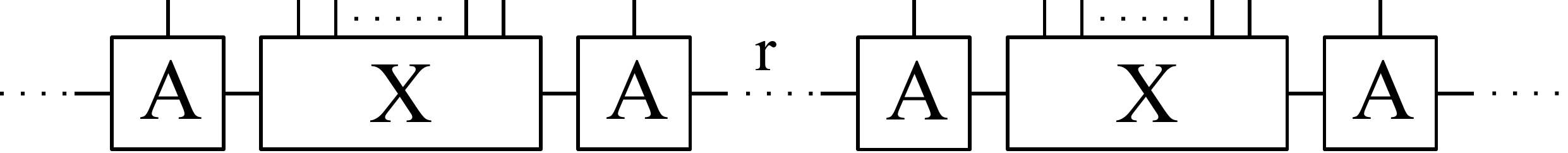},\\
\ket{\varphi_{\lambda,\delta}^{(3)}}&=\figgg{0.18}{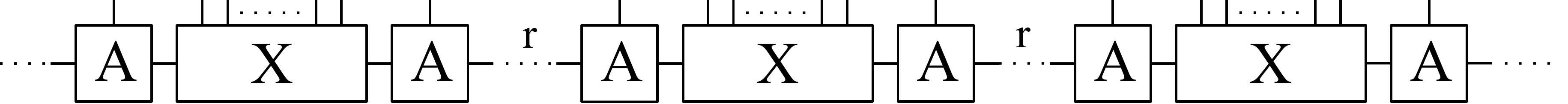}, ...
\\
\ket{\varphi_{\lambda,\delta}^{(N)}}&=\figgg{0.18}{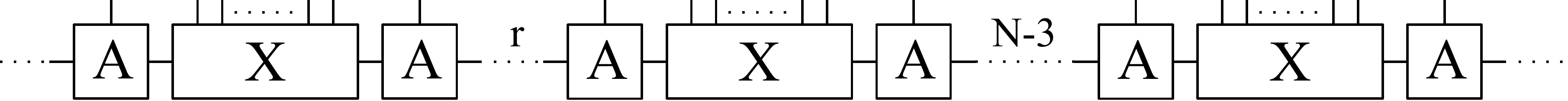},
\end{align*}
where the $r$ denotes how many $A$ tensors are missing in the diagram, and the $N-3$ refers to the number of $X$ blocks, with $r$ $A$ tensors between every two of them, which are also missing.

A value $M'$ can be found such that all these states belong to $S_{-M',M'}$. And, due to Lemma \ref{lemma:almostembeddings}, there exists a value $N_0$ such that the maps
\[
\begin{array}{cccc}
e_{N}:&S_{-M',M'}&\rightarrow& S_{2N+1} \\
 & \figgg{0.17}{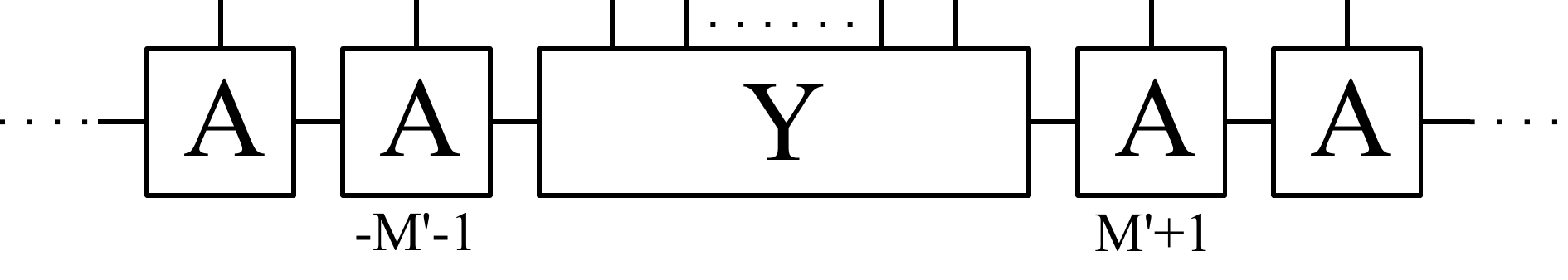} & \mapsto & \figgg{0.17}{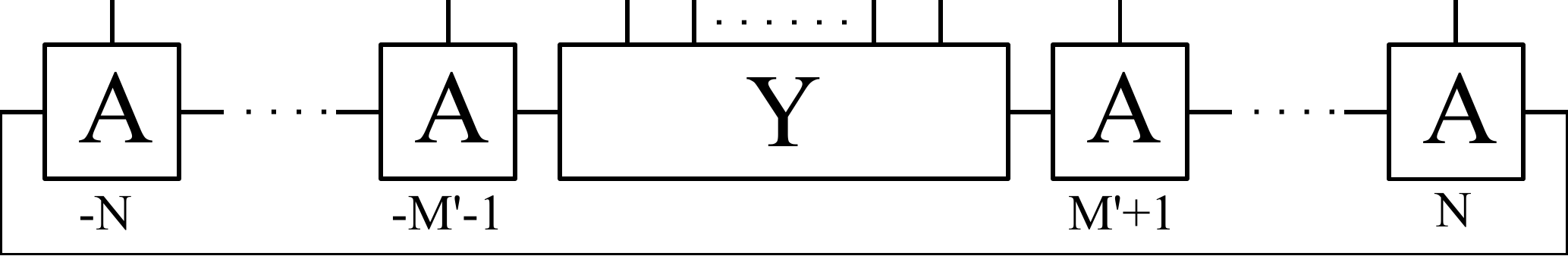}
\end{array}
\]
make each $e(\ket{\varphi_{\lambda,\delta}^{(i)}})$ an approximated eigenvector for $j\lambda$ and an error $\delta$, for the uncle Hamiltonian for the corresponding finite size chain, and for every $N\ge N_0$. The corresponding normalized vectors are approximated eigenvectors for the same values, and therefore satisfy the statement in the lemma.
\hspace*{\fill}\qed
\end{proof}

However, these $j\lambda$ need not be in the spectrum, but indicate the existence of elements in the spectrum close to them, as it is shown in the following lemma.

\begin{lemma} \label{lemma:finiteappeigenvalues} Let $A$ be a self-adjoint operator on a finite dimensional Hilbert space, and $\lambda$ a positive real value such that there exists a unitary vector $\ket{\varphi_{\lambda,\delta}}$ with $\|(A-\lambda\I)\ket{\varphi_{\lambda,\delta}}\| <\delta$. Then $\sigma(A)\cap (\lambda-\delta,\lambda+\delta)\ne \emptyset$.
\end{lemma}

\begin{proof}
Let $\{\lambda_i\}$ be the set of eigenvalues (possibly repeated) of $A$, and $\{\ket{\phi_i}\}$ a corresponding orthonormal basis of eigenvectors. Then 
$\ket{\varphi_{\lambda,\delta}}$ can be written as:
\begin{equation}
\ket{\varphi_{\lambda,\delta}}= \sum_i a_i \ket{\phi_i},
\end{equation}

\noindent and $A(\ket{\varphi_{\lambda,\delta}})=\sum_i a_i \lambda_i \ket{\phi_i}.$

If  $\sigma(A)\cap (\lambda-\delta,\lambda+\delta)$ were the empty set, we would have that 
\begin{equation}
\|(A-\lambda \I)\ket{\varphi_{\lambda,\delta}}\|=\|\sum_i a_i (\lambda_i-\lambda) \ket{\phi_i}\| > \|\sum_i a_i \delta \ket{\phi_i}\| =\delta ,
\end{equation}
which contradicts the conditions of the lemma.
\hspace*{\fill}\qed
\end{proof}

%We can relax the conditions in previous lemma by considering a non unitary vector, and setting then $\delta'=\delta/\|\ket{\varphi_{\lambda,\delta}}\|$ as the radius around $\lambda$ such that $\sigma(A)\cap (\lambda-\delta',\lambda+\delta')$ is necessarily non-empty. 

We are now ready to prove Theorem \ref{thm:finitesizespectra}, which we first recall.

\bigskip

\noindent {\bf Theorem \ref{thm:finitesizespectra}} {\it The spectra of the uncle Hamiltonians for finite size chains tend to be dense in the positive real line.}

\begin{proof}
For any given values $L, m \in \N$, we can set $n=Lm$, $\lambda=1/m$ and $\delta=1/3m$. For these values, approximated eigenvectors can be found following Lemma \ref{lemma:finiteappeigenvectors}, which indicate, due to Lemma \ref{lemma:finiteappeigenvalues}, that $(j\lambda-\delta,j\lambda+\delta)\cap \sigma(H')\ne\emptyset$, $j=1,\dots,n$, for every long enough chain. Therefore eigenvalues for every long enough chain can be found distributed in disjoint intervals centered on the set ${j/m, j=1,\dots,mL}$ as depicted in the following diagram. These sets of eigenvalues, however possibly different for every chain length, tend to be dense in $\R^+$, as we consider higher values for $L$ and $m$.
\hspace*{\fill}\qed
\end{proof}

\begin{figure}[h]
  % Requires \usepackage{graphicx}
\includegraphics[width=10cm]{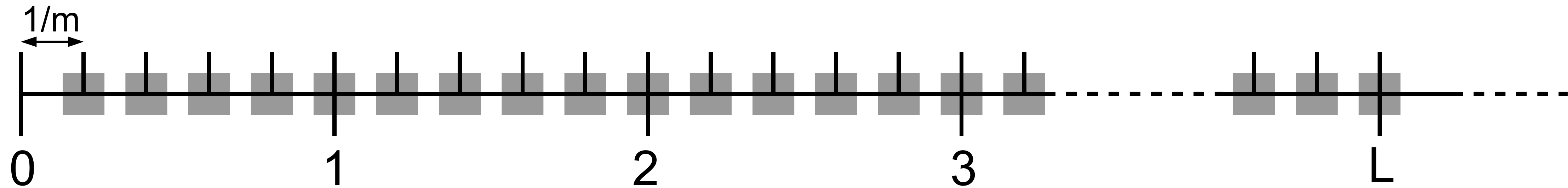}\\
\caption{Eigenvalues tending to be dense in $\R^+$.}\label{denseeigenvalues}
\end{figure}

\section{SPECTRAL REPRESENTATION THEOREM FOR UNBOUNDED OPERATORS} \label{appendix:spectraltheorem}

The uncle Hamiltonian whose spectrum we want to study is an unbounded self-adjoint operator. Therefore its residual spectrum is empty, and the whole spectrum is real. Moreover, since it is a positive operator the spectrum is contained in $\R^+$.

The elements in the union of both the continuous and the point spectra of an operator $A$ (called respectively $\sigma_c(A)$ and $\sigma_p(A)$), which in this case is the whole spectrum, can be characterized in this way: $\lambda \in \sigma_p(A) \cup \sigma_c(A)$ iff there exists a sequence of elements with norm one $\{\ket{\varphi_{\lambda,k}}\}_k\in A$ such that $\|(A -\lambda I ) \ket{\varphi_{\lambda,k}} \|\to 0$. Such sequences are called Weyl sequences. In the case that $\lambda \in \sigma_p ( A)$, a constant sequence exists verifying this condition --always the eigenvector. In the case $\lambda \in \sigma_c(A)$ the elements of the Weyl sequences can be thought of as 'almost eigenvectors', or 'approximated eigenvectors'.

Besides the projectors onto the eigenspaces, a spectral measure with some associated projectors can be defined.

\begin{definition} \cite{conway}
If $X$ is a set, $\Omega$ is a $\sigma$-algebra of subsets of $X$, and $\mathcal{H}$ is a Hilbert space, a spectral measure for $(X,\Omega,\mathcal{H})$ is a function $E:\Omega\to B(\mathcal{H})$ such that:

a) for each $\Delta$ in $\Omega$, $E(\Delta)$ is a projection;

b) $E(\emptyset)=0$ and $E(X)=\I$;

c) $E(\Delta_1\cap\Delta_2)=E(\Delta_1)E(\Delta_2)$ for $\Delta_1,\Delta_2\in\Omega$ (and therefore every pair of such projectors commute, since intersection of sets is commutative);

d) if $\{\Delta_n\}_{n=1}^\infty$ are pairwise disjoint sets from $\Omega$, then $E(\bigcup_{n=1}^\infty \Delta_n)=\sum_{n=1}^\infty E(\Delta_n)$.

\end{definition}

And the spectral representation theorem for unbounded operators states:

\begin{theorem} \cite{conway} Spectral theorem (for unbounded operators). If $N$ is a normal operator on $\mathcal{H}$, then there is a unique spectral measure $E$ defined on the Borel subsets of $\C$ such that:

a) $N=\int zdE(z)$

b) $E(\Delta)=0$ if $\Delta\cap\sigma(N)=\emptyset$

c) if $U$ is an open subset of $\C$ and $U\cap\sigma(N)\ne\emptyset$, then $E(U)\ne 0$

\end{theorem}

The spectral representation theorem for bounded operators is completely analogous. Using the right one we can proof that the uncle Hamiltonians are gapless in the thermodynamic limit -- one can sometimes restrict the uncle Hamiltonian to a subspace in which it is bounded.

With these tools we are able to show that the uncle Hamiltonians are generally gapless in the thermodynamic limit.
\bigskip

\begin{proof} {\it $H'_\w$ is gapless.}

\noindent The representation of $H'_P$ with respect to a ground state $\omega$, say $H'_\omega$, can be uniquely extended to a self-adjoint --and therefore normal-- unbounded operator \cite{conway}, with positive spectrum. 

With the spectral theorem above, it can be proven that if we have a unitary vector $\ket{\varphi}\in \bar{S}$ with $\bra{\varphi}H'_\w\ket{\varphi}=a$ which is orthogonal to the ground space of $H'_\w$, then $(0,a]\cap \sigma(H'_\w)\ne\emptyset$. 

Let us suppose that $\sigma(H'_\omega)\subseteq \{0\}\cup (a,\infty)$. In such a case the norm of $\ket{\varphi}$ would be
\begin{equation}
\braket{\varphi}{\varphi}=\bra{\varphi}\int_{(a,\infty)} dE(z)\ket{\varphi}=1
\end{equation}
\noindent because $\ket{\varphi}$ is orthogonal to the ground space, and we would have that
\begin{gather*}
a=\bra{\varphi}H'_\w\ket{\varphi}=\bra{\varphi}\int_{\R^+} zdE(z)\ket{\varphi}\stackrel{\ket{\varphi}\in  \ker(H'_\w)^\bot}{=}
\\
=\bra{\varphi}\left(\int_{(0,a]} zdE(z)+\int_{(a,\infty)} zdE(z)\right)\ket{\varphi}  \stackrel{\sigma(H'_\omega)\cap (0,a]=\emptyset}{=} \bra{\varphi}  \int_{(a,\infty)} zdE(z)\ket{\varphi}>
\\
> \bra{\varphi}\int_{(a,\infty)} adE(z)\ket{\varphi}=a \braket{\varphi}{\varphi} =a. \quad 
\end{gather*}

Hence if such $\ket{\varphi}$ exists there must be some part of the spectrum lying in $(0,a]$.

This was the case of the low energy states found for the uncle Hamiltonian for the GHZ state, since they were orthogonal to the ground space. However, for a general MPS, the low energy states found are not orthogonal to the ground space. Therefore we cannot state that $E(\{0\})$ in the spectral measure plays no role when looking at these states, but we can say that it will be negligible.

Since the family of states $\ket{\phi_n}$ tends to be orthogonal to the ground space, we have that $\bra{\phi_n}E(\{0\})\ket{\phi_n}$ tends also to 0. We can consider $\ket{\psi_n}=\ket{\phi_n}-E(\{0\})\ket{\phi_n}$, with norm tending to 1 and orthogonal to $\w_A$. 

The only difference is that from the fact that the energies of these states are close to $a$ we cannot infer directly that an element from $(0,a]$ lies in the spectrum, but we can prove that some element from $(0,r]$ does for every $r > a$ and, therefore, also an element from $(0,a]$ does. If this were not the case, we would have that
\begin{gather*}
a\sim \bra{\psi_n}H'_\w\ket{\psi_n}=\bra{\psi_n}\int_{\R^+} zdE(z)\ket{\psi_n}\stackrel{\ket{\psi_n}\in \ker(H'_\w)^\bot}{=}
\\
=\bra{\psi_n}\left(\int_{(0,r]} zdE(z) +\int_{(r,\infty)} zdE(z)\right)\ket{\psi_n}
\stackrel{\sigma(H'_\omega)\cap (0,r]=\emptyset}{=}
\\
= \bra{\psi_n}\int_{(r,\infty)} zdE(z)\ket{\psi_n}\ge
\\
\ge \bra{\psi_n}\int_{(r,\infty)} rdE(z)\ket{\psi_n}=r \braket{\psi_n}{\psi_n}\to r,
\end{gather*}
\noindent which contradicts the hypothesis $r>a$. 
\end{proof}
\bigskip
\bigskip

\end{document}